\documentclass[10pt,a4paper]{article}

\usepackage[numbers, compress]{natbib}
\usepackage{authblk}
\usepackage[utf8]{inputenc} 
\usepackage{hyperref}       
\usepackage{url}            
\usepackage{booktabs}       
\usepackage{nicefrac}       
\usepackage{microtype}      
\usepackage{xcolor}         

\usepackage{amsmath}
\usepackage{amsthm}
\usepackage{amssymb}
\usepackage{bbold}
\usepackage{graphicx}
\usepackage{caption}
\usepackage{subcaption}
\usepackage{scalerel,stackengine}
\usepackage[left=2.5cm,right=2.5cm,top=2cm,bottom=2cm]{geometry}

\makeatletter
\providecommand{\@LN}[2]{}
\makeatother

\newtheorem{theorem}{Theorem}
\newtheorem{remark}{Remark}

\newtheorem{lemma}{Lemma}

\newtheorem{proposition}{Proposition}

\newtheorem{assumptionex}{Assumption}
 \newenvironment{assumption}
   {\assumptionex}
   {\endassumptionex}
\newtheorem{assumptionexp}{Assumption}
\newenvironment{assumptionp}
  {\assumptionexp}
  {\endassumptionexp}

\title{Extending Kernel Testing To General Designs}

\author[,1]{A. Ozier-Lafontaine}
\author[1,2]{P. Arsenteva}
\author[,2]{F. Picard\thanks{joint last authors}}
\author[$*$,1]{B. Michel\thanks{To whom correspondence should be addressed: \texttt{Bertrand.Michel@ec-nantes.fr}, \texttt{franck.picard@ens-lyon.fr}}}

\affil[1]{Nantes Universit\'e, Centrale Nantes,Laboratoire de Math\'ematiques Jean Leray, CNRS UMR 6629, F-44000, Nantes, France}
\affil[2]{Laboratory of Biology and Modelling of the Cell, Université de Lyon, Ecole Normale Supérieure de Lyon, CNRS, UMR5239, Université Claude Bernard Lyon 1, Lyon, France}

\begin{document}

\stackMath
\newcommand\reallywidehat[1]{%
\savestack{\tmpbox}{\stretchto{%
  \scaleto{%
    \scalerel*[\widthof{\ensuremath{#1}}]{\kern-.6pt\bigwedge\kern-.6pt}%
    {\rule[-\textheight/2]{1ex}{\textheight}}
  }{\textheight}%
}{0.5ex}}%
\stackon[1pt]{#1}{\tmpbox}%
}

\newcommand{\R}{\mathbb{R}} 
\newcommand{\Nat}{\mathbb{N}} 
\newcommand{\CC}{\mathbb{C}} 
\newcommand{\Max}[1]{\underset{#1}{\operatorname{max}}}
\newcommand{\argmax}[1]{\underset{#1}{\operatorname{argmax}}}
\newcommand{\argmin}[1]{\underset{#1}{\operatorname{argmin}}}
\newcommand{\mtrx}[1]{\mathcal{M}_{#1}(\R)}
\newcommand{\var}{\operatorname{Var}}
\newcommand{\diag}{\operatorname{Diag}}
\newcommand{\Span}{\operatorname{Span}}
\newcommand{\trhs}{\operatorname{tr}}
\newcommand{\trmtrx}{\operatorname{trace}}

\newcommand{\lch}{\odot} 
\newcommand{\identity}{\mathbf{I}}
\newcommand{\fullofones}{\mathbf{J}}

\newcommand{\proba}{\mathbb{P}} 
\newcommand{\hproba}{\mathbb{P}_\nobs} 
\newcommand{\E}{\mathbb{E}}
\newcommand{\EP}{\mathbb{E}_{\proba}}
\newcommand{\EPi}[1]{\mathbb{E}_{\proba_{#1}}}
\newcommand{\cov}{\operatorname{Cov}}

\newcommand{\sY}{\mathcal{Y}}
\newcommand{\sYtribe}{\mathfrak{Y}}

\newcommand{\fy}{y} 
\newcommand{\dy}{Y} 
\newcommand{\dY}{\mathbf{Y}}
\newcommand{\hdy}{\widehat{Y}}
\newcommand{\hdY}{\widehat{\mathbf{Y}}}

\newcommand{\nobs}{n}
\newcommand{\ydim}{m}

\newcommand{\gram}{\mathbf{K}_{\dY}}
\newcommand{\grami}[1]{\mathbf{K}_{\dY_{#1}}}
\newcommand{\kernel}{k}
\newcommand{\kerneldot}{k(\cdot,\cdot)}

\newcommand{\fmap}[1]{\phi(#1)} 
\newcommand{\fmapdot}{\phi(\cdot)} 
\newcommand{\fmapk}[1]{\kernel(#1,\cdot)}

\newcommand{\Hk}{\mathcal{H}}
\newcommand{\HSk}{\operatorname{HS}(\mathcal{H})}
\newcommand{\Hb}{g}

\newcommand{\eY}{\Phi(\dY)} 
\newcommand{\eYi}[1]{\Phi(\dY_{#1})} 
\newcommand{\ey}[1]{\phi(\dy_{#1})} 
\newcommand{\efy}{\phi(\fy)} 
\newcommand{\eYT}{\Phi_{\tmax}(\dY)} 
\newcommand{\eyT}[1]{\phi_{\tmax}(\dy_{#1})} 
\newcommand{\heY}{\widehat{\Phi(\dY)}}
\newcommand{\hey}[1]{\widehat{\phi(\dy_{#1})}}

\newcommand{\kme}{\mu} 
\newcommand{\hkme}{\widehat{\mu}} 
\newcommand{\kcov}{\Sigma_{\mathbb{P}}} 
\newcommand{\kcovi}[1]{{\Sigma}_{#1}} 
\newcommand{\hkcov}{\widehat{\Sigma}_{\mathbb{P}}} 
\newcommand{\hkcovi}[1]{\widehat{\Sigma}_{#1}} 
\newcommand{\hkcovkt}{\mathbf{K}_{\Sigma}} 
\newcommand{\hPcov}{\widehat{{\Pi}}_{\proba,\tmax}}

\newcommand{\eigveckcov}{f^{\proba}}
\newcommand{\eigvalkcov}{\lambda^{\proba}}
\newcommand{\heigveckcov}{\widehat{f}^{\proba}}
\newcommand{\heigvalkcov}{\widehat{\lambda}^{\proba}}
\newcommand{\heigveckcovkt}{u}

\newcommand{\sX}{\mathcal{X}}
\newcommand{\dX}{\mathbf{X}}
\newcommand{\dx}{x}
\newcommand{\dXcol}{X}
\newcommand{\XXX}{\XXinv \dX^\prime}
\newcommand{\XXinv}{(\dX^\prime \dX)^{-}}
\newcommand{\XXXprime}{\dX\XXinv}
\newcommand{\LXXXprime}{\dX\XXinv\Lmanova'}
\newcommand{\lxxxi}{W}
\newcommand{\lxxxj}{W_{\ast,j}}
\newcommand{\XXlim}{W}
\newcommand{\nx}{p}
\newcommand{\rangx}{\nx}

\newcommand{\Pxn}{\boldsymbol{\Pi}_\nobs}
\newcommand{\Pxm}[1]{\boldsymbol{\Pi}_{#1}}
\newcommand{\Pxw}{\boldsymbol{\Pi}_W}
\newcommand{\Px}{\mathbf{P}_{\dX}}
\newcommand{\Pxvalue}{\dX \XXX }
\newcommand{\Pxi}[1]{{\pi}_{#1}}
\newcommand{\Pxij}[2]{{\pi_{#1,#2}}}
\newcommand{\Pxp}{{\mathbf{P}_{\dX}^\perp}}
\newcommand{\Pxpi}[1]{\pi^\perp_{#1}}
\newcommand{\Pxpij}[2]{{\pi_{#1,#2}^\perp}}
\newcommand{\Pxev}{v}

\newcommand{\PxpIij}[2]{{\pi_{#1,#2}^{I,\perp}}}
\newcommand{\PxpI}{{\mathbf{P}_{\dX}^{I,\perp}}}

\newcommand{\Hkn}{\Hk^\nobs}
\newcommand{\flandmarks}{\varphi} 
\newcommand{\lX}{{\dX^I}}
\newcommand{\lx}[1]{\dx_{\flandmarks(#1)}}
\newcommand{\lproj}{Q_\lY}
\newcommand{\helY}{\widehat{\elY}}

\newcommand{\lrcov}{\widehat{\Sigma}^{I}}
\newcommand{\lrcovec}{\widehat{f}^{\lY}}
\newcommand{\lrcoval}{\widehat{\lambda}^{\lY}}
\newcommand{\lrcovals}{\widehat{\Lambda^{\lY}}}

\newcommand{\lrcovkt}{\mathbf{K}^\lY_{\error}}
\newcommand{\lrcoveckt}{u^{\lY}}
\newcommand{\lrcovecskt}{\mathbf{U}^{\lY}}

\newcommand{\hresleigval}[1]{\widehat{\lambda^{\lY}_{\errors,{#1}}}}
\newcommand{\hresleigvals}{\widehat{\Lambda^\lY_{\errors,1:\nanchors}}}
\newcommand{\hresleigvecskt}{\mathbf{U}^\lY_{\errors,1:\nanchors}}
\newcommand{\hresleigveckt}[1]{u^\lY_{\errors,#1}}

\newcommand{\arcov}{\widehat{\Sigma}^{\Anc}}
\newcommand{\arcovt}{\widehat{\Sigma}_{\tmax}^{\Anc}}
\newcommand{\arcovbis}{\widetilde{\Sigma}^{\Anc}}
\newcommand{\arcovec}[1]{\widehat{f^{\Anc}_{#1}}}
\newcommand{\arcoval}[1]{\widehat{\lambda^{\Anc}_{#1}}}
\newcommand{\prarcov}[1]{\Pi_{\arcovec{#1}}}

\newcommand{\arcovecs}{{\widehat{f^{\Anc}}}}
\newcommand{\arcovals}{{\widehat{\Lambda^{\Anc}}}}

\newcommand{\arcovkt}{\mathbf{K}^{\Anc}_{\error}}
\newcommand{\arcovecskt}{\mathbf{U}^{\Anc}}
\newcommand{\arcoveckt}{u^{\Anc}}

\newcommand{\params}{\mathbf{\Theta}}
\newcommand{\paramsuni}{\Theta}
\newcommand{\param}{\theta}

\newcommand{\errors}{\mathbf{E}}
\newcommand{\error}{e}

\newcommand{\hparams}{\mathbf{\widehat{\Theta}}}
\newcommand{\hparam}{\widehat{\theta}}
\newcommand{\paramsproj}{\mathbf{K}_{\hparams}}
\newcommand{\aparamsproj}{{\mathbf{K}_{\hparams}^{\Anc}}}

\newcommand{\residuals}{\mathbf{\widehat{E}}}
\newcommand{\residual}{\widehat{e}}

\newcommand{\lresidual}{\widehat{e}^{\lY}}
\newcommand{\lresiduals}{\widehat{E^{\lY}}}
\newcommand{\aresidual}{\widehat{e}^{\Anc}}
\newcommand{\aresiduals}{\widehat{E^{\Anc}}}

\newcommand{\ecov}{\Sigma}
\newcommand{\ecovbar}{{\Bar\Sigma}_n}
\newcommand{\ecovbarI}{{\Bar\Sigma}_I}
\newcommand{\ecovt}{\Sigma_\tmax}
\newcommand{\ecovanc}{\Sigma_\nanchors}
\newcommand{\ecovecs}{f}
\newcommand{\ecovec}{f}
\newcommand{\ecovals}{\boldsymbol{\Lambda}}
\newcommand{\ecoval}{\lambda}

\newcommand{\rcov}{\widehat{\ecov}}
\newcommand{\rcovridge}{\widehat{\ecov}_{\gamma}}
\newcommand{\rcovt}{\widehat{\ecov}_{\tmax}}

\newcommand{\rcovecs}{\widehat {f}}
\newcommand{\rcovals}{\widehat{\boldsymbol{\Lambda}}}
\newcommand{\rcovec}{\widehat{f}}
\newcommand{\rcoval}{\widehat{\lambda}}

\newcommand{\rcovkt}{\mathbf{K}_{\error}}
\newcommand{\rcovecskt}{\mathbf{U}}
\newcommand{\rcovecsktnorm}{\widetilde{\mathbf{U}}_\tmax}
\newcommand{\rcoveckt}{u}
\newcommand{\rcovecktnorm}{\widetilde{u}}
\newcommand{\rcovalkt}{\nu}

\newcommand{\THLcov}{\rcovt^{-1} \Hmanova}
\newcommand{\THLcovkt}{\prey \mathbf{D} \prey^\prime}
\newcommand{\prTHL}{\mathbf{D}}
\newcommand{\nTHL}{\kappa}

\newcommand{\THLcovec}{\widehat{g}}
\newcommand{\THLcoveckt}{v}
\newcommand{\THLcovecskt}{\mathbf{V}}
\newcommand{\THLcoval}{\xi}
\newcommand{\THLcovals}{{\boldsymbol{\Lambda}}_{\xi}}
\newcommand{\THLcovalsnorm}{\widetilde{\boldsymbol{\Lambda}}_{\xi}}

\newcommand{\prrcov}[1]{\Pi_{\rcovec_{#1}}}
\newcommand{\pranc}[1]{\Pi_{\anc_{#1}}}
\newcommand{\precov}[1]{\Pi_{\ecovec_{#1}}}

\newcommand{\ecovtci}[1]{\Sigma_{{#1},\tmax^c}}
\newcommand{\ecovancci}[1]{\Sigma_{{#1},\nanchors^c}}

\newcommand{\linecov}{\mathbf{S}_\tmax}
\newcommand{\linrcov}{\widehat{\mathbf{S}}_\tmax}

\newcommand{\rcovny}{\mathbf{A}}
\newcommand{\rcovanc}{\widehat{\Sigma}^{\Anc}}
\newcommand{\rcovanckt}{\mathbf{}}

\newcommand{\linparams}{{\boldsymbol{\beta}}}
\newcommand{\linparam}{{\beta}}
\newcommand{\hlinparams}{\widehat{\boldsymbol{\beta}}}
\newcommand{\hlinparam}{\widehat{\beta}}

\newcommand{\linerrors}{\boldsymbol{\varepsilon}}
\newcommand{\linerror}{{\varepsilon}}
\newcommand{\linresiduals}{\widehat{\boldsymbol{\varepsilon}}}
\newcommand{\linresidual}{\widehat{\varepsilon}}

\newcommand{\hlSepsilonT}{\widehat{\Sigma}_{\error,\tmax}^{\lY}}

\newcommand{\Pepsilon}{{\Pi}_{\error}}
\newcommand{\PepsilonT}{{\Pi}_{\error,\tmax}}
\newcommand{\hPepsilon}{\widehat{{\Pi}}_{\errors}}
\newcommand{\hPepsilonT}{\widehat{{\Pi}}_{\errors,\tmax}}
\newcommand{\Peigvec}[1]{\Pi_{\eigvec_{#1}}}
\newcommand{\hPeigvec}[1]{\Pi_{\heigvec_{#1}}}

\newcommand{\wkcov}{\Sigma_{W}}
\newcommand{\wkcovridge}{\Sigma_{W,\gamma}}
\newcommand{\hwkcovridge}{\widehat{\Sigma}_{W,\gamma}}
\newcommand{\wkcovt}{\Sigma_{W,\tmax}}
\newcommand{\hwkcov}{\widehat{\Sigma}_{W}}
\newcommand{\hwkcovkt}{\mathbf{K}_{W}}
\newcommand{\hwkcovt}{\widehat{\Sigma}_{W,\tmax}}

\newcommand{\weigvec}{f}
\newcommand{\weigval}{\lambda}
\newcommand{\hweigvec}{\widehat{f}}
\newcommand{\hweigveckt}[1]{u_{W,#1}}
\newcommand{\hweigval}{\widehat{\lambda}}
\newcommand{\hweigvecskt}{\mathbf{U}_{W}}
\newcommand{\hweigvecstkt}{\mathbf{U}_{W,\tmax}}
\newcommand{\hweigvecs}{\widehat{f}_W}
\newcommand{\hweigvecst}{\widehat{f}_{W,\tmax}}
\newcommand{\hweigvals}{\boldsymbol{\Lambda}_{W}}
\newcommand{\hweigvalst}{\boldsymbol{\Lambda}_{W,\tmax}}

\newcommand{\bkcov}{\Sigma_{B}}
\newcommand{\hbkcov}{\widehat{\Sigma}_{B}}

\newcommand{\Sw}{\Sigma_{W}}
\newcommand{\Sb}{\Sigma_{B}}
\newcommand{\hSw}{\widehat{\Sigma}_{W}}
\newcommand{\hSb}{\widehat{\Sigma}_{B}}

\newcommand{\mmd}{\operatorname{MMD}}
\newcommand{\hmmd}{\widehat{\operatorname{MMD}}}
\newcommand{\hmmdt}{\widehat{\operatorname{MMD}}_{\tmax}}

\newcommand{\kfda}{\operatorname{D}^2}
\newcommand{\hkfdaridge}{\widehat{\operatorname{D}_{\gamma}}^2}
\newcommand{\hkfdat}{\widehat{\operatorname{D}_{\tmax}}^2}


\newcommand{\statmanova}{\widehat{\mathcal{F}}_{\tmax}}
\newcommand{\statmanovany}{\widehat{\mathcal{F}}_{\tmax}^{\operatorname{NY}}}
\newcommand{\tildemanova}{\widetilde{\mathcal{F}}_{\tmax}}

\newcommand{\statmanovakt}{\tilde{\mathcal{G}}}

\newcommand{\Hmanova}{\widehat{H}_{\Lmanova}}
\newcommand{\Lmanova}{\mathbf{L}}
\newcommand{\lmanova}{L}
\newcommand{\linHmanova}{\widehat{\mathbf{H}_{\linLmanova,\tmax}}}
\newcommand{\linLmanova}{\mathbf{L}}
\newcommand{\nLmanova}{d}

\newcommand{\dcook}{\mathcal{D}_{Cook}}

\newcommand{\tmax}{T}
\newcommand{\boundX}{M_{x}}
\newcommand{\boundk}{M_{k}}
\newcommand{\boundepsilon}{M_{\error}}
\newcommand{\boundrcov}{\zeta(\nobs,\xi)}
\newcommand{\boundlrcov}{\zeta(\nlandmarks,\xi)}
\newcommand{\boundarcov}{\zeta(\nobs,\nlandmarks,\xi)}


\newcommand{\A}[1]{$\mathtt{A_#1}$}
\newcommand{\Ahomoscedasticity}{\A{1}}
\newcommand{\Aboundedkernel}{\A{2}}
\newcommand{\Aboundeddesign}{\A{3}}
\newcommand{\Aboundedbasisofdesign}{\A{4}}
\newcommand{\Aconvergentdesign}{\A{5}}

\newcommand{\B}[1]{$\mathtt{B_#1}$}

\newcommand{\lY}{\mathbf{Z}}
\newcommand{\ly}{Z}

\newcommand{\nlandmarks}{q}
\newcommand{\Hkl}{\mathcal{H}_{\lY}}

\newcommand{\elY}{\Phi(\lY)} 
\newcommand{\elYi}[1]{\Phi(\lY_{#1})} 
\newcommand{\ely}[1]{\phi(\ly_{#1})} 
\newcommand{\hely}[1]{\widehat{\phi(\ly_{#1})}}

\newcommand{\Anc}{\operatorname{a}}
\newcommand{\anc}{\hat a}
 \newcommand{\ancMat}{\boldsymbol{\hat a}}
\newcommand{\Hka}{\mathcal{H}_{a}}
\newcommand{\nanchors}{m}
\newcommand{\Panc}{{\Pi}_{\Anc}}
\newcommand{\Ancm}{\operatorname{A}_\nanchors}
\newcommand{\Ancmi}[1]{\operatorname{A}_{\nanchors,#1}}

\newcommand{\hlkme}{\widehat{\mu}^{\lY}}
\newcommand{\hlkmei}[1]{\widehat{\mu}^{\lY_{#1}}}
\newcommand{\hlkcov}{\widehat{\Sigma}^{\lY}}
\newcommand{\hlkcovi}[1]{\widehat{\Sigma}^{\lY_{#1}}}
\newcommand{\hlkcovkt}{\mathbf{K}^{\lY}}
\newcommand{\hleigval}{\widehat{\lambda}^{\lY}}
\newcommand{\hleigvals}{\widehat{\boldsymbol{\Lambda}_{\nanchors}^{\lY}}}
\newcommand{\hleigveckt}{u^{\lY}}
\newcommand{\hleigvecskt}{\mathbf{U}_{\nanchors}^{\lY}}
\newcommand{\Pxlw}{\boldsymbol{\Pi}_W^{\nlandmarks}}
\newcommand{\hlwkcov}{\widehat{\Sigma}_W^{\lY}}
\newcommand{\hlwkcovkt}{\mathbf{K}_W^{\lY}}
\newcommand{\hlweigval}[1]{\widehat{\lambda}_{W,#1}^{\lY}}
\newcommand{\hlweigvals}{\widehat{\boldsymbol{\Lambda}}_{W,\nanchors}^{\lY}}
\newcommand{\hlweigveckt}[1]{u_{W,#1}^{\lY}}
\newcommand{\hlweigvecskt}{\mathbf{U}_{W,\nanchors}^{\lY}}

\newcommand{\fmapanc}[1]{\phi_{\Anc}(#1)}
\newcommand{\eYanc}{\Phi_{\Anc}(\dY)}
\newcommand{\eyanc}[1]{\phi_{\Anc}(\dy_{#1})}
\newcommand{\eYianc}[1]{\Phi_{\Anc}(\dY_{#1})}

\newcommand{\gramly}{\mathbf{K}_{\lY}}
\newcommand{\gramlyi}[1]{\mathbf{K}_{\lY_{#1}}}
\newcommand{\gramyz}{\mathbf{K}_{\dY,\lY}}
\newcommand{\gramyiz}[1]{\mathbf{K}_{\dY_{#1},\lY}}
\newcommand{\gramyzi}[1]{\mathbf{K}_{\dY,\lY_{#1}}}
\newcommand{\gramziy}[1]{\mathbf{K}_{\lY_{#1},\dY}}
\newcommand{\gramyyi}[1]{\mathbf{K}_{\dY,\dY_{#1}}}
\newcommand{\gramyiy}[1]{\mathbf{K}_{\dY_{#1},\dY}}
\newcommand{\gramyizi}[1]{\mathbf{K}_{\dY_{#1},\lY_{#1}}}
\newcommand{\gramzy}{\mathbf{K}_{\lY,\dY}}
\newcommand{\gramzyi}[1]{\mathbf{K}_{\lY,\dY_{#1}}}
\newcommand{\gramziyi}[1]{\mathbf{K}_{\lY{#1},\dY_{#1}}}

\newcommand{\hwkcovanc}{\widehat{\Sigma}_{W}^{\Anc}}
\newcommand{\hwkcovancktint}{\widetilde{\mathbf{K}}_{W}^{\Anc}}
\newcommand{\hwkcovanckt}{\mathbf{K}_{W}^{\Anc}}
\newcommand{\hweigvalanc}{\widehat{\lambda}^{\Anc}}
\newcommand{\hweigvecanc}{\widehat{f}^{\Anc}}
\newcommand{\hweigvecanckt}{u^{\Anc}}
\newcommand{\hweigvecsanckt}{\mathbf{U}^{\Anc}}
\newcommand{\hweigvecsanc}{\operatorname{f}^{\Anc}}

\newcommand{\mmdanc}{\widehat{\operatorname{MMD}}_{\Anc}^2}
\newcommand{\kfdatanc}{\widehat{\operatorname{D}^{\Anc}_{\tmax}}^2}

\newcommand{\probas}{\mathcal{P}(\sY)}
\newcommand{\probasc}{\widetilde{\mathcal{P}}(\sY)}
\newcommand{\probasci}[1]{\widetilde{\mathcal{P}}_{#1}(\sY)}

\newcommand{\sffdot}{T(\circledcirc)}
\newcommand{\sffjdot}{T(\circledcirc,\circledcirc)}
\newcommand{\sff}{T}
\newcommand{\sfmdot}{m(\circledcirc)}
\newcommand{\sfm}{m}
\newcommand{\sfcovdot}{S(\circledcirc)}
\newcommand{\sfcov}{S}
\newcommand{\sfwkcov}{\Sigma_{W}}
\newcommand{\sfbkcov}{\Sigma_{B}}
\newcommand{\sfjdot}[1]{{#1}(\circledcirc,\circledcirc)}

\newcommand{\eps}{\epsilon}
\newcommand{\limeps}{\underset{\substack{\eps\to0\\\eps>0}}{\lim}}

\newcommand{\gd}[1]{{#1}^\prime\big(\proba,\probaq\big)}
\newcommand{\gdi}[2]{{#1}^\prime\big(\proba_{#2},\probaq_{#2}\big)}
\newcommand{\gdq}[2]{{#1}^{\prime}\big(\proba,#2\big)}
\newcommand{\gdpq}[3]{{#1}^{\prime}\big(#2,#3\big)}
\newcommand{\egd}[1]{{#1}^{\prime}\big(\probay,\probaly\big)}
\newcommand{\egdi}[2]{{#1}^{\prime}\big(\probayi{#2},\probalyi{#2}\big)}
\newcommand{\sgd}[1]{{#1}^{S}\big(\dY,\lY\big)}

\newcommand{\pgdyl}[1]{{#1}^\prime\big((\joint),(\diracfy,\proba_2)\big)}
\newcommand{\pgdyr}[1]{{#1}^\prime\big((\joint),(\proba_1,\diracfy)\big)}
\newcommand{\pgd}[1]{{#1}^\prime\big((\joint),(\jointq)\big)}
\newcommand{\epgd}[1]{{#1}^\prime\big((\jointy),(\jointly)\big)}

\newcommand{\ifff}[1]{\operatorname{I}_{#1}(\proba,\fy)}
\newcommand{\ifpy}[3]{\operatorname{I}_{#1}(#2,#3)}
\newcommand{\eif}[1]{\operatorname{I}_{#1}^{E}(\proba,\dY,\fly)}
\newcommand{\sif}[1]{\operatorname{I}_{#1}^{S}(\dY,\fly)}

\newcommand{\pifyl}[1]{\operatorname{I}_{#1,\proba_{2}}(\proba_{1},\fy)}
\newcommand{\pifyr}[1]{\operatorname{I}_{#1,\proba_1}(\proba_{2},\fy)}

\newcommand{\fly}{z} 
\newcommand{\yly}{\dY\backslash\lY} 
\newcommand{\yfly}{\dY\backslash\{\fly\}} 

\newcommand{\probaq}{\mathbb{Q}}
\newcommand{\contam}{\mathbb{P}^{(\probaq,\eps)}}
\newcommand{\contami}[1]{\mathbb{P}_{#1}^{(\probaq_{#1},\eps)}}
\newcommand{\diracfy}{\delta_\fy}
\newcommand{\contamy}{\mathbb{P}^{(\diracfy,\eps)}}
\newcommand{\probay}{\mathbb{P}_{\nobs}}
\newcommand{\probaly}{\mathbb{P}_{\nobs,\nlandmarks}}
\newcommand{\deltaly}{\delta_{\fly}}
\newcommand{\probayly}{\mathbb{P}_{\yly}}
\newcommand{\probayfly}{\mathbb{P}_{\yfly}}
\newcommand{\probayi}[1]{\proba_{\nobs_{#1}}}
\newcommand{\probalyi}[1]{\proba_{\nobs_{#1},\nlandmarks_{#1}}}

\newcommand{\joint}{\proba_1,\proba_2}
\newcommand{\jointq}{\probaq_1,\probaq_2}
\newcommand{\jointy}{\probayi{1},\probayi{2}}
\newcommand{\jointly}{\probalyi{1},\probalyi{2}}
\newcommand{\contamjoint}{(\joint)^{(\jointq),\eps}}

\newcommand{\sfkme}{\mu}
\newcommand{\sfkcov}{\Sigma}
\newcommand{\sfkmedot}{\mu(\circledcirc)}
\newcommand{\sfkcovdot}{\Sigma(\circledcirc)}
\newcommand{\sfop}{A}
\newcommand{\sfopdot}{A(\circledcirc)}
\newcommand{\opeigvec}{f^{\sfop(\proba)}}
\newcommand{\opeigval}{\lambda^{\sfop(\proba)}}
\newcommand{\sfeigvecdot}[1]{f_{#1}^{\sfop}(\circledcirc)}
\newcommand{\sfeigvaldot}[1]{\lambda_{#1}^{\sfop}(\circledcirc)}
\newcommand{\sfeigvec}[1]{f_{#1}^{\sfop}}
\newcommand{\sfeigval}[1]{\lambda_{#1}^{\sfop}}
\newcommand{\fgd}[1]{{#1}^{\prime}}
\newcommand{\sfnorm}{N}
\newcommand{\sfnormdot}{N(\circledcirc)}

\newcommand{\sfweigvecdot}[1]{f_{#1}^{\sfwkcov}(\circledcirc,\circledcirc)}
\newcommand{\sfweigvaldot}[1]{\lambda_{#1}^{\sfwkcov}(\circledcirc,\circledcirc)}
\newcommand{\sfweigvec}[1]{f_{#1}^{\sfwkcov}}
\newcommand{\sfweigval}[1]{\lambda_{#1}^{\sfwkcov}}

\newcommand{\heYanc}{\widehat{\eYanc}}
\newcommand{\heyanc}[1]{\widehat{\eyanc{#1}}}
\newcommand{\hancEpsilon}{\widehat{\varepsilon^{\Anc}}}
\newcommand{\hancepsilon}[1]{\widehat{\varepsilon^{\Anc}_{#1}}}

\newcommand{\hresanceigval}[1]{\widehat{\lambda^{\Anc}_{\errors,{#1}}}}
\newcommand{\hresanceigveckt}[1]{u^{\Anc}_{\errors,#1}}

\newcommand{\hip}[1]{\left \langle {#1} \right \rangle _{\Hk}}
\newcommand{\hn}[1]{\left \| {#1} \right \| _{\Hk}}

\maketitle

\begin{abstract}
    Kernel-based testing has revolutionized the field of non-parametric tests through the embedding of distributions in an RKHS. This strategy has proven to be powerful and flexible, yet its applicability has been limited to the standard two-sample case, while practical situations often involve more complex experimental designs. To extend kernel testing to any design, we propose a linear model in the RKHS that allows for the decomposition of mean embeddings into additive functional effects. We then introduce a truncated kernel Hotelling-Lawley statistic to test for the effects of the model, demonstrating that its asymptotic distribution is chi-square, which remains valid with its Nystr\"om approximation. We discuss a homoscedasticity assumption that, although absent in the standard two-sample case, is necessary for general designs. Finally, we illustrate our framework using a single-cell RNA sequencing dataset and provide kernel-based generalizations of classical diagnostic and exploration tools to broaden the scope of kernel testing in any experimental design.
\end{abstract}

\section{Introduction}

Statistical hypothesis testing has witnessed significant advancements in recent years, largely due to the integration of kernel methods with non-parametric testing. Indeed, it has now become common practice to develop testing procedures that involve embedding probabilistic distributions into Reproducing Kernel Hilbert Spaces (RKHS) to assess the equality of distributions \cite{gretton_kernel_2012}, dependencies 
\cite{gretton_measuring_2005}, causality \cite{Zhang_Causal}, and essentially any probabilistic feature that may distinguish two populations. Kernel testing has also facilitated the development of tests for complex data types, such as graphs and sequences, which surpass traditional non-parametric procedures like the Wilcoxon or Kolmogorov tests \cite{mann_test_1947,smirnov_estimation_1939}. Typically, kernel testing involves embedding data into a feature space using a kernel, with the statistical test based on the distribution of the embedded data, such as means and variances. A significant milestone in the field was the Maximum Mean Discrepancy (MMD) test, which measures the distance between mean embeddings of two conditions \cite{gretton_kernel_2006}. The MMD statistic asymptotically follows a mixture of chi-square distributions, that can be approximated using permutations. 

Parallel to the work of Gretton et. al \cite{gretton_kernel_2012}, a studentized version of the MMD was proposed to account for the variance of the embeddings directly in the test statistic \cite{harchaoui_testing_2007}, which constitutes a kernelized version of the Hotelling-Lawley trace test \cite{olive_robust_2017}. This strategy results in a metric that can be interpreted as the distance between mean embeddings projected onto the kernel-Fisher discriminant axis. In other terms, the normalized version of the MMD reduces to the norm of a Kernel Fisher Discriminant Analysis classifier used for non-linear hypothesis testing. The original version is based on ridge-based regularization of the residual covariance operator \cite{harchaoui_testing_2007}. This approach has recently been proven optimal with respect to the test's ability to best separate two distributions within a class of alternatives \cite{hagrass_spectral_2022}. These recent theoretical developments underscore the theoretical benefits of properly account for a noise model within the RKHS to gain power in the testing procedure through studentization. However, in practice, the ridge-based regularization test lacks interpretability and the asymptotic distribution is complex to compute. This motivated the development of the truncation-based test that is based on the eigen-decomposition of the residual covariance operator into $T$ principal components \cite{harchaoui_regularized_2009}. This results in a statistic that asymptotically follows a $\chi^2(T)$ distribution, which greatly increases its applicability since it does not require permutation-based strategies, at the cost of a $\mathcal{O}(n^3)$ complexity. This complexity can be mitigated through Nystr\"om-type approximations. Despite methodological and theoretical developments in line with the MMD framework, the application of kernel testing in practice has been limited to the simple scenario of two-sample testing, whereas practical situations almost always involve more complex experimental designs. If kernel testing is to become a standard in various application areas, we need to offer extensions that are both statistically sound and efficiently implemented.

In this work, we propose a unified framework that enables kernel testing for any experimental design. We introduce a linear model in the RKHS that decomposes the mean embedding of observed distributions into linear functional effects. Testing is then performed using the kernelized Hotelling-Lawley statistic to test linear combinations of effects. Indeed, the truncated studentized-MMD can be easily interpreted as a truncated Hotelling-Lawley test based on data embeddings. Moreover, this connection between kernel testing and the standard multivariate framework allows us to reinterpret the standard MMD procedure as a linear model within the RKHS, linking the mean embeddings to any experimental design. Interestingly, according to our results, the comparison of linear combinations of any effect (contrast testing) cannot be conducted without specific conditions on the residual covariance operator—conditions that do not arise in the simpler two-sample test. These conditions relate to the standard homoscedasticity hypothesis required in linear models, which states that the variance of the noise must be independent of the expectation provided by the model. We reinterpret this hypothesis for covariance operators in the RKHS and propose a restricted heteroscedasticity hypothesis that states that the covariance operator of the noise is homogeneous up to its first $T$ eigen-components. Under this hypothesis, we provide the asymptotic distribution of the test for linear combinations of effects (contrast tests), demonstrating that it follows a $\chi^2$ distribution that is straightforward to compute. We also verify theoretically that this distribution remains valid when using a Nystr\"om approximation, which enhances the practical applicability of kernel testing. We present simulations to empirically corroborate our theoretical developments. To advance the full operability of kernel testing, we offer a comprehensive framework that benefits from powerful representations and diagnostic plots, enabling the complete application of kernel testing in any experimental situation. Moreover, all these quantities are implemented in our package \textit{ktest}.

To illustrate our framework we fully analyze a dataset from single-cell transcriptomic data. Thanks to the combination of single-cell isolation techniques and massive parallel sequencing, it is now possible to create high-dimensional molecular portraits of cell populations \cite{macosko_highly_2015,zheng_massively_2017}. These advances have resulted in the production of complex, high-dimensional data, revolutionizing our understanding of the complexity of living tissues. The field of single-cell data science presents new methodological challenges, key among them being the statistical comparison of single-cell RNA sequencing (scRNA-Seq) datasets between conditions or tissues. This step is crucial for distinguishing biological from technical variabilities and for asserting meaningful expression differences. Comparative analysis of single-cell datasets, regardless of their type, is an essential component of single-cell data science, providing biological insights and opening therapeutic perspectives through the identification of biomarkers and therapeutic targets. The effectiveness of kernel testing in this context has recently been demonstrated \cite{ozier-lafontaine_kernel-based_2024}; however, the practical application of kernel testing in the field hinges on the ability to conduct a comprehensive analysis tailored to any experimental design.

\section{Testing linear combinations of functional effects in the feature space}

\subsection{A linear model in the feature space}
\label{sec:KernelLinearModel}

We consider \(\dY = (\dy_1,\dots,\dy_\nobs)\), a set of \(\nobs\) observations from a measurable space \(\sY\), jointly recorded with some explanatory variables encoded in the design matrix \(\dX = (\dx_1,\dots, \dx_\nobs)^\prime \in \mtrx{\nobs,\nx}\), where for \(i\in\{1,\dots,\nobs\}\), \(\dx_i = (\dx_i^{1}, \dots,\dx_i^{p})^\prime\in\R^\nx\) is the experimental design that is supposed to be fixed. To generalize kernel testing to any experimental design we start by considering a featurization of the response variable. For this purpose we introduce \(\kerneldot:\sY\times \sY \to \R\), a positive definite kernel associated with a separable Reproducing Kernel Hilbert Space (RKHS) \((\Hk,\|\|_{\Hk})\) called the feature space, and the feature map \(\fmap\cdot:\dy\in\sY\mapsto\kernel(\dy,\cdot)\in\Hk\). Then we propose a linear model in the feature space, such that for \(i\in\{ 1, \dots, \nobs \}\), we have:
$\ey{i} = \dx_i^{1}  \param_{1} + \dots +\dx_i^{p}  \param_{p}  + \error_i$, where  \(\E(\error_i)=0_\Hk\) and \(\mathbb{V}(\error_i)= \Sigma_i\), the residual covariance operator. 
This model corresponds to the decomposition of the mean embedding of the distribution of the response variable into linear combination of mean embeddings \(\params=(\theta_1, \hdots, \theta_{p})' \in \mathcal{H}^p\) that correspond to different factors. A matrix formulation of this model can also be considered, with \(\eY=(\ey{1},\dots,\ey{\nobs})'\in \mathcal{H}^n\), \(\errors=(\error_1,\dots,\error_\nobs)'\in \Hkn\), which provides the linear model in the feature space that resumes to:
\begin{equation}
\label{eq:linear model on embeddings}
\begin{split}
\eY = \dX \params+ \errors, 
\end{split}
\end{equation}
where we use a natural extension of the matrix product between \(\dX \in \mtrx{\nobs,\nx}\) and \(\params \in \Hk^{\nx} \). Functional parameters can then be easily inferred using a least squares estimator such that 
\begin{equation}
\label{eq: empirical model parameters}
\begin{split}
 \hparams =\XXX  \eY  \in \mathcal{H}^p,
\end{split}
\end{equation}
where \(\XXinv\) is a generalized inverse of \(\dX^\prime\dX\), see Appendix~\ref{append:ls} for a proof. Then the predicted embeddings according to the model are defined as \(\heY = \Px \eY\), where the orthogonal projection \(\Px = \Pxvalue \in\mtrx{\nobs}\) on \(\operatorname{Im}(\dX) \subset \R^\nobs\)   is also seen as an operator on  \(\Hk^\nobs\). The residuals of the model are defined by \(\residuals =(\residual_1,\dots,\residual_\nobs)  =\Pxp \eY  = \Pxp \errors  \in \Hk^\nobs\) where \(\Pxp = \identity_\nobs - \Px\) is the orthogonal projection on \(\operatorname{Im}(\dX)^\perp\).

\paragraph{One-way and two-way designs.} We provide two examples of linear models in the feature space that are most commonly encountered in practice. The simplest example is the one-way design model with $U\geq2$ levels of sizes $\nobs_1,\dots,\nobs_U$, such that $\nobs = \sum_{u = 1}^{U}\nobs_u$. For $u\in\{ 1, \dots, U \}$ and $\ell\in\{ 1, \dots, \nobs_u \}$, $Y_{u,\ell}$ stands for the $\ell^{th}$ observation of level $u$. The model is such that : $\mathbb{E}(\ey{u,\ell}) = \alpha_u$, where $\params = (\alpha_{1},\dots,\alpha_{U})$ are the model functional parameters. The standard two-sample test corresponds to the model with $U=2$ levels, which can easily be generalized to multiple comparisons. A natural extension would be to consider a two-way design, for example, to introduce blocking factors. In this setting, the first effect takes $U$ distinct values, and the second one takes $V$ distinct values. For $u,v\in \{ 1, \dots, U \} \times \{ 1, \dots, V \}$, we observe $\ell = 1, \hdots, \nobs_{u,v}$ observations associated with effects $(u,v)$. We propose to model these effects with the following additive linear model in the feature space:  $\mathbb{E}(\ey{u,v,\ell}) = \alpha_{u} + \beta_{v}$. We can then design contrast matrices to test if there is an influence of the first or the second effect, as seen in  Section~\ref{sec:Appli-Reversion}.

\paragraph{Conditional Mean Embedding.} For simplicity, our model is proposed for a deterministic design, but it can also be adapted for a random design. In this latter setting, Model \eqref{eq:linear model on embeddings} can be easily reinterpreted in terms of conditional mean embedding \cite{song2009hilbert}. Specifically, our model assumes that the conditional mean embedding \(\mathcal{U}_{Y|X}\) is the linear operator mapping \(x \in \R^p\) to \(x' \params \in \Hk\), by considering the identity embedding for \(x \in \R^p\).

\subsection{Hypothesis testing with the truncated kernel Hotelling-Lawley statistic} 

We propose to perform hypothesis testing on the model functional parameters with a standard contrast testing approach. Let $\Lmanova \in \R^{\nLmanova \times \nx}$ a surjective matrix. As for $\dX$, $\Lmanova$ can be seen both as a linear operator from $\R^p$ to $\R^d$, or as a  linear operator from $ \Hk^{\nx}$ to  $\Hk^{d}$. We use the contrast matrix $\Lmanova$ to test if a combination of $\params$ is null or not and we formulate the null and alternative hypotheses as : $H_0 : \Lmanova \params = 0$ vs. $H_1 : \Lmanova \params \neq 0$. In the setting of the multivariate linear model, several popular statistical tests exist for this type of hypothesis. We propose to extend the Hotelling-Lawley test to the linear model in the feature space \eqref{eq:linear model on embeddings}, which requires considering general operators defined on Hilbert spaces. In this section we introduce the statistic without providing all the justifications, a rigorous definition of the kernelized Hotelling-Lawley statistic is proposed in Appendices \ref{sec:appNotations} and \ref{App:KernelLinearmodel}, and more specifically in  Section \ref{sec:KHKwelldefined}. We first introduce the so-called test operator (Hilbert-Schmidt) associated with \(\Lmanova\):
\begin{equation}
\label{eq:HL}
 \Hmanova = (\Lmanova\hparams)^{\star} (\Lmanova\XXinv\Lmanova^\prime)^{-1} (\Lmanova\hparams),
\end{equation}
which generalizes the "between sum of squares" from ANOVA and the "cross products matrix" from MANOVA models to more complex designs in the RKHS setting of this paper. Next, we define the residual covariance operator as
$
\rcov = \nobs^{-1} \sum_{i = 1}^{\nobs} \residual_i \otimes \residual_i . 
$
The spectral theorem applies to this operator: 
$
\rcov =  \sum_{t = 1}^{r} \rcoval_t  \big( \rcovec_t \otimes \rcovec_t \big)
$, 
where $(\rcovec_t)_{t\geq 0}$ is an orthonormal basis of $\Hk$, where $\rcoval_t >0  $ and $r$ is the rank of $\rcov$.  We then introduce a generalized inverse for $\rcov$ as
$
 \rcov^{-} =   \sum_{t = 1}^{r}\rcoval_t^{-1}\big( \rcovec_t \otimes \rcovec_t \big).
$ 
We finally introduce the truncated kernelized Hotelling-Lawley (TKHL) statistic as 
\begin{equation*}
\begin{split}
\statmanova = \nobs^{-1} \trhs\left( \THLcov\right),
\end{split}
\end{equation*}
where $
\rcovt^{-1} :=  \sum_{t = 1}^{\tmax \wedge r}\rcoval_t^{-1} \big( \rcovec_t \otimes \rcovec_t \big) 
$ is defined from  $\rcov^{-}$  by a spectral truncation of depth $\tmax$. An alternative ridge-based regularization has been proposed for $\rcov^-$, in the spirit of kernel FDA \cite{harchaoui_testing_2007}. The spectral truncation we propose has the advantages of $i)$ providing a tractable chi-square asymptotic distribution under the null, as we show in Theorem~\ref{theorem: Asymptotic distribution of the truncated Hotelling-Lawley test statistic}, and $ii)$  providing a natural framework for data and model exploration (Section~\ref{sec:Appli-Reversion}). Computing this statistic is based on a kernel trick detailed in the Appendix (Section~\ref{App:KernelTricks}).

\paragraph{One-way and two-way designs.} In the case of the one-way design, the following "pairwise" contrast matrix allows to test $H_0:  \alpha_{1} = \alpha_{2} = \dots = \alpha_{U}$ versus $H_1:\exists i,i^\prime, \alpha_{u} \neq \alpha_{u^\prime}$:
\begin{equation*}
\label{eq: Lmanova one versus all}
\begin{split}
\Lmanova_{\alpha} = 
\left (\begin{matrix} 
1      & -1     &  0     & \hdots & \hdots & 0 \\
0      &  1     & -1     & 0      & \hdots & 0 \\ 
\vdots &        & \ddots &        &        & \vdots \\ 
\vdots &        &        & \ddots &        & \vdots \\
0      & \hdots & \hdots & 0      & 1 		 & -1 
\end{matrix} \right)
\end{split}
\end{equation*}
When $U=2$, this test precisely reproduces the two-sample test. In this setting, the truncated Hotelling-Lawley statistic is equal to the truncated KFDA statistic \cite{harchaoui_testing_2007,harchaoui_regularized_2009}. Our framework provides a straighforward generalization to the comparison of $U$ levels thanks to the truncated Hotteling-Lawley test. In the case of the two-way design, the contrast matrices $\begin{bmatrix}\Lmanova_{\alpha} & \mathbb{0}_{\beta} \end{bmatrix}\in \mtrx{U-1,U+V}$ and $\begin{bmatrix}\mathbb{0}_{\alpha} & \Lmanova_{\beta} \end{bmatrix}\in \mtrx{V-1,U+V}$ are used to test for the effects of the first and the second factor respectively, with their corresponding contrast matrices from the one-way design $\Lmanova_{\alpha}\in \mtrx{U-1,U}$ and $\Lmanova_{\beta}\in \mtrx{V-1,V}$, and zero matrices $\mathbb{0}_{\alpha}\in \mtrx{V-1,U}$ and $\mathbb{0}_{\beta}\in \mtrx{U-1,V}$. This allows to test for the influence of each factor on the featurized response variable independently of the other factor. For example, in the presence of a blocking factor, this framework allows to determine whether it has a significant effect and thus should be considered when evaluating the treatment factor. Our approach enables the application of reasoning similar to that used in standard linear model procedures, thereby greatly extending the scope of kernel testing. We provide an example in Section ~\ref{sec:Appli-Reversion}.

\section{Asymptotic distribution of the TKHL statistic}
\label{sec:asymptotic_TKHL}

We derive the asymptotic distribution of the TKHL statistic, demonstrating that it follows a chi-square distribution under the null, thereby generalizing the known results for the HL trace statistic to the RKHS framework. Additionally, we introduce a homoscedasticity hypothesis for the error terms, as is typical in vectorial models, and provide a functional generalization within the RKHS by assuming restricted homoscedasticity to the first \(\tmax\) eigen-directions of the residual covariance operator.
\begin{assumption}\label{Ahomoscedasticity}
The covariances $\ecov_i$ of the errors are all identical on their $\tmax$ first directions: there exists a symmetric positive  operator $\ecovt =  \sum_{t = 1}^{\tmax}\ecoval_t \ecovec_t^{\otimes 2}$ such that 
$\ecov_i = \ecovt + \ecovtci{i}$ for all $i\in\{ 1, \dots, \nobs \}$, where $\ecovtci{i}= \sum_{t >\tmax} \ecoval_{i,t} \ecovec_{i,t}^{\otimes 2}$ is a symmetric positive operator.
\end{assumption}
\begin{assumption}\label{Asimple}
The eigenvalues of $\ecoval_t $ of $\ecovt$ are all simple, and there exists $\mu_{\tmax+1} >0$ such that $ \ecoval_{i,\tmax+1} \leq \mu_{\tmax+1} <  \ecoval_{\tmax}$ for any $i\in\{ 1, \dots, \nobs \}$.
\end{assumption} 
Under Assumptions~\ref{Ahomoscedasticity} and~\ref{Asimple} and for all $i \in \{1, \hdots, n\}$, the eigenvalues of  $\ecov_i$  satisfy $  \ecoval_1  > \dots > \ecoval_\tmax > \ecoval_{\tmax +1}   \geq \ecoval_{i,\tmax+1}$.   Next assumption require the kernel function to be bounded:
\begin{assumption}\label{Aboundedkernel}
The kernel is bounded : $\left \| k(\cdot, \cdot) \right \| _{\infty} = M_k < + \infty$.
\end{assumption} 
Finally, we  make the three assumptions on the design $\dX = \dX(\nobs) \in \mtrx{\nobs,\nx}$. In the following we will omit the dependency on $n$ since it is clear that it depends on $n$.
\begin{assumption}\label{Aboundeddesign} There exist $\boundX \in \R$ such that for all $\nobs \geq 1$, the design matrix $\dX$ is such that $\left \| \dX \right \| _{\infty} := \sup_{i,j} | x_{i,j}| \leq \boundX$.
\end{assumption}  
\begin{assumption}\label{ADiagDesign} The maximal coefficient on the  diagonal of $\Px$ satisfies $ \max_{i = 1 \dots n}  \left| (\Px)_{i,i} \right| \rightarrow 0 $ as $n$ tends to infinity.
\end{assumption}  
\begin{assumption}\label{Aconvergentdesign} The sampling gives a convergent design : $ \nobs^{-1}  \dX^\prime \dX   \to  \XXlim^{-1}  \in \R^{\nx \times \nx}$ as $n$ tends to infinity.
\end{assumption} 
Assumptions are discussed after the following Theorem which demonstrates that the asymptotic distribution of the TKHL statistic is not only distribution-free but also tractable.  
\begin{theorem}
\label{theorem: Asymptotic distribution of the truncated Hotelling-Lawley test statistic}
Under the linear model \eqref{eq:linear model on embeddings} in the feature space $\Hk$, assume that Assumptions \ref{Ahomoscedasticity} to 
\ref{Aconvergentdesign} are satisfied. If $H_0$ is true  then $
\nobs  \statmanova  \underset{\nobs\to \infty}{\overset{\mathcal{D}}{\longrightarrow}} \chi^2_{\nLmanova \tmax}$.
\end{theorem}
Our result generalizes the asymptotic convergence of the HL trace statistic, commonly used in multivariate linear models (see, for instance, Theorem 12.8 from \cite{olive_robust_2017}), to cases where data lies in a RKHS. In this functional framework, establishing the asymptotic properties is more intricate since the $T$ directions used to define the TKHL are stochastic.

The proof of the Theorem is given in Appendix~\ref{App:Proof-Theo-chi2}, see in particular Section~\ref{App:Proof-Theo-chi2:main} for a summary of the proof. The proof relies on the convergence in probability of $\rcov$ towards $\ecov$, which we establish thanks to a an approach inspired from \cite{shawe-taylor_eigenspectrum_2005} and \cite{zwald_convergence_2005}. This approach combines a perturbation bound with a bounded difference theorem. In our framework we must control the residual covariance operator, rather than a standard covariance operator as is typically done in kernel PCA.

Regarding the assumptions of the theorem, Assumption \ref{Ahomoscedasticity} relates to the classical homoscedasticity hypothesis required in linear models, generalized here to covariance operators with a restriction on the first $T$ directions. This hypothesis is typically verified visually using diagnostic plots, which we generalize for the feature space (Section~\ref{sec:Appli-Reversion}). In scenarios where testing involves the equality of $U$ distributions with a characteristic kernel, the null hypothesis entails the equality of the $U$ covariance operators. In such cases, the homoscedasticity assumption may be omitted as specified in Theorem~\ref{theorem: Asymptotic distribution of the truncated Hotelling-Lawley test statistic}. However, when extending this approach to general designs, maintaining the homoscedasticity assumption is crucial and cannot simply be substituted with a {\it full homoscedasticity} assumption, which demands equality of all covariances across all directions in $\Hk$. Given that the covariance operator happens to be itself a mean embedding that may characterize the distribution in certain contexts (see for instance Proposition 1 in \cite{bach_information_2022} and also \cite{bonnier_kernelized_2024}), a full homoscedasticity assumption would impose a stronger condition than the null hypothesis itself. Note that Assumption \ref{Ahomoscedasticity} is only required under the null. Assumption \ref{Asimple} is restrictive but essential as our statistic is studied direction-wise, which requires a strictly positive gap around each eigenvalue  \cite{blanchard_statistical_2007, zwald_convergence_2005}. Assumption \ref{Aboundedkernel} is a standard hypothesis for deriving statistical results on kernel methods and is typically satisfied by Gaussian or Laplacian kernels. Assumption \ref{Aboundeddesign} is used in addition to Assumption~\ref{Aboundedkernel} to obtain a bound on the errors and to show that their fourth moment is finite (see Lemma \ref{lemma: Borne sur la norme des residus}). Note that the fourth moment of the errors are also assumed to be finite in the standard HL test (see Theorem 12.8 from \cite{olive_robust_2017}). Finally, Assumptions \ref{ADiagDesign} and \ref{Aconvergentdesign} are also needed to obtain the asymptotic distribution of the standard HL statistic (see for instance Theorem 12.8 from \cite{olive_robust_2017}). 

\section{The Nystr\"om TKHL statistic}
\label{sec:NystromTKHL}

The computational complexity of kernel testing based on the TKHL statistic depends on the diagonalization of the residual covariance operator $\widehat{\Sigma}_T$, or its kernel-trick equivalent, the gram matrix $\rcovkt$ of general term $\langle \residual_i , \residual_{i^\prime} \rangle_{\Hk}$, which requires $\mathcal{O}(\nobs^3)$ operations. To reduce the computational burden, the Nyström approximation consists in finding a low-rank approximation such that $\rcovkt \simeq \rcovny \rcovny^\top$ with $\rcovny\in\mtrx{\nobs,\nanchors}$, $\nanchors \ll \nobs$, which reduces the cost to $\mathcal{O}(\nanchors^3)$. To choose a matrix $\rcovny$ that ensures the same asymptotic distribution as the original TKHL statistic, we uniformly sample $\nlandmarks$ landmarks from the $\nobs$ observations and define the $\nanchors \geq \tmax$ first unit eigenfunctions $(\anc_1,\dots,\anc_\nanchors)$ of their associated residual covariance operator as Nyström anchors. Finally we build matrix $\rcovny$, as the $\nobs \times \nanchors $ matrix whose rows are the projections of residuals $\residual_i$ onto the eigenfunctions $\anc_j$, with general term $\rcovny_{i,j}=\langle \residual_i , \anc_j\rangle_{\Hk}$, $i=1,\hdots, n$, $j=1, \hdots, m$. We then define the Nyström residual covariance operator $\arcov$ as $\arcov=\nobs^{-1}\sum_{i = 1}^{\nobs} \aresidual_i \otimes \aresidual_i$, where for $i\in\{ 1, \dots, \nobs \}$, $\aresidual_i$ is the orthogonal projection of $\residual_i$ onto $\Span(\anc_1,\dots,\anc_\nanchors)$. We denote by $\arcovec{1},\dots,\arcovec{\nanchors}$ the orthonormal eigenfunctions of $\arcov$ associated with the non-increasing eigenvalues $\arcoval{1},\dots,\arcoval{\nanchors}$, so that $\arcovt=\sum_{t = 1}^{\tmax}\arcoval{t} \arcovec{t}\otimes \arcovec{t}$ is the spectral truncation of $\arcov$. We can then introduce as before a generalized inverse for $\arcovt$ as $\arcovt{}^{-1} =  \sum_{t \leq \tmax, \,  \arcoval{t} >0,  }
\arcoval{t}^{-1} \arcovec{t}\otimes \arcovec{t}$. We finally define a Nyström version of the kernel Hotelling-Lawley Statistic  as
\begin{equation}
\label{eq:nystrom_statistic}
\begin{split}
\statmanovany = \frac{1}{\nobs} \trhs\Big( \arcovt{}^{-1} \Hmanova \Big).
\end{split}
\end{equation}
Note that according to this version of Nyström we don't use the landmarks to approximate $\Hmanova$, as the computation cost of this last is not the worst part in the TKHL statistic.
 
We now study the asymptotic distribution of the Nyström TKHL statistic defined in the linear model \eqref{eq:linear model on embeddings}. We consider  the Nyström TKHL statistic defined with $\nlandmarks$ landmarks,  $\nanchors$ anchors and a truncate parameter $\tmax \leq \nanchors$. In the following the parameters $\tmax$ and $\nanchors$ are fixed, contrary to $\nobs$ and $\nlandmarks$, and we assume that $\nlandmarks \leq \nobs $.  We need to rewrite Assumptions~\ref{Ahomoscedasticity} and \ref{Asimple} according to the numbers of anchors:
\setcounter{assumptionexp}{0} 
\begin{assumptionp}\label{AhomoscedasticityN}
 The covariances $\ecov_i$ of the errors are all identical on their $\nanchors$ first directions: there exists a symmetric positive  operator $\ecovanc=  \sum_{t = 1}^{\nanchors}\ecoval_t \ecovec_t^{\otimes 2}$ such that for $i\in\{ 1, \dots, \nobs \}$, $\ecov_i = \ecovanc + \ecovtci{i}$, where $\ecovtci{i}= \sum_{t >\tmax} \ecoval_{i,t} \ecovec_{i,t}^{\otimes 2}$ is a symmetric positive operator.
\end{assumptionp}
\begin{assumptionp}\label{AsimpleN}
The eigenvalues   of $\ecovanc$ are all simple, and there exists a positive constant $\mu_{\nanchors+1}$ such that for any $i\in\{ 1, \dots, \nobs \}$,  $ \ecoval_{i,\nanchors+1} \leq \mu_{\nanchors+1} < \ecoval_{\nanchors}$.
\end{assumptionp} 

\begin{theorem}
\label{theorem: Asymptotic distribution of the nystrom statistic}
Under Assumptions~\ref{AhomoscedasticityN}, \ref{AsimpleN}, 
\ref{Aboundedkernel}, \ref{Aboundeddesign},  \ref{ADiagDesign} 
 and \ref{Aconvergentdesign}, if $H_0$ is true  then: 
$\nobs  \statmanovany  \underset{\nlandmarks,\nobs\to \infty}{\overset{\mathcal{D}}{\longrightarrow}} \chi^2_{\nLmanova \tmax},
$
where in the convergence both $\nlandmarks$ and $\nobs$ tends to infinity and $\nlandmarks \leq \nobs$.
\end{theorem}
The proof of the Theorem is given in Section~\ref{App:Proof-Theo-chi2:Nystrom}. In the theorem, landmarks are selected uniformly from the observations. However, in practice, choosing landmarks uniformly across observations may not always be appropriate. For instance, when comparing unbalanced populations, under-represented groups could be over-sampled to ensure that their within-covariance matrices are accurately estimated. This approach also improves the accuracy of reconstructing the within covariance operator using the anchors.

\section{Simulations}

We evaluate the empirical performance of our method in the two-sample framework ($U=2$, Fig. \ref{fig:simulations}). We compare the empirical level and power of the TKHL test for different values of $n$ and various levels of Nystr\"om approximation ($q \in \{0.1 \times n, 0.50 \times n\}$, $m \in \{5, 25, 50\})$. We use the simulations based on the MNIST dataset \cite{choi2024} to differentiate the distribution of even from odd-number images. Each observation consists of a 7$\times$7 image with $\gamma\%$ of odd numbers: $\gamma=0$ for the null, and $\gamma \in \{0.1, 0.2\}$ to compute power ($\alpha=0.05$). Regarding type-I errors, the asymptotic TKHL test attains the nominal level, even for $n=200$. The Nystr\"om approximation does not alter the level, and moderate truncation has little impact on type-I errors. However, the empirical quantile of the TKHL statistic only matches the $\chi^2$ quantile for moderate values of $T$. Another simulation study (in the univariate case) showed that type-I errors increase with $T$ \cite{ozier-lafontaine_kernel-based_2024}. Hence $T$ is connected to the geometry of the distribution in the RKHS and should be chosen relatively low. The empirical power increases with $n$ and $\gamma$, as expected, with negligible impact of the Nystr\"om approximation: sub-sampling 10$\%$ of the data is sufficient to attain full power. Similarly, the impact of $m$, the number of anchors, is negligible. The effect of truncation on empirical power depends on the difficulty of the task. For low signal ($\gamma=0.1$), more principal directions are needed to reach maximum power ($T \sim 30$), whereas when for high signal ($\gamma=0.2$), full power is achieved with $T<10$ when $n=2000$. The MMD and TKHL show comparable power for $n=200$ and $n=2000$, but the power is higher for TKHL when $n=2000$. Our method being based on an asymptotic approximation, the TKHL is much faster to compute. Additionally, while the frameworks of the two tests are comparable for the two-sample test, there is no equivalent of our TKHL for more complex designs (i.e., the MMD only performs two-sample tests). Hence these simulations show that the asymptotic regime is reached at $n=200$ and that the Nystr\"om approximation has negligible impact while preserving computing resource. As for the truncation parameter, a trade-off should be established empirically to ensure good type-I error control to achieve best power.

\begin{figure}
\begin{center}
\includegraphics[width=\linewidth]{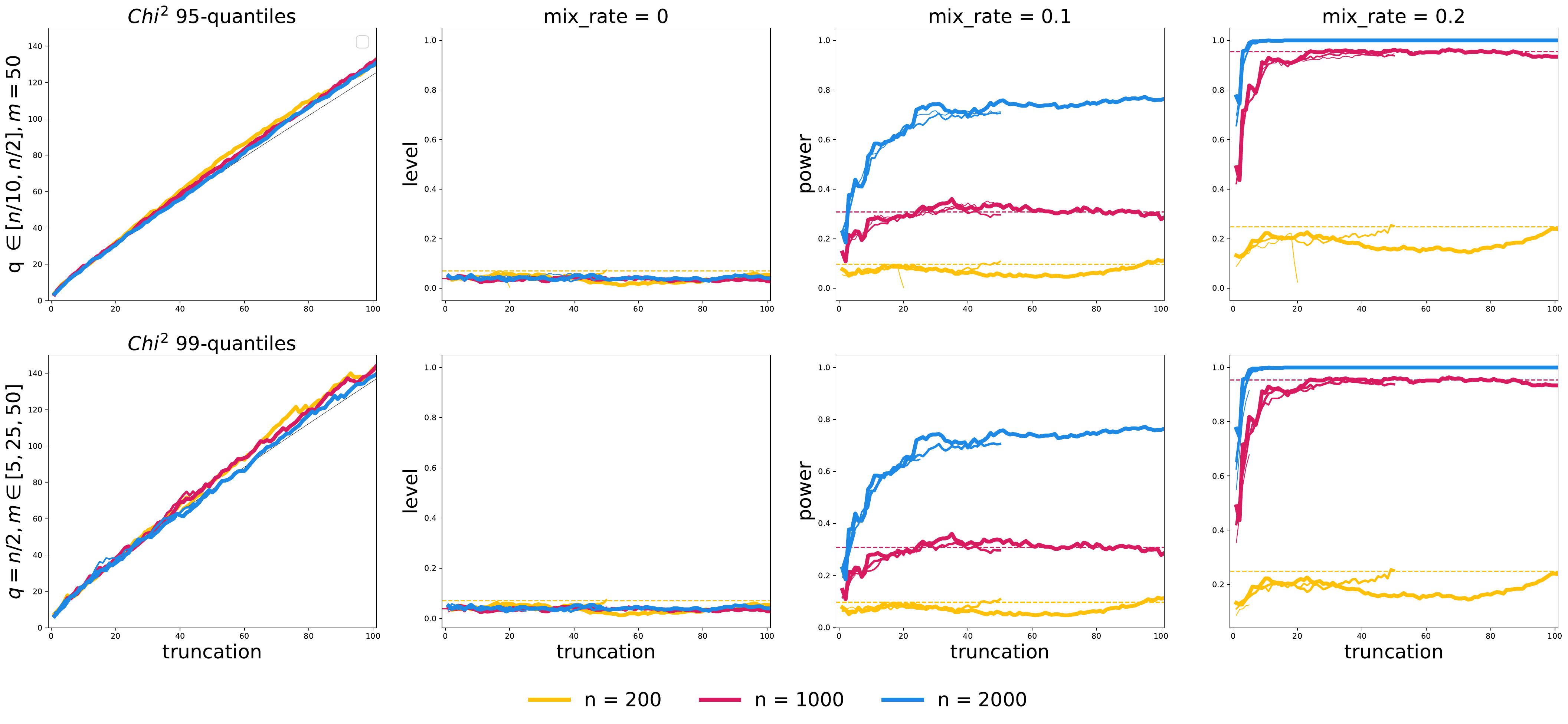}
\end{center}
\caption{Convergence of the TKHL empirical quantiles (95$\%$, 99$\%$, for different $n$), towards the theoretical $\chi^2$ quantile (black plain line), according to truncation $T$ under $H_0$ ($\gamma = 0$). Empirical level (col. 2) and power (col. 3-4), of the TKHL test (plain line) and MMD (dotted line) according to the truncation parameter $T$. Lines are interrupted because $T \leq m$. \texttt{Mix\_rate} corresponds to the proportion of odd numbers ($\gamma$). Curve widths correspond to different values of $q$ (row. 1) and $m$ (row. 2).}
\label{fig:simulations}
\end{figure}

\section{Leveraging kernel testing in single-cell transcriptomics}
\label{sec:Appli-Reversion}

\paragraph{Experimental design and model.} We analyze a scRT-qPCR dataset, a method that enables sensitive and targeted quantification of the transcriptome. The dataset\footnote{Data are available under licence CC0, under the SRA repository number SRP076011} contains the expression measurements of 83 genes on 685 cells \cite{zreika_evidence_2022}. In this context single-cell transcriptomics has been performed to investigate the cell differentiation process of chicken primary erythroid progenitor cells (T2EC). The experimental design is the following: to investigate the impact of the medium on cell differentiation, undifferentiated cells were initially put in a self-renewal medium (\texttt{0H} level of the medium effect), then put in a differentiation-inducing medium for 24h (\texttt{24H}). The population was then split into a first population maintained in the same medium for an additional 24h to achieve differentiation (\texttt{48HDIFF}), the second population was put back in the self-renewal medium to investigate potential reversion (\texttt{48HREV}). Before testing for the medium effect, we must account for the fact that the data were acquired on different dates, which induces an 8-level batch effect. We consider the linear model in the feature space: $\ey{u,v,\ell} = \alpha_{u} + \beta_{v} + \error_{u,v,\ell}$,
where $Y_{u,v,\ell}$ is the vector of log-normalized expressions (size 83) for cell $\ell \in \{1,\hdots,n_{u,v}\}$ in batch $u \in \{1, \hdots, U\}$ ($U=8$), and medium $v \in \{\texttt{0H},\texttt{24H},\texttt{48HDIFF},\texttt{48HREV}\}$, and $\phi(Y_{u,v,\ell})$ its featurization. Then $\alpha_{1},\dots,\alpha_{8}$ are the model parameters associated with the batch effect and $\beta_{\texttt{0H}},\beta_{\texttt{24H}},\beta_{\texttt{48HDIFF}},\beta_{\texttt{48HREV}}$ are the model parameters associated with the medium effect. Analyses were performed with the Gaussian kernel with bandwidth calibrated by the median heuristic \cite{garreau_large_2018}.

\paragraph{Diagnostic plots.} In practice, any analysis based on a linear model begins with diagnostic plots that enable the visual inspection of linearity and homoskedasticity hypotheses, as well as to perform model checking \cite{johnson_applied_2002,olive_robust_2017}. To provide such an operational framework in the feature space, we consider the projection of the mean effects \(\heY = \Px \eY\) and the residuals of the model \(\eY - \heY\) onto the eigen-elements of the residual covariance operator \(\widehat{\Sigma}_T\) obtained from the full model. Visual inspection on the first direction ($\widehat{f}_1$) confirm that the assumptions of linearity and homoscedasticity in the feature space are met (Fig. \ref{fig:proj}), also confirmed on other directions (not shown). 

\paragraph{Testing with the TKHL statistic.}We initiated our analysis by examining technical variability in the experiment, starting with a test for batch effects before assessing the main "medium" effect, which is of biological interest. We tested the null hypothesis of no batch effect, \(H_0: \alpha_1 = \hdots = \alpha_8\), using the linear model without the medium effect $\ey{u,\ell} = \alpha_{u} + \error_{u,\ell}$ and the contrast matrix \(\mathbf{L}_\alpha\) of \eqref{eq: Lmanova one versus all}. using the TKHL method based on the residual covariance operator associated with the one-way model with $\alpha_1,\dots,\alpha_8$ only and the corresponding \(\mathbf{L}_\alpha\) written in \eqref{eq: Lmanova one versus all}. This test revealed a significant difference between batches (\(p\)-value = \(1.35 \times 10^{-3}\) for \(T=1\). This effect was also detected by a linear MANOVA \cite{Coolyer2018}, \(p\)-value \(<10^{-3}\). This initial finding indicates that testing for the medium effect must be conducted after correcting for the batch effect, which in our setting is achieved by testing for the medium effect, \(H_0:\beta_{\texttt{0H}}=\beta_{\texttt{24H}}=\beta_{\texttt{48HDIFF}}=\beta_{\texttt{48HREV}}\), using the TKHL method based on the residual covariance operator associated with the full model and a matrix \(\mathbf{L}_{\beta|\alpha}\) corresponding to \eqref{eq: Lmanova one versus all} supplemented with columns of zeros on  corresponding to $\alpha_1,\dots,\alpha_8$. This test showed highly significant differences (\(p\)-value = \(4.08 \times 10^{-44}\) for \(T=1\)), indicating that gene expression varies significantly across at least one medium compared to others (not detected by MANOVA, \(p\)-value = 0.993). To deepen the analysis, we complemented this global test with individual pairwise comparisons to identify media that differ from one another. For this, we considered the pairwise hypotheses \(H_0^{v,v'}:\beta_v = \beta_{v'}\) for $(v,v') \in \{\texttt{0H},\texttt{24H},\texttt{48HDIFF},\texttt{48HREV}\}^2$, using for instance $\mathbf{L}_{\texttt{0H},\texttt{24H}}=(1,-1,0,0,0,0,0,0,0,0,0,0)$. This test revealed no statistically significant difference between the media \(\texttt{0H}\) and \(\texttt{48HREV}\) (\(p\)-value = \(1\) for \(T=1\), after the Benjamini–Hochberg (BH) adjustment for multiple comparison). Other media pairs were confirmed to be different, with the most distinct media being \(\texttt{0H}\) and \(\texttt{48HDIFF}\) (\(p\)-value = \(3.7 \times 10^{-31}\) for \(T=1\), BH-adjusted). The test statistic can be used as a kernelized distance between media, which scores the proximity between biological conditions (Supp. Fig. \ref{fig:heatmaps}), that are consistent across different values of the truncation parameter.

\paragraph{Representation of the linear decomposition.} To enhance the visual exploration of our model in the feature space, we propose a representation of the observations based on the linear decomposition suggested by our model. We extend the approach commonly employed in Kernel Fisher Discriminant Analysis (KFDA), which involves inspecting the empirical density of observations projected onto the discriminant axis \cite{ozier-lafontaine_kernel-based_2024}. Specifically, for the linear model in the feature space, we consider the eigen-directions of the operator \( \widehat{\Sigma}_T^{-1} \widehat{H}_{\mathbf{L}} \), used to test for the "batch" effect. Although this operator has a rank of \(I-1\), we simplify our analysis by focusing on the projection of observations onto its first direction (Fig. \ref{fig:proj}) that we call the discriminant axis. We believe this representation provides a powerful visualization of the variability of the data with respect to the linear decomposition of the mean embedding. When applied to single-cell transcriptomic data, this visualization initially demonstrates that despite the statistical significance of the "batch" effect, cells from different batches appear closely positioned on the discriminant axis dedicated to the batch effect, suggesting minimal practical differences between batches. Subsequently, when applied to the operator used for assessing the "medium" effect, the visualization underscores the proximity of cells in the initial medium (\texttt{0H}, blue) to those returned to the self-renewal medium (\texttt{48HREV}, purple), indicating a reversal in their differentiation. Cells from other media (\texttt{24H}, green; \texttt{48HDIFF}, red) are organized chronologically along the discriminant axis dedicated to the medium effect, which may then be interpreted as a differentiation trajectory of cells over time.

\paragraph{Detecting influential observations with the kernelized Cook distance.} We conclude our analysis by proposing a kernelized version of the Cook's distance that is classically used to detect influential observations in linear models \cite{diaz-garcia_sensitivity_2005}. Denoting by $\hparams_{(i)}$ the estimated parameters obtained when ignoring the $i^{th}$ observation, we propose a truncated and kernelized Cook's distance adapted to any linear combination of $\hparams$ to match with the TKHL test we perform, such that
\begin{equation*}
\mathcal{D}_{\mathbf{L},\tmax}(i) 
: = 
 \nLmanova^{-1} \trhs\left( (\Lmanova(\hparams - \hparams_{(i)}))^\star 
(\Lmanova\XXinv\Lmanova')^{-1}  \otimes  \widehat{\cov}_\tmax^{-1} 
(\Lmanova (\hparams- \hparams_{(i)}))\right),
\end{equation*}
The justification and the computation of this quantity are developed in Appendix \ref{app:diagnostics}. We conclude our analysis by displaying \(\mathcal{D}_{\mathbf{L}_\beta | \alpha}\) according to the position of cells along the discriminant axis (Fig. \ref{fig:cook}) to identify observations that leverage the medium effect. Our results show that the most influential observations are cells from media \(\texttt{0H}\) and \(\texttt{48HREV}\), positioned at the right margin of the discriminant axis. This positioning may indicate cells with unexpected gene expression patterns in relation to the differentiation process. Since the interpretability of kernel methods relies solely on individual observations rather than input features, identifying outlier observations is crucial in practice to fully exploit the potential of the linear model in the feature space.


\begin{figure}
  \centering
   \begin{subfigure}{0.49\textwidth}
     \centering
          \includegraphics[height=5.2cm]{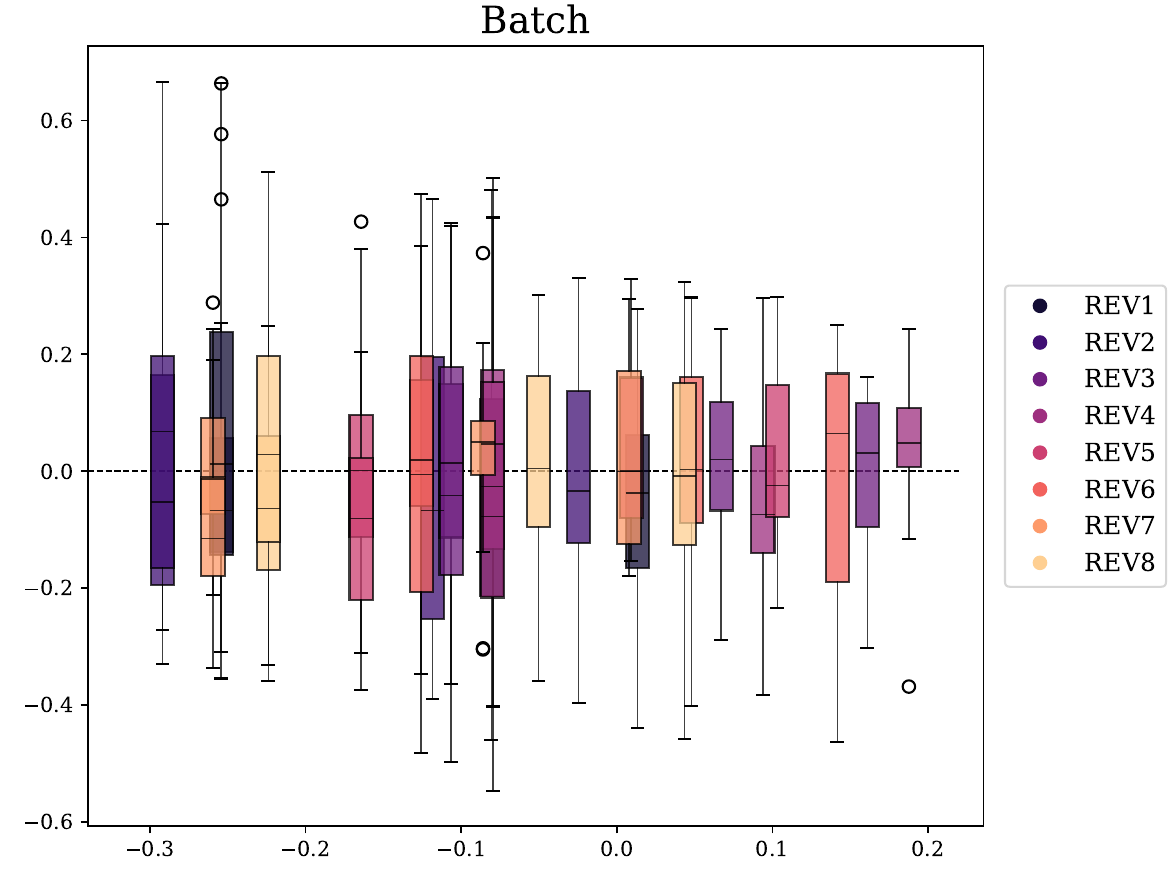}
      \label{fig:residual_plot_batch}
      \caption{}
    \end{subfigure}
    \begin{subfigure}{0.49\textwidth}
      \centering
      \includegraphics[height=5.2cm]{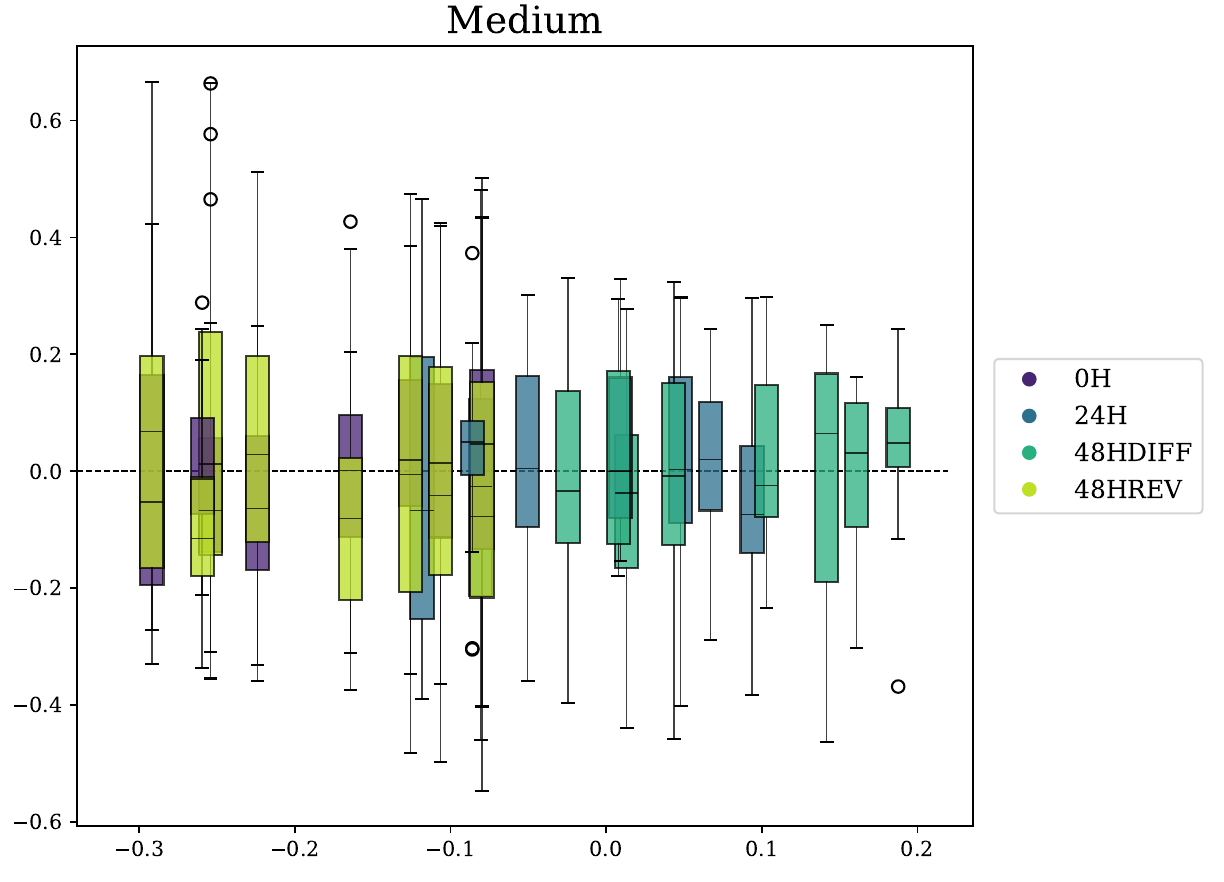}
      \label{fig:residual_plot_medium}
      \caption{}
    \end{subfigure}
    
    \begin{subfigure}{0.49\textwidth}
        `\centering
      \includegraphics[height=5.2cm]{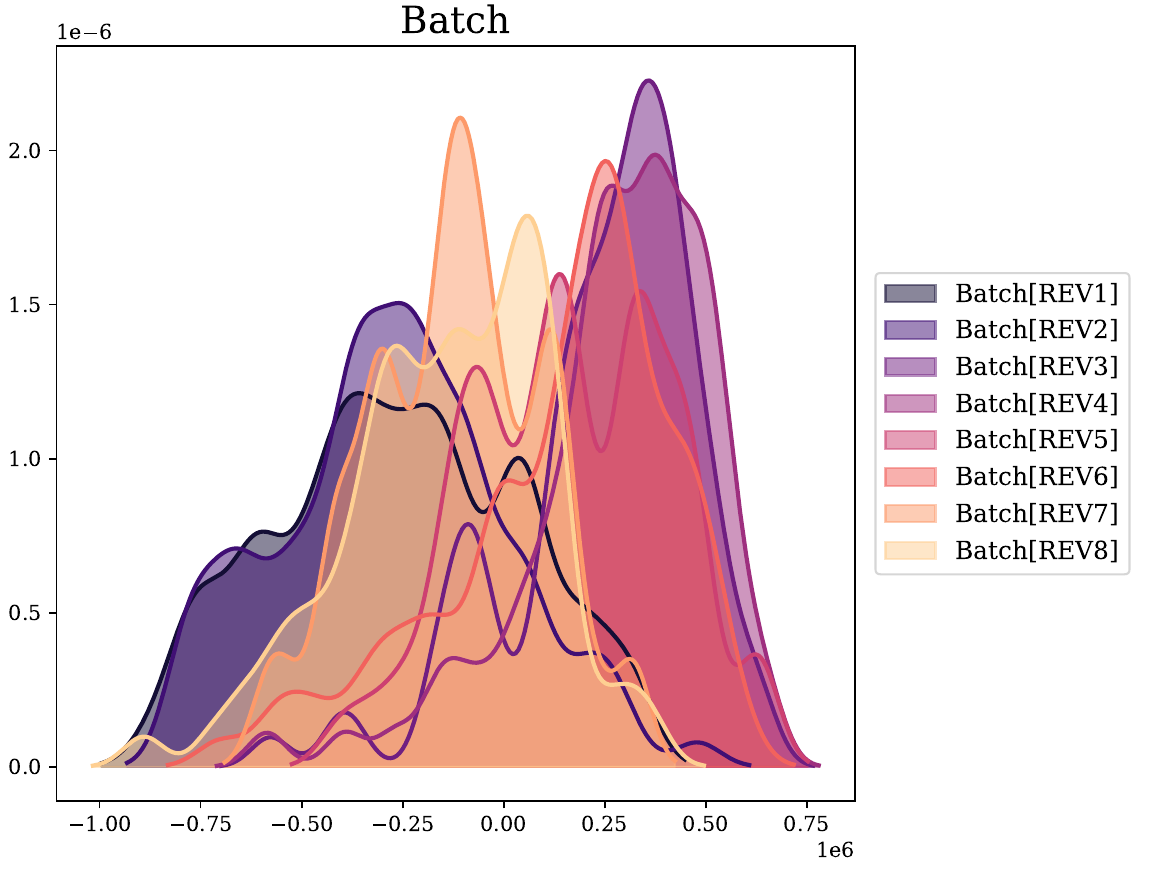}
      \label{fig:projections_batch}
      \caption{}
    \end{subfigure}
    \begin{subfigure}{0.49\textwidth}
      \centering
      \includegraphics[height=5.2cm]{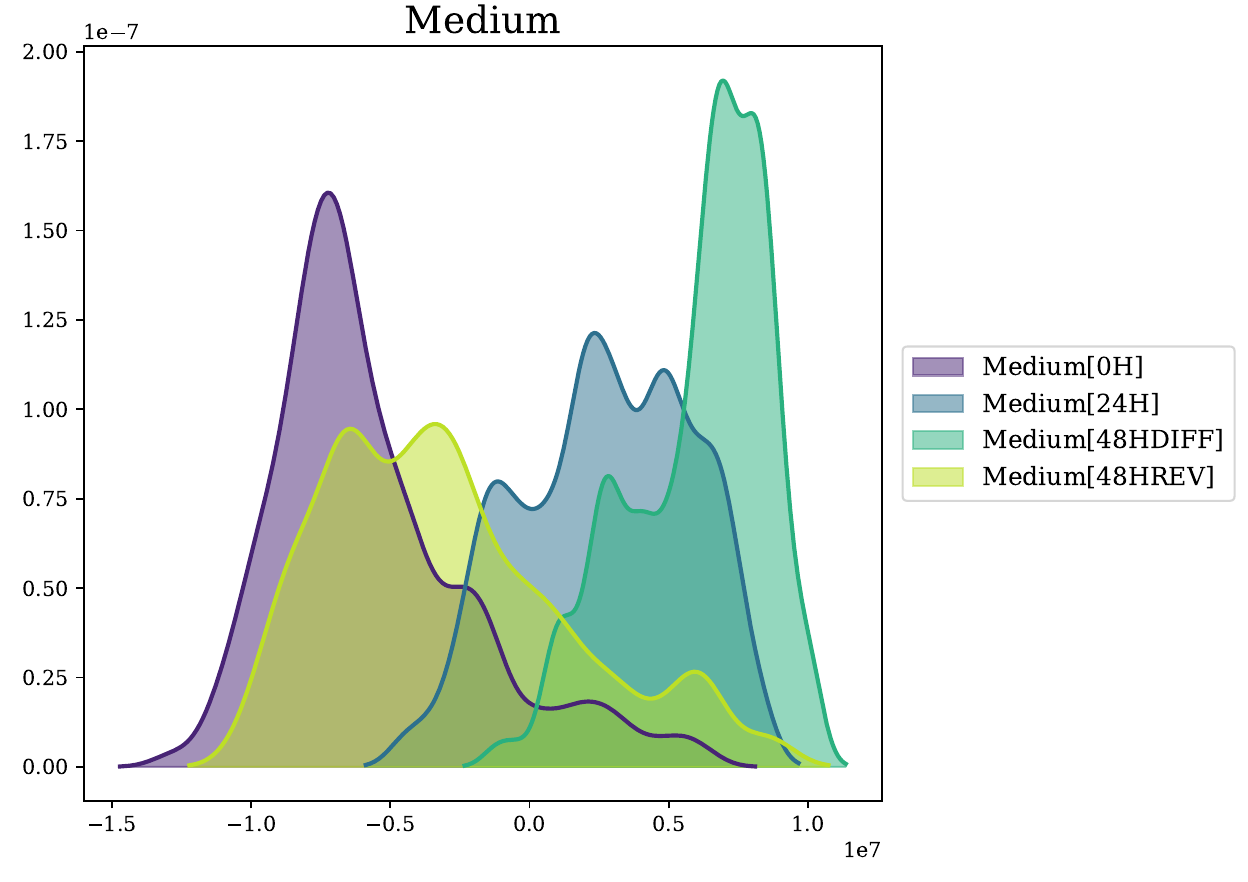}
      \label{fig:projections_medium}
      \caption{}
    \end{subfigure}
  \caption{(a, b) Residual plots associated with $t=1$, with colors indicating batches (a) and media (b). (d, e) Distributions of the response embeddings, projected on discriminant axes associated with testing for the batch effect (d) and the medium effect (e).}
\label{fig:proj}
\end{figure}

\begin{figure}
  \centering
   \begin{subfigure}{\textwidth}
     \centering
          \includegraphics[width=\linewidth]{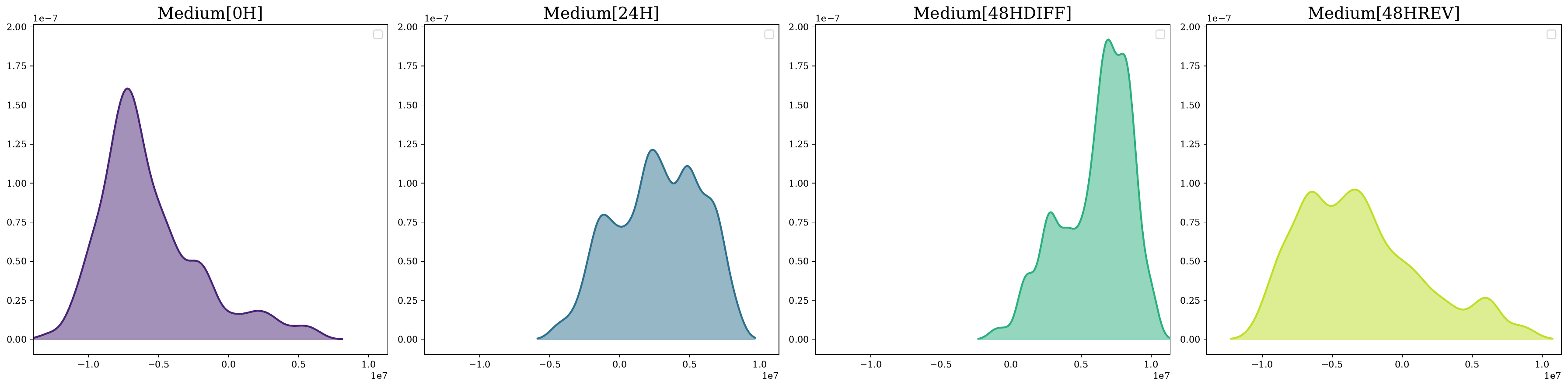}
      \label{fig:projections_by_level}
      \caption{}
    \end{subfigure}
    \begin{subfigure}{1\textwidth}
      \centering
      \includegraphics[width=\linewidth]{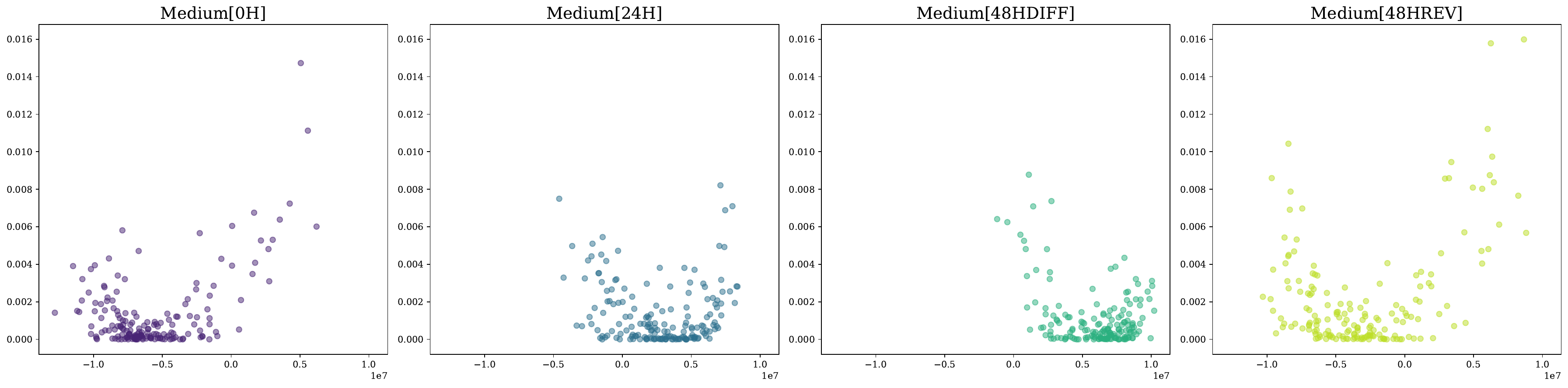}
      \label{fig:cook_by_level}
      \caption{}
    \end{subfigure}
  \caption{(a) Distributions of the response embeddings for each medium, projected on discriminant axes associated with testing for the medium effect. (b) Influence of the observations for each medium, plotted against their projection on the discriminant axis associated with the TKHL test statistic ($\tmax=1$).}
\label{fig:cook}
\end{figure}

\section{Discussion and Limitations}
\label{sec:discussions}

Here, we introduce a novel framework that generalizes the standard two-sample kernel test to accommodate any general design, thereby significantly expanding the applicability of kernel mean embedding to practical experimental designs. Our approach combines the benefits of linear models—particularly in terms of the interpretability of the linear decomposition of mean embeddings in the RKHS—with the proven power of non-parametric kernel testing. The statistic we propose constitutes a functional generalization of the Hotelling-Lawley trace test and can also be easily understood as a ratio between the model-driven variance operator and the residual variance operator. Additionally, classical strategies such as Wilk's Lambda, Pillai's Trace, and Roy's Largest Root, which are all based on the two operators $\Hmanova$ and $\rcov$, could also be easily adapted to the RKHS framework. The connection we propose between linear models and kernel testing allows us to enrich kernel testing with well established diagnostic and exploratory tools, extending its  applicability and interpretability. 

Our theoretical work establishes that the truncated kernel Hotelling-Lawley (TKHL) statistic asymptotically follows a chi-square distribution under the null hypothesis. From a computational standpoint, regularization via truncation proves more tractable than ridge penalization. This advantage is further underscored by our second theorem, which demonstrates that this asymptotic behavior also holds true for the Nyström version of the TKHL, thereby extending its applicability to large datasets where such assumptions are practical. Although permutations offer an alternative strategy for computing quantiles of the null distribution, the complexity of studentized functional contrast testing, which scales as \(\mathcal{O}(n^3)\), makes easily computed approximations essential, if not indispensable, for handling large data. Regarding the selection of landmarks in the Nyström procedure, the generalization of kernel testing to any design presents a new challenge: selecting landmarks that accurately reflect the variability induced by the experimental design. Traditional uniform sampling should be adapted to incorporate more advanced approaches \cite{li2016fast}. An alternative strategy to reduce the computational burden could be to consider Random Fourier Features, as recently proposed for MMD \cite{choi2024}.

In future work, we plan to conduct a non-asymptotic analysis of our test procedure, similar to recent minimax power results achieved in ridge-regularized scenarios \cite{hagrass_spectral_2022}. This direction could pave the way for theoretically grounding the choice of the truncation hyper-parameter, currently determined empirically. From a theoretical point of view, the main theorems are valid for any $T$. But the simulation study already provides some information regarding empirical performances: the truncation parameter impacts the power by allowing a more detailed description of the distributions' geometry in the RKHS as it increases. However, a trade-off must be made because as \(T\) increases, the eigenvalues decrease and become more difficult to estimate. Another promising avenue involves establishing a link between the linear decomposition we propose in the RKHS and the initial distributions compared in the input space. This link is guaranteed in the two-sample case by the injectivity of the feature map and the characteristic property of the kernel. However, general interpretation of the test in terms of input distribution is no longer possible when the mean embedding is decomposed onto linear subspaces. Hence, conclusions drawn from our tests are only valid in the RKHS, and not in the input space. Consequently, the relationship between the linear functional effects and their distributional counterparts in the feature space remains to be elucidated. 

Finally, our current example employs a Gaussian kernel calibrated by the median heuristic, but our flexible framework could potentially incorporate aggregation strategies to broaden the choice of kernels, as suggested by recent studies \cite{schrab_mmd_2022}.

\newpage

\newpage
\appendix

\section{Operators on $\Hk$ and on $\Hk^k$}
\label{sec:appNotations}

In  this section, we provide notation and we define properly all the objects involved in the definition of the  TKHL statistic.
 
\subsection{Operators on $\Hk$}

Let $(\Hk,\langle ,  \rangle_{\Hk}) $ be a separable Hilbert space. A linear operator $C$ from $\Hk$ to $\Hk$ is called Hilbert-Schmidt if $\sum_{i \geq 1} \| C f_i \|_{\Hk}^2   < \infty$, where $(f_i)_{i \geq 1}$ is an orthonormal basis of $\Hk$, and in this case the sum is independent of the chosen orthonormal basis. The set of Hilbert-Schmidt operators $\HSk$   endowed with the inner product 
\begin{equation*}
 \langle   C , T  \rangle_{\HSk}  = \sum_{i \geq 1}  \langle   C f_i , T  f_i \rangle_{\Hk}   = \sum_{i,j \geq 1}   \langle   C f_i ,  f_j \rangle_{\Hk} \langle   T f_i ,  f_j \rangle_{\Hk} \end{equation*}
 is also an Hilbert space. 
 
A linear operator $C$ on  $\Hk$ is called self-adjoint if $\langle C g , h \rangle_{\Hk} = \langle  g , C h \rangle_{\Hk} $ for any $(g,h) \in \Hk^2$. A linear operator $C$ on  $\Hk$ is called class-trace if $\sum_{i \geq 1}  | \langle   C f_i ,   f_i \rangle_{\Hk} | < \infty $. In this case the sum is independent of the chosen orthonormal basis and the trace of $C$ is defined as 
\begin{equation*}
\trhs (C) = \sum_{i \geq 1}    \langle   C f_i ,   f_i \rangle_{\Hk}  .
\end{equation*}

A finite rank operator is trace-class, and a trace-class operator is Hilbert-Schmidt. 
  
For $(f,g) \in \Hk \setminus {0}$, we consider the rank one operator $f \otimes g$ defined by 
\begin{equation*}
 h \in \Hk \mapsto  f \otimes g  (h) =  \langle  g , h \rangle_{\Hk}  f  \in \Hk , 
\end{equation*}
and we have $\| f \otimes f  \|_{\HSk} = \|f \|_{\Hk}^2  $ and $ \| f \otimes g  \|_{\HSk} \leq \|f \|_{\Hk} \|g \|_{\Hk} .$

\subsection{Norm and operators on $\Hk^k$}
\label{subs:opHkn}

Let $k \in \mathbb N ^*$. We first introduce a scalar product on $\Hk^{k}$ in a natural way by
\begin{equation*}
\langle \mathbf{g},\mathbf{h} \rangle_{\Hk^k} := \sum_{j=1}^k \langle g_j,h_j \rangle_{\Hk}  
\end{equation*}
for $\mathbf{g}  = (g_1, \dots, g_k)  \in \Hk^{k}$ and $\mathbf{h}  = (h_1, \dots, h_k)  \in \Hk^{k}$. We denote by
$ \|  \:  \|_{\Hk^k}$ the associated norm on $\Hk^{k}$.

We are now in position to define rigorously all the operators that appears in the HL statistic  $\statmanova$. A matrix $\dX \in \mtrx{\nobs,k}$ is also seen as a linear application from $\Hk^{k}$ to $\Hk^{\nobs}$: for $\mathbf{g} \in \Hk^{k}$, $\dX \mathbf{g} \in \Hk^n$ and
 \begin{equation}
\label{eq:MatrixHk}
(\dX \mathbf{g})_i  = \sum_{j=1}^k  x_{i,j} g_j   \quad  i = 1, \dots n .
 \end{equation} 
We next define rigorously the adjoint operator that appears in the HL statistic. For $\mathbf{g} \in \Hk^k$, we consider the associate application between the two Hilbert Spaces $\HSk$ and $\Hk^k$ defined by
\begin{equation}
\label{eq:gA}
 \mathbf{g}:  \mathbf{A}\in \HSk \mapsto  \mathbf{g}\mathbf{A} = \left(\mathbf{A} g_1, \dots \mathbf{A} g_k \right)' \in \Hk^k .
 \end{equation}
The application $u : (\mathbf{A},\mathbf{h}) \in  \HSk \times \Hk^k \mapsto  \langle  \mathbf{g}\mathbf{A} , \mathbf{h} \rangle_{\Hk^k} $ is clearly a sesquilinear form and it is also bounded because
\begin{eqnarray*}
\label{eq:product_hn_hs}
|u(\mathbf{A},\mathbf{h}) | & \leq & \|\mathbf{g}\mathbf{A} \|_{\Hk^k}   \|\mathbf{h} \|_{\Hk^k} \\
& \leq &    \sum_{j=1}^k \| \mathbf{A} g_j \|_{\Hk}      \|\mathbf{h} \|_{\Hk^k}  \\
& \leq &     \sum_{j=1}^k \| \mathbf{A} \|   \| g_j \|_{\Hk}   \|\mathbf{h} \|_{\Hk^k} \\
& \leq &  \left(   \sum_{j=1}^k  \| g_j \|_{\Hk}  \right) \| \mathbf{A} \|_{\HSk } \|\mathbf{h} \|_{\Hk^k} . 
\end{eqnarray*}
We can therefore consider the adjoint operator of this application (see for instance Theorem 2.2 in \cite{conway_course_1997}), which is denoted by $\mathbf{g^\star} : \Hk^k \mapsto \HSk $ and we check that $\mathbf{g^\star}$ is actually the operator:
 \begin{equation}
\label{eq:adjoint}
  \mathbf{g^\star} : \mathbf{h}\in \Hk^k \mapsto  \mathbf{g}^\star \mathbf{h} = \sum_{i = 1}^{k}h_i \otimes g_i \in\HSk  .
 \end{equation}  
Indeed, for any  $(  \mathbf{h}, \mathbf{A}) \in   \Hk^k \times \HSk$, on the one hand, by definition of $\mathbf{g}\mathbf{A}$ we have
\begin{eqnarray*}
\langle  \mathbf{g}\mathbf{A} , \mathbf{h} \rangle_{\Hk^k} &=&  \sum_{i=1}^k \langle   \mathbf{A} g_i, h_i   \rangle_{\Hk} \\
&=&  \sum_{i=1}^k  \sum_{j,\ell} \langle g_i  , f_j   \rangle_{\Hk}    \langle  h_i, f_\ell  \rangle_{\Hk}  \langle   \mathbf{A} f_j, f_\ell   \rangle_{\Hk},
\end{eqnarray*}
where $(f_j)_{j \geq 1}$ is an orthonormal basis of $\Hk$ so $g_i = \sum_{j \geq 1}  \langle  g_i, f_j   \rangle_{\Hk} f_j $ and $h_i = \sum_{\ell \geq 1}  \langle  h_i, f_\ell  \rangle_{\Hk} f_\ell $. On the other hand, 
\begin{eqnarray*}
\langle  \mathbf{A} ,  \sum_{i = 1}^{k}h_i \otimes g_i \rangle_{\HSk}  &=&  \sum_{j, \ell}  \langle   \mathbf{A} f_j, f_\ell  \rangle_{\Hk}  \left\langle \sum_{i = 1}^{k}h_i \otimes g_i  f_j, f_\ell  \right \rangle_{\Hk}   \\
&=& \sum_{j, \ell} \sum_{i=1}^k  \langle   \mathbf{A} f_j, f_\ell  \rangle_{\Hk}  
\langle  g_i , f_j \rangle_{\Hk} 
\langle  h_i  , f_\ell   \rangle   _{\Hk}.
\end{eqnarray*}

With above definition, for $\mathbf{g}$  and $\mathbf{M}=\big( m_{i,j}\big)_{i,j\in\{ 1, \dots, k \}} \in\mtrx{k}$, we then we have that:
\begin{equation}
\label{eq:formequadra}
\mathbf{g}^\star \mathbf{M} \mathbf{g} = \sum_{i,j= 1}^{k}m_{i,j} g_i \otimes g_j \in \HSk.
 \end{equation} 
For  $\mathbf{g}, \mathbf{h}  \in\Hk^k$  we also consider the operator $ \mathbf{g}  \mathbf{h} ^\star :  \Hk^k \mapsto \Hk^k  $
defined by $ \mathbf{g}  \mathbf{h} ^\star  \mathbf{f} : = \mathbf{g} (  \mathbf{h} ^\star  \mathbf{f}) $ for $\mathbf{f}  \in \Hk^k$, according to the applications defined in \eqref{eq:gA} and \eqref{eq:adjoint}. It can be checked that this operator exactly corresponds to the linear operator defined by the matrix (as in \eqref{eq:MatrixHk}):
\begin{equation}
\label{eq:gstarg}
\mathbf{g}  \mathbf{h} ^\star := \big(\left \langle g_i , h_j \right \rangle_{\Hk}\big)_{i,j\in\{ 1, \dots, k \}}\in\mtrx{k}. 
\end{equation}

For $v=(v_1,\dots,v_k)  \in \R^k$ and $\mathbf{g} \in \Hk^k$, we will abuse the notation of matrix multiplication by defining
$v' \mathbf{g}$
for the linear combination of the elements of $\mathbf{g}$ with weights $v$: 
\begin{equation*}
\begin{split}
v' \mathbf{g}  = \sum_{i = 1}^{k} v_i g_i \in \Hk.
\end{split}
\end{equation*}
For any $v, w \in\R^k$, we then have that
\begin{equation*}
\left \langle v' \mathbf{g}  , w' \mathbf{h} \right \rangle_{\Hk} = v^\prime (\mathbf{g} \mathbf{h}^\star) w. 
\end{equation*}

For $\mathbf{g}\in\Hk^k$, we finally define the linear application   
\begin{equation}
\label{HkcontreH}
  h  \in \Hk \mapsto \mathbf{g}h = \big(\left \langle g_1 , h \right \rangle_{\Hk},\dots,\left \langle g_k , h \right \rangle_{\Hk} \big)^\prime \in\R^k  .
 \end{equation}

\section{Linear model in the feature space $\Hk$}
\label{App:KernelLinearmodel}

\subsection{Least squares estimation on the linear model in the feature space}
\label{append:ls}

We check that the estimator given in \eqref{eq: empirical model parameters}  is the least squares estimator of $\params$ in the model~\eqref{eq:linear model on embeddings} for this norm on $\Hk^{\nobs}$. The Hilbert space $\Hk$ being separable,  let us consider $(\Hb_s)_{s\geq 1}$  an orthonormal basis of $\Hk$. For $s \geq 1$  and $h \in \Hk$, we define $h^s = \left \langle h , \Hb_s \right \rangle_{\Hk}$ such that $h = \sum_{s \geq 1} h^s \Hb_s$. From the definition of the norm $ \| \,  \|_{\Hk^{\nobs}}$, we easily find that the least squares problem in $\Hk^{\nobs}$  reduces to a collection of independent and classical  univariate least squares problems: for $\params = (\paramsuni_0, \dots,\paramsuni_{\nx-1})' \in \Hk^{\nx} $ where $\paramsuni_j = \sum_{s \geq 1} \param_j^s \Hb_s $, then
\begin{eqnarray*}
\| \eY  -  \dX \params\|^2_{\Hk^{\nobs}} &=& \sum_{i=1 \dots \nobs} \left\| \ey{j}  -  \dx_i'\params   \right\|_{\Hk}^2  \\
&=& \sum_{i=1 \dots \nobs} \sum_{s \geq 1} \left| \ey{j}^s  -  \sum_{j=0 \dots \nx-1} x_{i,j} \param_j^s  \right|^2  \\
&=&  \sum_{s \geq 1} \left\{ \sum_{i=1 \dots \nobs} \left| \ey{j}^s  -  \sum_{j=0 \dots \nx-1} x_{i,j} \param_j^s  \right|^2   \right\} \\
&=& \sum_{s \geq 1}  \| \eY^s   -  \dX \params ^s \|_{\R^{\nobs}} ^2 .
\end{eqnarray*}
As $\dX$ is assumed to be full rank, the least squares estimator of $\params^s = (\param_0^s, \dots, \param_{\nx-1}^s )' \in \R^{\nx}$ is given by
$
\hparams^s = \XXX \eY^s
$.
From this, we deduce that for the least square estimation of the  linear model \eqref{eq:linear model on embeddings}, we have
\begin{equation*}
 \hparams= \XXX \eY .
 \end{equation*}
Moreover, the predictions in the linear model \eqref{eq:linear model on embeddings} are defined as $\widehat{ \eY } =  \dX \XXX \eY  = \Px  \eY$ where $\Px = (\Pxij{i}{j})_{1\leq i,j \leq n}$  is the (matrix of the) orthogonal projection on $Im(\dX)$ in $(\R^n, \| \, \|_2 )$. We also define the residuals as 
 $\residuals =(\residual_1,\dots,\residual_\nobs)$ satisfy   $\residuals =\Pxp \eY  = \Pxp \errors \in\Hk^\nobs$, where  $\Pxp = \identity_\nobs - \Px = (\Pxpij{i}{j})_{1\leq i,j \leq n} $ is the orthogonal projection  on $Im(\dX)^\perp$.

According to the following lemma, we check that the matrices $\Px$ and $\Pxp$ can be seen as orthogonal matrices in $\Hk^n$. 
\begin{lemma} 
\label{lem:proj_est_une_proj}
Let $\boldsymbol{P}$ an orthogonal projection in $(\R^n, \langle,\rangle_2 )$. Then $\boldsymbol{P}$ is also an orthogonal projection in $(\Hk^n,\langle,\rangle_{\Hk^n})$.  In particular,  $\left \| \boldsymbol{P} f  \right \|_{\Hk^n} \leq \left \| f \right \| _{\Hk^n}$ for any $\mathbf{f} \in \Hk^n$.
\end{lemma}
\begin{proof}
We easily check that for any $\mathbf{f} \in \Hk^n$, $\boldsymbol{P} ( \boldsymbol{P}\mathbf{f}) =  (\boldsymbol{P}   \boldsymbol{P}) \mathbf{f} = \boldsymbol{P}\mathbf{f}$. Thus $\boldsymbol{P}$ is also a projection on $\Hk^n$. Next, it can be easily checked that for any $(v,w) \in (\R^n)^2$, and $(\mathbf{f},\mathbf{g}) \in (\Hk ^n)^2 $, 
\begin{equation}
\label{eq:ortho1}
 \langle ( v' v ) \mathbf{f} ,(w' w)  \mathbf{f} \rangle_{\Hk^n} = \langle   v'  \mathbf{f} , w'   \mathbf{f} \rangle_{\Hk } \langle   v     , w      \rangle_{\R^n}.
\end{equation}
and
\begin{equation}
\label{eq:ortho2}
\langle ( v' v ) \mathbf{f} ,   \mathbf{g} \rangle_{\Hk^n} = \langle  \mathbf{f} ,  ( v' v ) \mathbf{g} \rangle_{\Hk^n}.
\end{equation}
since $\boldsymbol{P}$ is an orthogonal projection in $(\R^n, \| \, \|_2 )$, we can write $\boldsymbol{P} = \sum_{i=1}^r v_i v_i' $ where $(v_i)_{i = \dots r}$ is an orthonormal basis of $Im ( \boldsymbol{P})$ in $\R^n$.  From \eqref{eq:ortho1} and \eqref{eq:ortho2}, we easily find that $
\langle  \boldsymbol{P} \mathbf{f} ,   \mathbf{g} \rangle_{\Hk^n} = \langle  \mathbf{f} ,  \boldsymbol{P} \mathbf{g} \rangle_{\Hk^n}.$
 
\end{proof} 
 
\subsection{The KHL and TKHL statistics are well defined}
\label{sec:KHKwelldefined}
The test operator associated with $\Lmanova$ is defined as 
\begin{equation*}
 \Hmanova = (\Lmanova\hparams)^{\star} (\Lmanova\XXinv\Lmanova^\prime)^{-1} (\Lmanova\hparams), 
 \end{equation*}

where $\Lmanova = (\lmanova_{\ell}^j)_{\ell \in \{1,\dots,\nLmanova\},j\in \{1,\dots,\nx\}} \in \R^{\nLmanova \times \nx}$ is a surjective matrix. 

Note that that $ \Hmanova$ is a well-defined  Hilbert-Schmidt operator. Indeed, $ \hparams\in \Hk^n$, and according to \eqref{eq:MatrixHk}, $\Lmanova\hparams \in \Hk^d$. Next, the matrix $  \Lmanova\XXinv\Lmanova^\prime$ is an $d\times d$ invertible matrix because $\Lmanova$ is surjective. According to \eqref{eq:formequadra}, $\Hmanova \in \HSk$ and can be written as
\begin{equation*}
\Hmanova = \sum_{i,j= 1}^{d} c_{i,j} (\Lmanova\hparams)_i \otimes (\Lmanova\hparams)_j  ,
\end{equation*}
where $C =(\Lmanova\XXinv\Lmanova^\prime)^{-1}. $  It is also a finite rank and class-trace operator.

The residual covariance operator  $\rcov = \frac{1}{\nobs} \sum_{i = 1}^{\nobs} \residual_i \otimes \residual_i  $ is clearly a bounded and finite rank operator, it is thus compact. It is also self-adjoint and thus according to the spectral theorem there exists an orthonormal basis of $\Hk$ consisting of eigenvectors of $\rcov$. It can then we written as 
\begin{equation*}
\rcov =  \sum_{t = 1}^{r} \rcoval_t  \big( \rcovec_t \otimes \rcovec_t \big)  
\end{equation*}
with $\rcoval_t \neq 0$ and where $r$ is the rank of $\rcov$. Moreover $\rcov$ is obviously  positive semi-definite positive, thus $\rcoval_t >0$ for any non null eigenvalue, and we will assume that they are indexed in decreasing order.
We then introduce a generalized inverse for $\rcov$ as
\begin{equation*}
 \rcov^{-} =   \sum_{t = 1}^{r}\frac{1}{\rcoval_t} \big( \rcovec_t \otimes \rcovec_t \big) ,
\end{equation*}
which is also a  positive semi-definite bounded and finite rank self-adjoint operator on $\Hk$. 

The operator $\rcov^{-} \Hmanova$ is thus trace-class as $\rcov^{-}$ is bounded and $\Hmanova$ is trace-class. We can thus define the kernalized Hotelling-Lawley statistic as 
\begin{equation*}
\begin{split}
\widehat{\mathcal{F}} =  \trhs(\nobs ^{-1} \rcov^{-} \Hmanova) ,
\end{split}
\end{equation*}
where $\trhs$ is the trace operator on $\Hk$. The spectral regularization of parameter $\tmax \leq r$ is defined as:
\begin{equation}
\label{eq: Spectral decomposition}
 \begin{split}
\rcovt^{-1} =& \sum_{t = 1}^{\tmax}\frac{1}{\rcoval_t} \big( \rcovec_t \otimes \rcovec_t \big) .
 \end{split}
 \end{equation} 
It is also a  positive semi-definite bounded and finite rank self-adjoint operator on $\Hk$. The Truncated kernel Hotelling-Lawley (TKHL) trace statistic can thus be finally  defined as
\begin{equation*}
\begin{split}
\statmanova = \trhs\left(\frac{1}{\nobs} \THLcov\right).
\end{split}
\end{equation*}

\section{Proof of Theorem~\ref{theorem: Asymptotic distribution of the truncated Hotelling-Lawley test statistic}}

\label{App:Proof-Theo-chi2}

This section is dedicated to the proof of Theorem~\ref{theorem: Asymptotic distribution of the truncated Hotelling-Lawley test statistic}. The first subsection gives a summary of the proof. We then provide the main tools for the proof in Propositions~\ref{proposition: Distrib asymptotique de la stat alternative}, \ref{proposition: Inegalite de concentration sur sigma epsilon emprique et theorique} and  \ref{lemma: Borne exponentielle sur les inverses des projecteurs}.

In this section, we consider the linear model \eqref{eq:linear model on embeddings} in the feature space $\Hk$.

\subsection{Main proof of Theorem~\ref{theorem: Asymptotic distribution of the truncated Hotelling-Lawley test statistic}}
\label{App:Proof-Theo-chi2:main}

First, we define an alternative truncated Hotelling-Lawley trace based on the true covariance structure shared between the errors $\ecovt$ instead of $\rcovt$, such that $\tildemanova = \trhs(\ecovt^{-1}\Hmanova)$, where $\ecovt^{-1}=\sum_{t = 1}^{\tmax}\ecoval_t^{-1} \big(\ecovec_t \otimes \ecovec_t\big)$ is well defined according to Assumption~\ref{Asimple}.  We have : 
\begin{eqnarray*}
|\statmanova - \tildemanova |  
& =& \left|\trhs\left(\frac{1}{\nobs}\big(\rcovt^{-1} - \ecovt^{-1}\big)\Hmanova\right)\right| \\
&=& \left|\left \langle \big(\rcovt^{-1} - \ecovt^{-1}\big) , \frac{1}{\nobs}\Hmanova \right \rangle_{\HSk}\right| \\
& \leq& \left \| \rcovt^{-1} - \ecovt^{-1} \right \| _{\HSk} \left \|\frac{1}{\nobs} \Hmanova \right \| _{\HSk}.
\end{eqnarray*}


Let $\Pxp = \identity_\nobs - \Px = (\Pxpij{i}{j})_{1\leq i,j \leq n} $ the orthogonal projection on $Im(\dX)^\perp$. We  define $\ecovbar = \ecovt + \frac{1}{\nobs} \sum_{i = 1}^{\nobs}\Pxpij{i}{i} \ecovtci{i}$.
By applying a bounded differences Theorem, under Assumptions \ref{Ahomoscedasticity}, \ref{Aboundedkernel} and \ref{Aboundeddesign}, Proposition \ref{proposition: Inegalite de concentration sur sigma epsilon emprique et theorique}  provides an exponential bound on to $\left \| \rcov - \ecovbar \right \| _{\HSk}$. Following the approach of \cite{zwald_convergence_2005}, we can then we make use of results from operator perturbation theory  to derive the exponential bound on $\left \| \rcovt^{-1} - \ecovt^{-1} \right \| _{\HSk}$ given in Proposition~\ref{lemma: Borne exponentielle sur les inverses des projecteurs}.

We then show in Lemma \ref{lemma: Borne sur H} that $\nobs ^{-1}\Hmanova$ is bounded under Assumptions \ref{Ahomoscedasticity} and \ref{Aboundedkernel}, and we can conclude that $\statmanova \overset{P}{\longrightarrow}\tildemanova$ as $n$ tends to infinity. It only remains to derive the asymptotic distribution of $\tildemanova$.

Proposition \ref{proposition: Distrib asymptotique de la stat alternative} shows that $\tildemanova$ corresponds to a standard Hotelling-Lawley Statistic in $\R^\tmax$. We then invoke a result from the literature (see for instance Theorem 12.8 from \cite{olive_robust_2017}) to find that if $H_0$ is true, then $
\nobs \tildemanova \underset{}{\overset{\mathcal{D}}{\longrightarrow}} \chi^2_{\nLmanova \tmax}
$ as $\nobs$ tends to infinity.

\subsection{Statistics for the test operator}

\begin{proposition}
\label{proposition: Distrib asymptotique de la stat alternative}
Under Assumptions \ref{Ahomoscedasticity},  \ref{Aboundedkernel}, \ref{Aboundeddesign}, \ref{ADiagDesign} and  \ref{Aconvergentdesign}, if $H_0$ is true then we have:
\begin{equation*}
\begin{split}
\nobs \tildemanova =  \trhs(\ecovt^{-1}\Hmanova) \underset{\nobs\to \infty}{\overset{\mathcal{D}}{\longrightarrow}} \chi^2_{\nLmanova \tmax}.
\end{split}
\end{equation*}
\end{proposition}

\begin{proof}
We recall that the $i^{th}$ equation of the linear model \eqref{eq:linear model on embeddings}  can be written as
\begin{equation*}
\begin{split}
\ey{i} = \dx_i'\params  + \error_i,
\end{split}
\end{equation*}
where $\E\error_i = 0$. According to Assumption \ref{Ahomoscedasticity}, we have $\cov(\error_i) = \ecov_i = \ecovt + \ecovtci{i}$, and $\ecovt = \sum_{t = 1}^{\tmax}\ecoval_t \ecovec_t \otimes \ecovec_t$ where $\ecovec_1,\dots,\ecovec_\tmax$ is an orthonormal set of eigenfunctions of $\ecovt$ associated with the eigenvalues $\ecoval_1,\dots,\ecoval_\tmax$. For $t\in \{ 1, \dots, \tmax \}$ and $h \in \Hk$, we define $h^t= \left \langle h , \ecovec_t \right \rangle_{\Hk}$.
For $i \in \{ 1, \dots, \nobs \}$, let $Z_i = (\ey{i}^1,\dots,\ey{i}^\tmax)$ in $\R^\tmax$, $\linparam_i = (\param_i^1,\dots,\param_i^\tmax)$ in $\R^\tmax$ and $\linparams = (\linparam_0 ,\dots,\linparam_{\nx-1})^\prime \in \mtrx{\nx,\tmax}$. We also define $\linerror_i = (\error_i^1,\dots,\error_i^\tmax)$ in $\R^\tmax$. The projection of the $i^{th}$ equation of the linear model \eqref{eq:linear model on embeddings}  on $\operatorname{Span}(\ecovec_1,\dots,\ecovec_\tmax) \subset \Hk$  is:
\begin{equation*}
\begin{split}
Z_i =  \dx_i'\linparams  + \linerror_i. 
\end{split}
\end{equation*}
We recognize a multivariate linear model in $\R^\tmax$ that has the matrix form:
\begin{equation}
\label{eq:ModelLinearT}
\begin{split}
\mathbf{Z} = X \linparams + \linerrors,
\end{split}
\end{equation}
where $\mathbf{Z} = (Z_1,\dots,Z_n)^\prime \in \mtrx{\nobs,\tmax}$ and $\linerrors = (\linerror_1 ,\dots,\linerror_\nobs)^\prime \in \mtrx{\nobs,\tmax}$. Remark that the errors $\linerror_1 ,\dots,\linerror_\nobs$ are independant with $\E \linerror_1 = 0$ and same covariance matrix $\linecov = (\sigma_{t,t^\prime})_{t,t^\prime \in \{ 1, \dots, \tmax \}} \in \mtrx{\tmax}$. For $i\in\{ 1, \dots, \nobs \}$ and $t,t^\prime \in \{ 1, \dots, \tmax \}$, we have:
 \begin{equation*}
 \begin{split}
 \sigma_{t,t^\prime} 
 =& \E (\error_i^t \error_i^{t^\prime} ) 
 \\ =& \E \left \langle \error_i , \ecovec_t \right \rangle_{\Hk}\left \langle \error_i , \ecovec_{t^\prime} \right \rangle_{\Hk}
 \\ =& \left \langle \ecovec_t , \E (\error_i\otimes \error_i)  \ecovec_{t^\prime}\right \rangle_{\Hk}
 \\ =& \left \langle \ecovec_t , \ecov_i \ecovec_{t^\prime} \right \rangle_{\Hk}.
\end{split}
\end{equation*}
By construction, as $t^\prime\in\{ 1, \dots, \tmax \}$, $\ecovec_{t^\prime}$ is an eigenvector of $\ecov_i$ associated with the eigenvalue $\ecoval_{t^\prime}$, thus according to Assumption~\ref{Ahomoscedasticity},
\begin{equation*}
\begin{split}
  \sigma_{t,t^\prime}
  =& \ecoval_{t^\prime}\left \langle \ecovec_t , \ecovec_{t^\prime} \right \rangle_{\Hk}
\\  
 =& \left\{
    \begin{array}{ll}
        \ecoval_t & \mbox{if } t=t^\prime \\
       0 & \mbox{otherwise.}
    \end{array}
\right.
\end{split}
\end{equation*}
Thus, we have $\linecov =\operatorname{diag}(\ecoval_1,\dots,\ecoval_\tmax)$ where $\ecoval_j >0$, $j=1,\dots,\tmax$.

In the multivariate linear model \eqref{eq:ModelLinearT}, the least squares estimator $\hlinparams = (\hlinparam_0,\dots,\hlinparam_{\nx-1})^\prime \in \mtrx{\nx,\tmax}$ of $\linparams$ is defined by 
\begin{equation*}
  \begin{split}
  \hlinparam = \XXX \mathbf{Z}.
  \end{split}
  \end{equation*}  
By developing $\mathbf{Z}$ in this equation, it can be easily checked that that for $j \in \{ 1, \dots, \nx \}$, $\hlinparam_j = (\hparam_j^1,\dots,\hparam_j^\tmax)^\prime \in \R^\tmax$.  
Moreover, the covariance matrix $\linecov$, which is by definition diagonal can be estimated by  $\linrcov =  \operatorname{diag}(\hat \ecoval_1,\dots,\hat \ecoval_\tmax)$ where  $\hat \ecoval_j$ is the empirical variance of the $Z_{i,j}$'s.

\medskip 

In this multivariate linear model~\eqref{eq:ModelLinearT}, the Hotelling-Lawley trace associated with the hypotheses $\tilde H_0: \Lmanova \linparams = 0_{\R ^\nLmanova}$ versus $\tilde H_1: \Lmanova \linparams \neq 0_{\R^\nLmanova}$ is defined by
\begin{equation*}
\begin{split}
\mathcal{G}_n = \trmtrx\left(\frac{1}{\nobs}\linrcov^{-1} \linHmanova\right),
\end{split}
\end{equation*}
where  $\linHmanova = (\Lmanova\hlinparams)^\prime ( \Lmanova \XXinv \Lmanova^\prime )^{-1} (\Lmanova \hlinparams)\in\mtrx{\tmax}$. Note that  in $\mathcal{G}_n$ we use the  trace operator  $\trmtrx$ defined on $\mtrx{\tmax}$, to distinguish it from the  trace $\trhs$ defined on the HS operators on $\Hk$. 

We now check that the assumptions of Theorem 12.8 from \cite{olive_robust_2017} are satisfied (in particular Assumption D1 in this theorem). First, Assumption \ref{Ahomoscedasticity} provides the homoscedasticity assumption in $\R^\tmax$. Next, Assumptions \ref{Aboundedkernel} and \ref{Aboundeddesign} are sufficient to guarantee that the errors are bounded (see Lemma~\ref{lemma: Borne sur la norme des residus}) and thus the $\varepsilon_i$ have obviously a finite fourth moment. Finally, Assumption \ref{ADiagDesign} provides the control of the coefficient on the diagonal of $\Px$, and the convergence of $\frac 1 n  X'X $ is provided by Assumption~\ref{Aconvergentdesign}. All this gives that if $H_0$ is true, then
\begin{equation*}
\begin{split}
\nobs \mathcal{G}_n \underset{\nobs\to \infty}{\overset{\mathcal{D}}{\longrightarrow}} \chi^2_{\nLmanova \tmax}. 
\end{split}
\end{equation*}
In particular, in the proof of the above result, it is also shown that
\begin{equation*}
\begin{split}
\nobs \statmanovakt_n = \trmtrx(\linecov^{-1} \linHmanova) \underset{\nobs\to \infty}{\overset{\mathcal{D}}{\longrightarrow}} \chi^2_{\nLmanova \tmax}.
\end{split}
\end{equation*}
Thus, it remains to check that $\tildemanova = \statmanovakt_n$ to complete the proof.
 
\medskip

We denote $C =   (\Lmanova \XXinv \Lmanova^\prime) ^{-1}   = (c_{i,j})_{i,j \in \{ 1, \dots, \nx \}} \in \mtrx{\nx}$, and then $\linHmanova = (L \hlinparams) ^\prime C (L \hlinparams )
= \sum_{i,j = 1}^{\nx}c_{i,j} (L\hlinparam)_i (L \hlinparam)_j^\prime$. For $i,j \in \{ 0, \dots, \nx-1 \}$, we have:
\begin{equation*}
\begin{split}
\linecov^{-1} (L\hlinparam)_i (L \hlinparam)_j^\prime
 =& \left (\begin{matrix} 
\ecoval_1 ^{-1} (L\hlinparam)_i^1 (L \hlinparam)_j^1& \dots & \ecoval_1^{-1} (L\hlinparam)_i^1 (L \hlinparam)_j^\tmax 
\\ \vdots  & \ddots & \vdots 
\\ \ecoval_\tmax^{-1} (L\hlinparam)_i^\tmax (L \hlinparam)_j^1& \dots & \ecoval_\tmax^{-1} (L\hlinparam)_i^\tmax (L \hlinparam)_j^\tmax 
 \end{matrix} \right)
\end{split}.
\end{equation*}
Then $\nobs\statmanovakt_n = \trmtrx\left(\linecov ^{-1}\linHmanova\right) = \sum_{i,j = 0}^{\nx-1} c_{i,j} \sum_{t = 1}^{\tmax} \ecoval_t^{-1} \hparam_i^t \hparam_j^t$. On the other hand, we have that $ \nobs \tildemanova 
=  \trhs\left(\ecovt^{-1} \Hmanova\right) =  \left \langle  \Hmanova, \ecovt^{-1} \right \rangle_{\HSk}  $ because $\ecovt^{-1} = \sum_{t = 1}^{\tmax} \rcoval_t ^{-1}  \rcovec_t \otimes \rcovec_t  $ is self-adjoint and finite rank.  Thus,
\begin{equation*}
\begin{split}
\nobs \tildemanova 
= & \sum_{i,j = 0}^{\nx-1}\sum_{t = 1}^{\tmax} c_{i,j} \ecoval_t^{-1} \left \langle \ecovec_t \otimes \ecovec_t , (\Lmanova\hparams)_i \otimes (\Lmanova\hparams)_{j} \right \rangle_{\HSk}
\\=& \sum_{i,j = 0}^{\nx-1}\sum_{t = 1}^{\tmax} c_{i,j}\ecoval_t^{-1} \left \langle \ecovec_t,(\Lmanova\hparams)_i \right \rangle_{\Hk}\left \langle \ecovec_t,(\Lmanova\hparams)_{j} \right \rangle_{\Hk}
\\=& \sum_{i,j = 0}^{\nx-1}\sum_{t = 1}^{\tmax} c_{i,j} \ecoval_t^{-1} (\Lmanova\hparams)_i^t (\Lmanova\hparams)_j^t 
\\ =& \nobs \statmanovakt_n.
\end{split}
\end{equation*}
\end{proof}

\begin{lemma}
\label{lemma: Borne sur H}
Under Assumptions \ref{Aboundedkernel} and \ref{Aboundeddesign}, if $H_0$ is true, then we have:
\begin{equation*}
\begin{split}
\left \| \frac{1}{\nobs}  \Hmanova \right \| _{\HSk} \leq \boundepsilon^2,
\end{split}
\end{equation*}
where where $\boundepsilon = \sqrt{\boundk} + \nx \boundX \Max{j \in \{ 1, \dots, \nx \}}\left \| \param_j \right \| _{\Hk}$.
\end{lemma}

\begin{proof}
We inject the expression of $\hparams$ in the expression of $\Hmanova$:
\begin{equation*}
\begin{split}
\Hmanova = (\Lmanova \XXX \eY)^\star (\Lmanova \XXinv \Lmanova^\prime )^{-1} (\Lmanova\XXX \eY)
\end{split}
\end{equation*}
According to the  linear model~\eqref{eq:linear model on embeddings}, we have 
\begin{equation*}
\begin{split}
 \Lmanova \XXX \eY  =  \Lmanova \params + \Lmanova\XXX \errors .
\end{split}
\end{equation*}
The term $\Lmanova \params$ is null under $H_0$. Then we can write:
\begin{equation*}
\begin{split}
\Hmanova = \errors^\star
\mathbf{D} \errors
\end{split}
\end{equation*}
where $\mathbf{D} = \XXXprime \Lmanova^\prime (\Lmanova \XXinv \Lmanova^\prime )^{-1} \Lmanova \XXX =(d_{i,j})_{i,j\in \{1,\dots,\nobs \} } \in \mtrx{\nobs}$ is an orthogonal projector (of $\R^n$), as $\mathbf{D}^{\prime}=\mathbf{D}$ and $\mathbf{D}^2=\mathbf{D}$. 
Thus we can write $\Hmanova = \errors^\star
\mathbf{D}' \mathbf{D}  \errors$. Moreover, it can be easily checked that $ (D E) ^\star =  E  ^\star  D'$. Thus we have: 
\begin{eqnarray*}
\left \| \Hmanova \right \| _{\HSk}
&  =& \left \| (\mathbf{D} \errors)^\star (\mathbf{D} \errors) \right \| _{\HSk} \\
&  \leq  & \sum_{i = 1}^{\nobs} \left \| (\mathbf{D} \errors)_i \otimes (\mathbf{D} \errors)_i \right \| _{\HSk} \\ 
& \leq & \sum_{i = 1}^{\nobs} \left \| (\mathbf{D} \errors)_i \right \| _{\Hk}^2=  \left \| \mathbf{D} \errors \right \| _{\Hk^\nobs}^2.
\end{eqnarray*}
By Lemma \ref{lem:proj_est_une_proj}, we have $\left \| \mathbf{D} \errors \right \| _{\Hk^\nobs}\leq\left \|  \errors \right \| _{\Hk^\nobs}$, thus, according to Lemma \ref{lemma: Borne sur la norme des residus}, we have : 
\begin{equation*}
\begin{split}
\left \| \Hmanova \right \| _{\HSk}
\leq& \left \| \errors \right \| _{\Hk^\nobs}^2 
\\ \leq& \sum_{i = 1}^{\nobs}\left \| \error_i \right \| _{\Hk}^2
\\ \leq& \nobs \boundepsilon^2 .
\end{split}
\end{equation*}
Thus we have $\left \| \nobs^{-1} \Hmanova \right \| _{\HSk}\leq \boundepsilon^2$. 
\end{proof}

\subsection{Probability bound on the residual covariance}

We define $\ecovbar = \ecovt + \frac{1}{\nobs} \sum_{i = 1}^{\nobs}\Pxpij{i}{i} \ecovtci{i}$. Next Proposition is adapted from  \cite{shawe-taylor_eigenspectrum_2005} and \cite{zwald_convergence_2005}. An obvious  consequence of Proposition \ref{proposition: Inegalite de concentration sur sigma epsilon emprique et theorique} is that $ \rcov  - \ecovbar$ tends to $0$ in probability as $n$ tends to infinity.

\begin{proposition}
\label{proposition: Inegalite de concentration sur sigma epsilon emprique et theorique}
If Assumptions  \ref{Ahomoscedasticity},  \ref{Aboundedkernel} and \ref{Aboundeddesign} are satisfied, then we have with probability $1-e^{-\xi}$: 
\begin{equation}
\label{eq:inegalite de concentration sur sigma epsilon empirique et theorique}
\begin{split}
 \left \| \rcov - \ecovbar \right \| _{\HSk} 
 \leq \frac{\nx }{\nobs} \left \| \ecovt \right \| _{\HSk}
 +   \frac{\boundepsilon^2}{\sqrt \nobs} \left(1 +  4   \sqrt{\frac{3 \xi}{2} } \right) =: \boundrcov,
\end{split}
\end{equation}
where $\boundepsilon = \boundk^{\frac{1}{2}}+ \boundX \Max{j \in \{ 1, \dots, \nx \}}\left \| \param_j \right \| _{\Hk}$.  
\end{proposition}

\begin{proof}
We observe that:
\begin{equation*}
\begin{split}
\left \| \rcov - \ecovbar \right \| _{\HSk} =& \left \| \rcov - \E \rcov + \E \rcov - \ecovbar  \right \| _{\HSk}
\\ \leq&   \left \| \rcov - \E \rcov \right \| _{\HSk} + \left \|\E \rcov - \ecovbar  \right \| _{\HSk}.
\end{split}
\end{equation*}

We know from Lemma \ref{lemma: biais de sigma epsilon} that:
\begin{equation*}
\begin{split}
\left \|\E \rcov - \ecovbar  \right \| _{\HSk} 
=& \left \| \frac{\nobs - \nx}{\nobs} \ecovt - \ecovt  \right \| _{\HSk}
\\=&  \frac{\nx }{\nobs} \left \|  \ecovt  \right \| _{\HSk}.
\end{split}
\end{equation*}

To apply the McDiarmid Inequality (see Theorem \ref{theorem: McDiarmid}), we introduce to the function:
\begin{eqnarray*}
&\Hk^\nobs &\longrightarrow \R
\\f:&\errors&\longmapsto \left \| \rcov - \E(\rcov)  \right \| _{\HSk}.
\end{eqnarray*}
Let $i_0\in\{ 1, \dots, \nobs \}$, $\errors = (\error_1,\dots, \error_\nobs)$ and $\tilde \errors_{i_0} = (\tilde \error_1,\dots, \tilde \error_\nobs)$ in $\Hk^\nobs$ such that $\error_{i_0} \neq \tilde \error_{i_0}$ and $\forall i \neq i_0, \error_i = \tilde \error_i$. Lemma \ref{lemma: Borne de la fluctuation de la norme de la diff entre sigma epsilon et son esperance} gives us that:
\begin{equation*}
| f(\errors) - f( \tilde \errors_{i_0})| \leq \frac{2 \boundepsilon}{\nobs} \left(\boundepsilon +  \left \| 
(\Px \errors)_{i_0}\right \| _{\Hk} + \left \| 
(\Px \widetilde \errors)_{i_0}\right \| _{\Hk}\right)  = :c_{i_0} 
\end{equation*}
and thus
\begin{eqnarray}
\sum_{i_0 =1} ^n c_{i_0}^2 & \leq &  \frac{16 \boundepsilon^4}{n}   +  \frac{16 \boundepsilon^2}{n^2}   \sum_{i_0}^n  \left \| 
(\Px \errors)_{i_0} \right \| _{\Hk}^2 +  \frac{16 \boundepsilon^2}{n^2}    \left \| 
(\Px \widetilde \errors)_{i_0} \right \| _{\Hk}^2   \\
& \leq  &  \frac{16 \boundepsilon^4}{n}   +  \frac{16 \boundepsilon^2}{n^2}     \left \| 
\Px \errors  \right \| _{\Hk_n}^2 +  \frac{16 \boundepsilon^2}{n^2}    \left \| \Px \widetilde \errors  \right \| _{\Hk_n}^2  
\end{eqnarray}
According to Lemmas~\ref{lem:proj_est_une_proj} and \ref{lemma: Borne sur la norme des residus}, $\| \Px \errors  \| _{\Hk_n} \leq \| \errors \| _{\Hk_n} \leq \sqrt{n} \boundepsilon $ and finally  $ \sum_{i_0 =1} ^n c_{i_0}^2 \leq \frac{48 \boundepsilon^4}{n}   $.

According to  McDiarmid Inequality  to $f$, we have with probability $1-e^{-\xi}$: 
\begin{equation*}
\left \| \rcov - \E \rcov \right \| _{\HSk} \leq \E \left \| \rcov - \E \rcov \right \| _{\HSk} + 4 \boundepsilon^2 \sqrt{\frac{3 \xi}{2n} }.
\end{equation*} 
By Lemma \ref{lemma: bound_on_expectation}, we bound the expectation term to obtain that with probability $1-e^{-\xi}$:
\begin{eqnarray*}
 \left \| \rcov - \E \rcov \right \| _{\HSk} 
 &\leq &      \frac{\boundepsilon^2\sqrt{\nobs - \rangx}}{\nobs} + 4 \boundepsilon^2 \sqrt{\frac{3 \xi}{2n} }  \\
  &\leq &      \frac{\boundepsilon^2}{\sqrt \nobs} \left(1 +  4   \sqrt{\frac{3 \xi}{2} } \right)  .
\end{eqnarray*}

\end{proof}


 \begin{lemma}
\label{lemma: biais de sigma epsilon}
Under Assumption \ref{Ahomoscedasticity}, we have that:
\begin{equation*}
\begin{split}
\E(\rcov) = \frac{\nobs - \rangx}{\nobs} \ecovt + \frac{1}{\nobs} \sum_{i = 1}^{\nobs} \Pxpij{i}{i} \ecovtci{i}.
\end{split}
\end{equation*}
\end{lemma}

\begin{proof}
The $i$-th residual can be written as 
$
\residual_i =   \sum_{j = 1}^{\nobs}\Pxpij{i}{j} \error_j.
$
where $\Pxpi{i} = (\Pxpij{i}{1},\dots , \Pxpij{i}{\nobs})^\prime \in \R^\nobs$ is the $i^{th}$ column of $\Pxp$. Since $\Pxp$ is an orthogonal projector, we have
\begin{equation}
\label{eq:identity_orthogonal_projector} 
\begin{split}
\forall j \in \{ 1, \dots, \nobs \},\quad \sum_{i = 1}^{\nobs} \Pxpij{i}{j}^2 = \Pxpij{j}{j}.
\end{split}
\end{equation}

Consequently, the expectation is such that: 
\begin{equation*}
\begin{split}
\E (\rcov) =&  \frac{1}{\nobs} \E (\sum_{i = 1}^{\nobs} \residual_i \otimes \residual_i)
\\ =& \frac{1}{\nobs} \sum_{i = 1}^{\nobs} \sum_{j = 1}^{\nobs} \sum_{k = 1}^{\nobs} \Pxpij{i}{j} \Pxpij{i}{k} \E (\error_j \otimes \error_k)
\\ =& \frac{1}{\nobs} \sum_{i = 1}^{\nobs} \left ( \sum_{j = 1}^{\nobs} \Pxpij{i}{j}^2 \E (\error_j \otimes \error_j) + \sum_{j \neq k } \Pxpij{i}{j} \Pxpij{i}{k} \E(\error_j) \otimes \E(\error_k) \right ).
\end{split}
\end{equation*}
Then, for $j\in\{ 1, \dots, \nobs \}, \E(\error_j) =0$. We use Assumption \ref{Ahomoscedasticity} to obtain :  
\begin{equation*}
\begin{split}
\E (\rcov) =& \frac{1}{\nobs}  \sum_{i,j = 1}^{\nobs} \Pxpij{i}{j}^2 (\ecovt + \ecovtci{j})
\\ =& \frac{\nobs - \rangx}{\nobs} \ecovt + \frac{1}{\nobs} \sum_{j = 1}^{\nobs} \Pxpij{j}{j} \ecovtci{j}
\end{split}
\end{equation*}
where we apply Lemma \ref{lemma: sommes de pi ij} and Equation \eqref{eq:identity_orthogonal_projector} to obtain the last equality.  
\end{proof}

\begin{lemma}
\label{lemma: Borne sur la norme des residus}
If Assumptions \ref{Aboundedkernel} and \ref{Aboundeddesign} are satisfied, then for $i \in \{1,\dots,\nobs\}$,
\begin{equation*}
\begin{split}
 \left \|\error_i  \right \| _{\mathcal{H}}<\boundepsilon := \sqrt{\boundk} + \nx \boundX \Max{j \in \{ 1, \dots, \nx \}}\left \| \param_j \right \| _{\Hk}.
\end{split}
\end{equation*}
\end{lemma}
\begin{proof}
We have:
\begin{equation*}
\begin{split}
\left \| \error_i  \right \| _{\Hk} 
=& \left \| \ey{i} - (\dX\params)_i  \right \| _{\Hk} 
\\ \leq& \left \| \ey{i}  \right \| _{\Hk} + \left \| (\dX \params)_i  \right \| _{\Hk} 
\end{split}
\end{equation*}
From Assumption \ref{Aboundedkernel}, we know that $\left \| \ey{i}  \right \| _{\Hk} \leq \sqrt{\boundk}$, and from the triangular inequality, we have:
\begin{eqnarray*}
\left \|(\dX \params)_i    \right \| _{\Hk} 
& = & \left \|   \sum_{j = 0}^{\nx-1}\param_j \dx_i^j \right \| _{\Hk} \\
&\leq & \nx \boundX \Max{j \in \{ 0, \dots, \nx-1 \}} \left \| \param_j \right \|_{\Hk} 
\end{eqnarray*}
and the sum of the two bounds gives the result. 
\end{proof}

\begin{lemma}
\label{lemma: Borne de la fluctuation de la norme de la diff entre sigma epsilon et son esperance}
 Let  $\errors$, $\tilde \errors_{i_0}$ and $f$  defined as in the proof of Proposition~\ref{proposition: Inegalite de concentration sur sigma epsilon emprique et theorique}.  Under Assumptions \ref{Aboundedkernel} and \ref{Aboundeddesign} we have: 
\begin{equation*}
| f(\errors) - f( \tilde \errors_{i_0})| \leq \frac{2 \boundepsilon}{\nobs} \left(\boundepsilon +  \left \| 
(\Px \errors)_{i_0}\right \| _{\Hk} + \left \| 
(\Px \widetilde \errors)_{i_0}\right \| _{\Hk}\right) ,
\end{equation*}
where $\boundepsilon$ is defined in Lemma \ref{lemma: Borne sur la norme des residus}.
\end{lemma}

\begin{proof}
Denote $\widetilde \ecov  = \frac{1}{\nobs} \sum_{i = 1}^{\nobs} \tilde \error_i \otimes \tilde \error_i$, we have:
\begin{equation*}
\begin{split}
| f(\errors) - f( \tilde \errors_{i_0})| 
=& \left |\left \| \rcov - \E(\rcov)  \right \| _{\HSk} - \left \| \widetilde \ecov - \E(\rcov)  \right \| _{\HSk} \right |
\\ \leq& \left \| \rcov -  \widetilde \ecov  \right \| _{\HSk}.
\end{split}
\end{equation*}
As $\residual_i = \error_i -    \Pxi{i}'\errors$ where $\Pxi{i}$ is the $i$-th column (or row) of $\Px$, we have:
\begin{equation*}
\begin{split}
\rcov =& 
\frac{1}{\nobs} \sum_{i = 1}^{\nobs} \error_i \otimes \error_i
- \frac{1}{\nobs} \sum_{i = 1}^{\nobs}  \error_i \otimes (\Pxi{i}'\errors )
- \frac{1}{\nobs} \sum_{i = 1}^{\nobs}  (\Pxi{i}'\errors) \otimes \error_i
+ \frac{1}{\nobs} \sum_{i = 1}^{\nobs}  (\Pxi{i}'\errors) \otimes (\Pxi{i}'\errors). 
\end{split}
\end{equation*}
As $\Px$ is an orthogonal projector, we have that $\Px = \Px^\prime = \Px^2$. Thus for $i,j \in \{ 1, \dots, \nobs \}$, we have $\Pxij{i}{j} = \Pxij{j}{i}$ and $\sum_{k = 1}^{\nobs}\Pxij{i}{k} \Pxij{j}{k} = \Pxij{i}{j}$. Then by developing each term $\Pxi{i}'\errors= \sum_{j = 1}^{\nobs}\Pxij{i}{j} \error_j$, we obtain that: 
\begin{equation*}
\begin{split}
\frac{1}{\nobs} \sum_{i = 1}^{\nobs}  \error_i \otimes (\Pxi{i}'\errors ) 
&= \frac{1}{\nobs} \sum_{i = 1}^{\nobs}  (\Pxi{i}'\errors) \otimes \error_i 
= \frac{1}{\nobs} \sum_{i = 1}^{\nobs}  (\Pxi{i}'\errors) \otimes (\Pxi{i}'\errors).
\end{split}
\end{equation*}
It leads to:
\begin{equation*}
\begin{split}
\rcov =&\frac{1}{\nobs} \sum_{i = 1}^{\nobs} \error_i \otimes \error_i
 - \frac{1}{\nobs} \sum_{i = 1}^{\nobs} \error_i  \otimes  (\Pxi{i}'\errors).
\end{split}
\end{equation*}
We replace $\rcov$ and $\widetilde \ecov$  by this expression and use the triangular inequality to obtain: 
\begin{equation*}
| f(\errors) - f( \tilde \errors_{i_0})| 
\leq   
\frac{1}{\nobs}  \left \| \sum_{i = 1}^{\nobs} \error_i \otimes \error_i - \tilde \error_i \otimes \tilde \error_i \right \| _{\HSk} 
 + \underset{=A}{\underbrace{\frac{1}{\nobs}  \left \| \sum_{i = 1}^{\nobs} \tilde \error_i \otimes  (\Pxi{i}' \widetilde \errors ) - \error_i \otimes  ( \Pxi{i}'\errors) \right \| _{\HSk}}}.
\end{equation*}
The first term is such that:
\begin{eqnarray}
\frac{1}{\nobs}  \left \| \sum_{i = 1}^{\nobs} \error_i \otimes \error_i - \tilde \error_i \otimes \tilde \error_i \right \| _{\HSk} 
&=& \frac{1}{\nobs}  \left \|  \error_{i_0} \otimes \error_{i_0} - \tilde \error_{i_0} \otimes \tilde \error_{i_0} \right \| _{\HSk} \notag \\ 
&\leq& \frac{2 \boundepsilon^2}{\nobs} \label{eq:Terme1LemmeDiffb} .
\end{eqnarray}
The second term can be decomposed as:
\begin{equation*}
\begin{split}
A =& 
\frac{1}{\nobs} \left \| 
\sum_{i = 1,i\neq i_{0}}^{\nobs}
  \error_i \otimes  \Big( (\Pxi{i}'( \widetilde \errors - \errors )  \Big) 
+ \tilde \error_{i_0} \otimes  (\Pxi{i_0}' \widetilde \errors ) 
-  \error_{i_0} \otimes ( \Pxi{i_0}' \widetilde \errors ) 
\right \| _{\HSk} 
\\ =& 
\frac{1}{\nobs} \left \| 
\sum_{i = 1}^{\nobs}
  \error_i \otimes  \left( (\Pxi{i}'( \widetilde \errors - \errors )  \right) 
+ ( \widetilde \error_{i_0} -  \error_{i_0} ) \otimes ( \Pxi{i_0}' \widetilde \errors ) 
\right \| _{\HSk} 
\\ \leq& 
\underset{A_1}{\underbrace{\frac{1}{\nobs} \left \| 
\sum_{i = 1}^{\nobs}  \error_i  \otimes 
\left (\Pxi{i}'( \widetilde \errors - \errors ) \right) 
\right \| _{\HSk} 
}}
+ 
\underset{A_2}{\underbrace{\frac{1}{\nobs} \left \| 
( \widetilde \error_{i_0} -  \error_{i_0} ) \otimes ( \Pxi{i_0}' \widetilde \errors ) 
\right \| _{\HSk} 
}}.
\end{split}
\end{equation*}
We first bound $A_1$. Note that par definition of $\widetilde \errors$, we have $\Pxi{i}'( \widetilde \errors - \errors ) = \pi_{i,i_0} (\widetilde \error_{i_0} -  \error_{i_0}) $ and thus 
$
\sum_{i = 1}^{\nobs}
  \error_i \otimes  \left( (\Pxi{i}'( \widetilde \errors - \errors )  \right)  = \sum_{i = 1}^{\nobs}
  \error_i \otimes  \left(\pi_{i,i_0} (\widetilde \error_{i_0} -  \error_{i_0})   \right)  = 
 \left( \sum_{i = 1}^{\nobs}
 \pi_{i,i_0}  \error_i \right) \otimes   (\widetilde \error_{i_0} -  \error_{i_0})   = (\Px \errors)_{i_0} \otimes   (\widetilde \error_{i_0} -  \error_{i_0})
$. From this we obtain that
\begin{eqnarray*}
A_1 &\leq &  \frac{1}{\nobs} \left \| 
(\Px \errors)_{i_0} \otimes   (\widetilde \error_{i_0} -  \error_{i_0})
\right \| _{\HSk}  \\
&\leq &    \frac{2 \boundepsilon}{\nobs} \left \| 
(\Px \errors)_{i_0}\right \| _{\Hk} 
\end{eqnarray*}
Regarding $A_2$, by noting that $\Pxi{i_0}' \widetilde \errors = (\Px \widetilde \errors)_{i_0}$, we   get that $A_2 \leq \frac{2 \boundepsilon}{\nobs} \left \| 
(\Px \widetilde \errors)_{i_0}\right \| _{\Hk} $ and thus $A \leq \frac{2 \boundepsilon}{\nobs} \left( \left \| 
(\Px \errors)_{i_0}\right \| _{\Hk} + \left \| 
(\Px \widetilde \errors)_{i_0}\right \| _{\Hk}\right) $. We conclude the proof by combining this bound with \eqref{eq:Terme1LemmeDiffb}.
\end{proof}

\begin{lemma}
\label{lemma: bound_on_expectation}
Under  Assumptions \ref{Aboundedkernel} and \ref{Aboundeddesign}, we have: 
\begin{equation*}
\begin{split}
\E \left \| \rcov - \E \rcov \right \| _{\HSk} 
\leq &  \frac{\boundepsilon^2\sqrt{\nobs - \rangx}}{\nobs},
\end{split}
\end{equation*}
where $\boundepsilon$ is defined in Lemma \ref{lemma: Borne sur la norme des residus}.
\end{lemma}

\begin{proof}
Note that:
\begin{equation*}
\begin{split}
\E \left \| \rcov - \E \rcov \right \| _{\HSk}^2 
=& \E  \left \| \rcov  \right \| _{\HSk}^2  
- 2 \E  \left \langle \rcov , \E \rcov \right \rangle_{\HSk}
+  \left \| \E \rcov \right \| _{\HSk} ^2
\\ =&\E  \left \| \rcov  \right \| _{\HSk}^2  
- 2  \left \langle\E  \rcov , \E \rcov \right \rangle_{\HSk}
+  \left \| \E \rcov \right \| _{\HSk} ^2
\\ =& \E  \left \| \rcov  \right \| _{\HSk}^2  -   \left \| \E \rcov \right \| _{\HSk} ^2 .
\end{split}  
\end{equation*}
By Jensen's inequality, we have that:
\begin{equation*}
\begin{split}
\E \left \| \rcov - \E \rcov \right \| _{\HSk} 
\leq& \left [\E \left \| \rcov - \E \rcov \right \| _{\HSk}^2  \right ] ^{\frac{1}{2}}
\\ \leq& \left[\E  \left \| \rcov  \right \| _{\HSk}^2  -   \left \| \E \rcov \right \| _{\HSk} ^2   \right ]^{\frac{1}{2}}. 
\end{split}
\end{equation*}
Remember that $\Px = (\Pxij{i}{j})_{1\leq i,j \leq n}$ is the orthogonal projection on $Im(\dX)$. We can develop $\rcov$ to obtain that: 
\begin{equation} 
\label{equation_sigma_epsilon} 
\begin{split} \rcov = \frac{1}{\nobs} \sum_{i = 1}^{\nobs}(1-\Pxij{i}{i}) \error_i \otimes \error_i - \frac{1}{\nobs} \sum_{i,j = 1,i\neq j}^{\nobs}\Pxij{i}{j} \error_i \otimes \error_j. \end{split} 
\end{equation} 
Thus:
\begin{equation*}
\begin{split}
\left \| \rcov  \right \| _{\HSk}^2  = A - 2B + C,
\end{split}
\end{equation*}
where
\begin{equation*}
\begin{split}
A =& \frac{1}{\nobs^2} \sum_{i = 1}^{\nobs}\sum_{j = 1}^{\nobs} (1 - \Pxij{i}{i})(1 - \Pxij{j}{j}) \left \langle \error_i \otimes \error_i, \error_j \otimes \error_j \right \rangle_{\HSk},
\\ B =& \frac{1}{\nobs^2} \sum_{i = 1}^{\nobs}\sum_{j = 1}^{\nobs} \sum_{k = 1,k\neq j}^{\nobs}(1 - \Pxij{i}{i}) \Pxij{j}{k} \left \langle \error_i \otimes \error_i, \error_j \otimes \error_k \right \rangle_{\HSk}, 
\\ C =& \frac{1}{\nobs^2} \sum_{i = 1}^{\nobs}\sum_{j = 1,j\neq i}^{\nobs}\sum_{k = 1}^{\nobs}\sum_{l = 1,l\neq k}^{\nobs}\Pxij{i}{j} \Pxij{k}{l} \left \langle \error_i \otimes \error_j, \error_k \otimes \error_l \right \rangle_{\HSk} .
\end{split}
\end{equation*}
We now compute the expectation of each term: 
\begin{equation*}
\begin{split}
\E (A) =&  
\frac{1}{\nobs^2} \sum_{i = 1}^{\nobs} (1 - \Pxij{i}{i})^2 \E \left( \left \langle \error_i \otimes \error_i, \error_i \otimes \error_i \right \rangle_{\HSk} \right )  
\\&+  \frac{1}{\nobs^2} \sum_{i = 1}^{\nobs}\sum_{j = 1,j\neq i}^{\nobs} (1 - \Pxij{i}{i})(1 - \Pxij{j}{j}) \left \langle \E \left (\error_i \otimes \error_i \right), \E \left (\error_j \otimes \error_j \right ) \right \rangle_{\HSk} 
\\=&  
\frac{1}{\nobs^2} \sum_{i = 1}^{\nobs} (1 - \Pxij{i}{i})^2 \E \left( \left \| \error_i \right \| _{\Hk}^4 \right )  
\\&+  \frac{1}{\nobs^2} \sum_{i = 1}^{\nobs}\sum_{j = 1,j\neq i}^{\nobs} (1 - \Pxij{i}{i})(1 - \Pxij{j}{j}) \left \langle \ecov_i, \ecov_j\right \rangle_{\HSk}.
\end{split}
\end{equation*}
We can directly see that every term of $B$ contains at least one indice $j$ or $k$ different from the others, involving that the expectation of each term of $B$ is null, which give that $\E(B) = 0$. The same happens for every term of $C$ where indices $k$ and $l$ are different from indices $i$ and $j$. Consequently we have:
\begin{equation*}
\begin{split}
\E (C) 
=& \frac{1}{\nobs^2} \sum_{i = 1}^{\nobs}\sum_{j = 1, j \neq i}^{\nobs} \Pxij{i}{j}^2 \left (  \E \left \langle \error_i \otimes \error_j , \error_i \otimes \error_j \right \rangle_{\HSk} + \E \left \langle \error_i \otimes \error_j , \error_j \otimes \error_i \right \rangle_{\HSk}   \right ) 
\\ =& \frac{2}{\nobs^2} \sum_{i = 1}^{\nobs}\sum_{j = 1, j \neq i}^{\nobs} \Pxij{i}{j}^2 \E \left \| \error_i \right \| _{\Hk}^2  \E \left \| \error_j \right \| _{\Hk}^2.
\end{split}
\end{equation*}
According to Equation \eqref{equation_sigma_epsilon}, $\E \rcov$ is such that: 
\begin{equation*} 
\begin{split} 
\E \rcov  =& 
\frac{1}{\nobs} \sum_{i = 1}^{\nobs}(1 - \Pxij{i}{i}) \E (\error_i \otimes \error_i) + \frac{1}{\nobs} \sum_{i,j = 1, i \neq j}^{\nobs}\sum_{j = 1,j\neq i}^{\nobs} \Pxij{i}{j} \E \error_i \otimes \E \error_j  
\\ =& \frac{1}{\nobs} \sum_{i = 1}^{\nobs} (1 - \Pxij{i}{i}) \ecov_i. \end{split}
\end{equation*} 
Then: 
\begin{equation*} 
\begin{split} 
\left \| \E \rcov \right \| _{\HSk}^2 
=& \frac{1}{\nobs^2} \sum_{i = 1}^{\nobs}\sum_{j = 1}^{\nobs}(1 - \Pxij{i}{i})(1 - \Pxij{j}{j})  
\left \langle \ecov_i , \ecov_j \right \rangle_{\HSk}
\\ =& \frac{1}{\nobs^2} \sum_{i = 1}^{\nobs}(1 - \Pxij{i}{i}) ^2\left \| \ecov_i \right \| _{\HSk} ^2 + \frac{1}{\nobs^2}\sum_{i = 1}^{\nobs}\sum_{j = 1,j\neq i}^{\nobs}(1 - \Pxij{i}{i})(1 - \Pxij{j}{j}) \left \langle \ecov_i , \ecov_j \right \rangle_{\HSk}.\end{split} 
\end{equation*}
Now we can sum $\E(A), \E(C)$ and $\left \| \E \rcov \right \| _{\HSk}^2$:
\begin{equation*}
\begin{split}
\E \left \| \rcov \right \| _{\HSk}^2  - \left \| \E \rcov \right \| _{\HSk}^2 
=&  \E(A) + \E(C) - \left \| \E \rcov \right \| _{\HSk}^2
\\ 
=& \frac{1}{\nobs^2} \sum_{i = 1}^{\nobs} (1 - \Pxij{i}{i})^2 
\Big( \E \left \| \error_i \right \| _{\Hk}^4 - \left \| \ecov_i \right \| _{\HSk} ^2 \Big)
\\&+ \frac{2}{\nobs^2} \sum_{i,j = 1, i\neq j}^{\nobs} \Pxij{i}{j}^2 \E \left \| \error_i \right \| _{\Hk}^2  \E \left \| \error_j \right \| _{\Hk}^2
\end{split}
\end{equation*}
Note that under Assumptions \ref{Aboundedkernel} and \ref{Aboundeddesign}, Lemma \ref{lemma: Borne sur la norme des residus} gives that for $i\in\{ 1, \dots, \nobs \}$, $\E \left \| \error_i \right \| _{\Hk}^4 - \left \| \ecov_i \right \| _{\HSk} ^2 \leq \E \left \| \error_i \right \| _{\Hk}^4 \leq \boundepsilon^4$ and that $\E \left \| \error_i \right \| _{\Hk}^2 \leq \boundepsilon^2$. Thus, we have :
\begin{equation*}
\begin{split}
\E \left \| \rcov \right \| _{\HSk}^2  - \left \| \E \rcov \right \| _{\HSk}^2 
\leq& \frac{1}{\nobs^2} \sum_{i = 1}^{\nobs} (1 - \Pxij{i}{i})^2 
\boundepsilon^4 
+ \frac{2}{\nobs^2} \sum_{i,j = 1, i\neq j}^{\nobs} \Pxij{i}{j}^2 \boundepsilon^4
\\ \leq&  
\frac{\boundepsilon^4}{\nobs^2} (\nobs - \rangx + 2 \rangx)
\\ \leq& \frac{\boundepsilon^4 (\nobs + \rangx)}{\nobs^2}
\end{split}
\end{equation*}
by applying Lemma \ref{lemma: sommes de pi ij} to both terms. It finally gives that 
\begin{equation*}
\begin{split}
\E \left \| \rcov - \E \rcov \right \| _{\HSk} 
\leq &  \frac{\boundepsilon^2\sqrt{\nobs - \rangx}}{\nobs} .
\end{split}
\end{equation*}

\end{proof}

\subsection{Probability bound on the inverse covariance operator}

Under Assumption~\ref{Ahomoscedasticity}, we define $\ecovbar = \ecovt + \frac{1}{\nobs} \sum_{i = 1}^{\nobs}\Pxpij{i}{i} \ecovtci{i}$. 
Under Assumptions~\ref{Ahomoscedasticity} and \ref{Asimple}, and according to Lemma~\ref{lem:ecovbar}, $\ecovbar$ is symmetric positive operator and admits a spectral decomposition $\ecovt = \sum_{t = 1}^t  \ecoval_t \ecovec_t \otimes \ecovec_t  $ where the eigenvalues $ \ecoval_1 > \dots >  \ecoval_\tmax $ and eigenfunctions  $\ecovec_1,\dots,\ecovec_\tmax$ correspond to those of  of $\ecovt$.
\begin{proposition}
\label{lemma: Borne exponentielle sur les inverses des projecteurs}
Assume that Assumptions \ref{Ahomoscedasticity}, \ref{Asimple}, \ref{Aboundedkernel} and  \ref{Aboundeddesign} are verified.
There exists $N_\tmax$ such that for $\nobs \geq N_\tmax$, we have with probability $1-e^{-\xi}$:
\begin{equation*}
\begin{split}
\left \| \rcovt^{-1} - \ecovt^{-1} \right \| _{\HSk} 
\leq\frac{\boundrcov}{\ecoval_{\tmax} - \boundrcov} \left ( \sum_{t = 1}^{\tmax} \frac{2}{\tilde \delta_t}+\frac{\tmax}{\ecoval_\tmax} \right),
\end{split}
\end{equation*}
where $\boundrcov$ is given in Proposition~\ref{proposition: Inegalite de concentration sur sigma epsilon emprique et theorique}.  
\end{proposition}

\begin{remark}
In the Proposition, for $t < \tmax$, all the spectral gaps  $\tilde \delta_t$ do not depend on $n$. Regarding the spectral gap at $\tmax$, note that $\widetilde {\delta}_\tmax = \delta_\tmax \wedge \delta_{\tmax-1} \geq (\ecoval_\tmax - \mu_\tmax )/2 \wedge (\ecoval_{\tmax -1} -  \ecoval_\tmax )/2 > 0 $ (see   Remark~\ref{rq:spectralgapT+1}). Consequently the exponential bound of the proposition implies the convergence in probability. 
\end{remark}

\begin{proof} For $t \in \{ 1, \dots, \nobs -1\}$,  denote $\delta_t = (\ecoval_t - \ecoval_{t+1})/2$ and $\tilde \delta_t = \operatorname{min}(\delta_t,\delta_{t-1})$ for the spectral gaps of $\ecov$. Let $\precov{t} = \ecovec_t \otimes \ecovec_t$ and $\prrcov{t} = \rcovec_t \otimes \rcovec_t$ the orthogonal projector on the axis spanned by $\ecovec_t$ (resp. $\rcovec_t$). 

According to Proposition~\ref{proposition: Inegalite de concentration sur sigma epsilon emprique et theorique}, we have that with probability  $1-e^{-\xi}$:
\begin{equation*}
\left \| \rcov - \ecovbar \right \| _{\HSk} \leq   \boundrcov .
\end{equation*}
Moreover, according to the Hoffman-Wielandt Inequality in infinite dimensional setting \cite{bhatia_hoffman-wielandt_1994} (see also the proof of Theorem 3 from \cite{zwald_convergence_2005}), we have
\begin{equation}
\label{eq:bhatia_hoffman-wielandt}
\max_{t  = 1, \dots, n }| \rcoval_t - \ecoval_t | \leq \left \| \rcov - \ecovbar \right \| _{\HSk}.
\end{equation}
According Lemma \ref{lem:ecovbar}, the eigenvalues $\ecoval_1, \dots \ecoval_\tmax$ of $\ecovbar$ are simple.  In particular the spectral gap at $\tmax$ is strictly positive:  $\ecoval_{\tmax} - \ecoval_{\tmax+1} \geq \ecoval_{\tmax} - \mu_{\tmax +1} >0  $, see also Remark~\ref{rq:spectralgapT+1}.
 
It then can be easily checked that there exists $  N_\tmax$ such for $n \geq N_\tmax$, the event 
\begin{equation*}
 A  = \left\{ \| \rcov - \ecovbar   \| _{\HSk} \leq \boundrcov \wedge \min_{t =1 \dots \tmax }  \delta_t \right\}  
\cap \left\{ \rcoval_\tmax >0\right\}  \cap \left\{ \rcoval_1 > \dots > \rcoval_\tmax > \rcoval_{\tmax+1} \right\}
\end{equation*}
is satisfied with probability at least $1-e^{-\xi}$. Moreover, $N_\tmax$ only depends on  $\ecoval_\tmax$, on $\delta_1,\dots \delta_\tmax $, and on the constants involved in  $\boundrcov$, that is $p$, $\boundepsilon$ and $\| \ecovt \| _{\HSk}$.

In particular,  Theorem~\ref{theorem: Theorem 2 from Zwald and Blanchard} applies on the event $A$: for any $ t \in \{1, \dots, \tmax\}$,
\begin{equation}
\label{eq:projhat}
\left \| \prrcov{t} - \precov{t}    \right \| _{\HSk} \leq \frac{2 \left \| \rcov - \ecov \right \| _{\HSk}}{\widetilde {\delta}_t}.
\end{equation} 
On the event $A$, the rank $r$ of $\rcovt$ is larger than $\tmax$ and thus  $\rcovt^{-1} = \sum_{t = 1}^{\tmax} \rcoval_t^{-1} \rcovec_t \otimes \rcovec_t$. We can use the spectral decomposition of the inverse operators to obtain that on $A$:
\begin{eqnarray}
\left \| \rcovt^{-1} - \ecovt^{-1} \right \| _{\HSk} 
& =& \left \| \sum_{t = 1}^{\tmax} \rcoval_t^{-1} \rcovec_t \otimes \rcovec_t -\ecoval_t^{-1} \ecovec_t \otimes \ecovec_t   \right \| _{\HSk} \notag \\
& =& \left \| \sum_{t = 1}^{\tmax} \rcoval_t^{-1} \prrcov{t} -\ecoval_t^{-1} \precov{t}    \right \| _{\HSk} \notag \\
& \leq& \sum_{t = 1}^{\tmax} \rcoval_t^{-1} \left \| \prrcov{t} - \precov{t}\right \| _{\HSk} +  \sum_{t = 1}^{\tmax}  | \rcoval_t^{-1} - \ecoval_t^{-1}| \left \| \precov{t} \right \| _{\HSk}. \label{SigmainvBound}
\end{eqnarray}
 For the first term of \eqref{SigmainvBound},  on the event $A$ we find from \eqref{eq:projhat} that  for $n \geq \tilde N_\tmax$
\begin{eqnarray*}
 \sum_{t = 1}^{\tmax} \rcoval_t^{-1} \left \| \prrcov{t} - \precov{t}\right \| _{\HSk} 
 & \leq& \frac{\boundrcov}{\rcoval_\tmax}  \sum_{t = 1}^{\tmax} \frac{2}{\tilde \delta_t} \\
  &\leq& \frac{\boundrcov}{\ecoval_{\tmax} - \boundrcov}  \sum_{t = 1}^{\tmax} \frac{2}{\tilde \delta_t}.
\end{eqnarray*}
Regarding the second term in the upper bound \eqref{SigmainvBound}, on $A$ and for $n \geq \tilde N_\tmax$, we have that $ \frac{1}{\rcoval_{\tmax}} \leq  \frac{1}{\ecoval_{\tmax} - \boundrcov} $ and thus for all $t \in \{ 1, \dots, \tmax \}$  
\begin{equation*}
\begin{split}
| \rcoval_t^{-1} - \ecoval_t^{-1}| 
=& \rcoval_t^{-1} \ecoval_t^{-1} | \rcoval_t - \ecoval_t | 
\\ \leq& \ecoval_\tmax^{-1}  \rcoval_\tmax^{-1} \boundrcov,
\\ \leq& \frac{1}{\ecoval_\tmax} \frac{\boundrcov}{\ecoval_{\tmax} - \boundrcov}.
\end{split}
\end{equation*}
Also note that   $\left \| \precov{t} \right \| _{\HSk} = \left \| \ecovec_{t} \right \| _{\Hk}^2 = 1$. Thus, on   the event $A$, we have that 
\begin{equation*}
\begin{split}
\sum_{t = 1}^{\tmax}  | \rcoval_t^{-1} - \ecoval_t^{-1}| \left \| \precov{t} \right \| _{\HSk} 
\leq \frac{\tmax}{\ecoval_\tmax} \frac{\boundrcov}{\ecoval_{\tmax} - \boundrcov} .
\end{split}
\end{equation*}
Thus, for $\nobs \geq N_{\tmax}$, we have with probability $1-e^{-\xi}$:
\begin{equation*}
\begin{split}
\left \| \rcovt^{-1} - \ecovt^{-1} \right \| _{\HSk}  
\leq \frac{\boundrcov}{\ecoval_{\tmax} - \boundrcov} \left ( \sum_{t = 1}^{\tmax} \frac{2}{\tilde \delta_t}+\frac{\tmax}{\ecoval_\tmax} \right).
\end{split}
\end{equation*}
\end{proof}

\begin{lemma}
\label{lem:ecovbar} 
Under Assumptions~\ref{Ahomoscedasticity} and \ref{Asimple}, the operator $\ecovbar = \ecovt + \frac{1}{\nobs} \sum_{i = 1}^{\nobs}\Pxpij{i}{i} \ecovtci{i}$    is a symmetric and positive operator that admits a spectral decomposition $\ecovbar = \sum_{t = 1}^t  \ecoval_t \ecovec_t \otimes \ecovec_t  $ where the eigenvalues are in descending order, with the first ones $ \ecoval_1 > \dots >  \ecoval_\tmax $ and the corresponding eigenfunctions  $\ecovec_1,\dots,\ecovec_\tmax$ being the same as those of of $\ecovt$. Moreover, the spectral gap at $\tmax$ is strictly positive:  $\ecoval_{\tmax} - \ecoval_{\tmax+1} \geq \ecoval_{\tmax} - \mu_{\tmax +1} >0  $.
\end{lemma}
\begin{remark}
\label{rq:spectralgapT+1}
Note that   $\ecoval_t$ depends on $n$ for $t \geq  \tmax+1 $ and it does not for $t \leq \tmax$. Thus, the spectral gap  $\delta_t = \ecoval_t - \ecoval_{t+1}$ does not depend on $n$ for $t \leq \tmax-1$. Moreover, the last spectral gap $\ecoval_{\tmax} - \ecoval_{\tmax+1}$ is lower bounded by $\ecoval_{\tmax} - \mu_{\tmax +1}$ for all $n \geq 1$. 
\end{remark}
\begin{proof}
The operator $\ecovbar$ is symmetric and positive as a linear combination of symmetric and positive operators, and because $\Pxpij{i}{i} =  \sum_{j=1}^n\Pxpij{i}{j}^2 \geq 0 $ (see Lemma \ref{lemma: sommes de pi ij}). We directly check that for $t \leq \tmax$, the eigenfunction $f_t$ for the eigenvalue $\ecoval_t$ of $\ecovt$ is also an eigenfunction of $\ecovbar$  for the eigenvalue $\ecoval_t$.
Moreover, for any $h \in \Span{(f_1, \dots, f_\tmax)}^\perp$, we have   $\| \ecovbar h  \|_\Hk \leq  \frac{1}{\nobs} \sum_{i = 1}^{\nobs}\Pxpij{i}{i} \|\ecovtci{i} h  \|_\Hk \leq  \|  h  \|_\Hk \max_{i=1 \dots n} \ecoval_{i,\tmax +1} \left(  \frac{1}{\nobs} \sum_{i = 1}^{\nobs}\Pxpij{i}{i} \right) <    \frac{\nobs - p}{\nobs}  \mu_{\tmax +1}\|  h  \|_\Hk      $. Thus  $ 0 \leq \ecoval_{t}  < \mu_{\tmax +1} <   \ecoval_\tmax$ for any $ t > \tmax +1$. In particular $\ecoval_\tmax - \ecoval_{\tmax+1} \geq \ecoval_\tmax  -  \mu_{\tmax +1} >0 $.
\end{proof}

\section{Proof of Theorem~\ref{theorem: Asymptotic distribution of the nystrom statistic}: KFDA with Nyström}

\label{App:Proof-Theo-chi2:Nystrom}
 
This section provides the proof for the consistency of the Nyström TKHL statistic, see Theorem~\ref{theorem: Asymptotic distribution of the nystrom statistic}.

\subsection{Main proof of Theorem~\ref{theorem: Asymptotic distribution of the nystrom statistic}}

Let $\Pxp = \identity_\nobs - \Px = (\Pxpij{i}{j})_{1\leq i,j \leq n} $ the orthogonal projection on $Im(\dX)^\perp$.
Under Assumptions~\ref{AhomoscedasticityN} and \ref{AsimpleN}, as for the non-Nyström case, we consider  
\begin{eqnarray*}
\ecovbar  &:=& \ecovt + \frac{1}{\nobs} \sum_{i = 1}^{\nobs}\Pxpij{i}{i} \ecovtci{i}  \\
   &=&  \ecovanc + \frac{1}{\nobs} \sum_{i = 1}^{\nobs}\Pxpij{i}{i}  \ecovancci{i},
\end{eqnarray*}
where the second equality is valid because $\tmax \leq \nanchors$. According to Lemma~\ref{lem:ecovbar}, we have that $\ecovbar$ is a symmetric and semi-definite positive operator that admits a spectral decomposition $\ecovbar = \sum_{t \geq 1}   \ecoval_t \ecovec_t \otimes \ecovec_t  $ where the eigenvalues are in descending order, with the first ones $ \ecoval_1 > \dots >  \ecoval_\nanchors $ and the corresponding eigenfunctions $\ecovec_1,\dots,\ecovec_\nanchors$ being the same as those of $\ecovanc$.

Let $I$ be a $n$-uple of indices sampled uniformly in $\{1, \dots,n\}$ to define the $\nlandmarks$ landmarks for Nyström. Let also $\lY=(\ey{i})_{i \in I} = (\ly_1,\dots,\ly_ \nlandmarks)$ be the $\nlandmarks$ landmarks and $\lX$ the associated design matrix extracted from $\dX$. Let $\lrcov$ be the landmark residual covariance operator induced by the linear model \eqref{eq:linear model on embeddings}  based on $\lY,\lX$. For this model we define (as in the non-Nyström case) the operator $\ecovbarI =\ecovanc + \frac{1}{\nlandmarks} \sum_{i=1}^q \PxpIij{i}{i} \ecovancci{i}$, where $\PxpI = ( \PxpIij{i}{j})$ is the orthogonal projection onto $\operatorname{Im}(\lX)^\perp$ in $\R^\nlandmarks$. According to Lemma~\ref{lem:ecovbar} (adapted to this model), we have that $\ecovbarI$ is a symmetric and semi-definite positive operator. It admits a spectral decomposition which first $m$ eigenvalues  $ \ecoval_1 > \dots >  \ecoval_\nanchors $ and  corresponding eigenfunctions $\ecovec_1,\dots,\ecovec_\nanchors$ are the same as those of  $\ecovanc$. 

For $s \leq \nanchors-1$, we denote by $\delta_t = (\ecoval_t - \ecoval_{t+1})/2$ the  spectral gaps of $\ecovanc$ (and of $\ecovbarI$).

We define the $\nanchors$ first unit eigenfunctions $(\anc_1,\dots,\anc_\nanchors)$ of the landmark residual covariance operator $\lrcov$ as our Nyström anchors. The Nyström residual covariance operator $\arcov$ is finally defined as $ \arcov=\frac{1}{\nobs}\sum_{i = 1}^{\nobs} \aresidual_i \otimes \aresidual_i$ 
where  $\aresidual_i$ is the orthogonal projection of $\residual_i$ onto $\Span(\anc_1,\dots,\anc_\nanchors)$. 

According to Proposition~\ref{proposition:concentration nystrom residual covariance}, when $\nlandmarks$ is large enough, we have with probability $1-2 e^{-\xi}$: 
\begin{equation*}
\left \| \arcov -  \ecovbar  \right \| _{\HSk} \leq 
\frac12 \boundrcov +  2 \nanchors  \boundepsilon^2  \boundlrcov \sum_{s = 1}^{\nanchors} \frac{1}{\widetilde {\delta}_s}.
\end{equation*}
From this, we obtain (see Proposition~\ref{prop:SigmainvN}) that with probability $1-2 e^{-\xi}$,
\begin{equation*}
 \left \| \arcovt{}^{-1} - \ecovt^{-1} \right \| _{\HSk} 
\leq \frac{\eta(n,q,\zeta)}{\ecoval_{\tmax} - \eta(n,q,\zeta)} \left ( \sum_{t = 1}^{\tmax} \frac{2}{\tilde \delta_t}+\frac{\tmax}{\ecoval_\tmax} \right)
\end{equation*}
with $\eta(n,q,\zeta) = \frac{C_1}{\sqrt n}  \left( 1 + \sqrt \zeta \right) + \frac{C_2}{\sqrt q}  \left( 1 + \sqrt \zeta \right) \sum_{s = 1}^{\nanchors} \frac{1}{\widetilde {\delta}_s}  $, 
where $C_1$ and $C_2$ are constants that depend  on $(\lambda_s)_{s\leq \nanchors}$,  $\mu_{\nanchors +1}$, $\rangx$ and $\boundepsilon$. Thus, we have
$ \left \| \arcovt{}^{-1} - \ecovt^{-1} \right \| _{\HSk} 
\leq \frac{C'}{\sqrt q} \left( 1 + \sqrt \zeta \right)$ where $C'$ is a  constant which does not depends on $q$ (and $n$). Thus $\arcovt{}^{-1}$ converges in probability to $\ecovt^{-1}$ as $q$ tends to infinity.

As the landmarks are not used to approximate $\Hmanova$, we conclude  the proof of Theorem~\ref{theorem: Asymptotic distribution of the nystrom statistic} by following the lines of the proof of Theorem \ref{theorem: Asymptotic distribution of the truncated Hotelling-Lawley test statistic}: we use the convergence in probability of $\arcovt{}^{-1}$ (instead of the these of $\rcovt^{-1}$) together with Proposition~\ref{proposition: Distrib asymptotique de la stat alternative}.

\subsection{Probability bound on the Nyström residual covariance}

\begin{proposition}
\label{proposition:concentration nystrom residual covariance}
Under Assumptions~\ref{AhomoscedasticityN}, \ref{AsimpleN}, \ref{Aboundedkernel} and \ref{Aboundeddesign}, when $\nlandmarks$ is large enough, we have with probability $1-2 e^{-\xi}$:
\begin{equation*}
\begin{split}
\left \| \arcov -  \ecovbar  \right \| _{\HSk} \leq 
\frac12 \boundrcov +  2 \nanchors  \boundepsilon^2  \boundlrcov \sum_{s = 1}^{\nanchors} \frac{1}{\widetilde {\delta}_s} := \eta(n,q,\zeta)
\end{split}
\end{equation*}
where $\zeta$ is  defined in Proposition~\ref{proposition: Inegalite de concentration sur sigma epsilon emprique et theorique}, and where $\widetilde {\delta}_s = \delta_s \wedge \delta_{s-1}$ for $s < \nanchors$ and $\delta_\nanchors := (\ecoval_{\nanchors} - \mu_{\nanchors +1})/2$.
\end{proposition}
\begin{proof}
For $s\in\{ 1, \dots, \nanchors \}$, we denote by $\precov{s}$ and $\pranc{s}$ the orthogonal projections onto the one-dimensional axis supported by $\ecovec_s$ and $\anc_s$ respectively. Then we have $\precov{s}=\ecovec_s\otimes\ecovec_s$, $\pranc{s} = \anc_s\otimes\anc_s$ and by definition of $\arcov$:
\begin{equation*}
\begin{split}
\arcov = \frac{1}{\nobs	}\sum_{i = 1}^{\nobs} \Big(\sum_{s = 1}^{\nanchors}\pranc{s}\residual_i\Big) \otimes \Big(\sum_{s^\prime = 1}^{\nanchors}\pranc{s^\prime}\residual_i\Big).
\end{split}
\end{equation*}
We also define :
\begin{equation*}
\begin{split}
\arcovbis = \frac{1}{\nobs} \sum_{i = 1}^{\nobs} \Big(\sum_{s = 1}^{\nanchors}\precov{s}\residual_i\Big) \otimes \Big(\sum_{s^\prime = 1}^{\nanchors}\precov{s^\prime}\residual_i\Big) ,
\end{split}
\end{equation*}
and we have
\begin{equation}
    \label{eq:decompNys}
\left \| \arcov -  \ecovbar \right \| _{\HSk}  \leq \left \| \arcov - \arcovbis \right \| _{\HSk}  + \left \| \arcovbis - \ecovbar \right \| _{\HSk}.
\end{equation}
First, for $i\in\{ 1, \dots, \nobs \}$ and $s\in\{ 1, \dots, \nanchors \}$, we have : 
\begin{equation*}
\begin{split}
(\pranc{s} \residual_i) \otimes (\pranc{s^\prime} \residual_i) - (\precov{s} \residual_i) \otimes (\precov{s^\prime} \residual_i)  
=& \big((\pranc{s} - \precov{s}) \residual_i \big)\otimes (\pranc{s^\prime} \residual_i) + (\precov{s} \residual_i) \otimes \big((\pranc{s^\prime} - \precov{s^\prime}) \residual_i\big)
\end{split}
\end{equation*}
According to the triangular inequality, we have :
\begin{eqnarray*}
\left \| \arcov - \arcovbis \right \| _{\HSk} 
& \leq & \frac{1}{\nobs} \sum_{i = 1}^{\nobs} \sum_{s,s^\prime = 1}^{\nanchors} 
\left \| \big((\pranc{s} - \precov{s}) \residual_i \big)\otimes (\pranc{s^\prime} \residual_i) 
\right \| _{\HSk} 
\\ & &+ \left \|
(\precov{s} \residual_i) \otimes \big((\pranc{s^\prime} - \precov{s^\prime}) \residual_i\big)   \right \| _{\HSk} 
\\
&\leq&  
\frac{1}{\nobs} \sum_{i = 1}^{\nobs} \sum_{s,s^\prime = 1}^{\nanchors}
\left \| \pranc{s} - \precov{s} \right \| _{\HSk} \left \| \pranc{s^\prime} \right \| _{\HSk} \left \| \residual_i \right \| _{\Hk}^2 
\\ & &+ \left \| \pranc{s^\prime} - \precov{s^\prime} \right \| _{\HSk} \left \| \precov{s} \right \| _{\HSk} \left \| \residual_i \right \| _{\Hk}^2  \\
&\leq&  2  \nanchors  \boundepsilon^2 \sum_{s = 1}^{\nanchors}\left \|  \pranc{s} - \precov{s}  \right \| _{\HSk}  .
\end{eqnarray*} 
because $  \| \precov{s} \| _{\HSk} =  \| \pranc{s}   \| _{\HSk} = 1$ and $\| \residual_i  \| \leq \boundepsilon  $, see Lemma \ref{lemma: Borne sur la norme des residus}. 

 We now  apply Proposition~\ref{proposition: Inegalite de concentration sur sigma epsilon emprique et theorique} to $\lrcov$ with the sample $\lY$ which associated design matrix design is $\lX$.  Assumption \ref{AhomoscedasticityN} is still satisfied for $\lY$, which corresponds to  Assumption \ref{Ahomoscedasticity} for a truncate parameter $\nanchors$ (instead of $\tmax$). Assumptions   \ref{Aboundedkernel} and \ref{Aboundeddesign} are also obviously valid. Thus, conditionally to the Nyström sample $I$, with probability larger than $1-e^{-\xi}$,  
 \begin{equation*}
 \left \| \lrcov - \ecovbarI \right \| _{\HSk}  \leq 
 \frac{\nx }{\nlandmarks} \left \| \ecovanc \right \| _{\HSk}
 +   \frac{\boundepsilon^2}{\sqrt \nlandmarks} \left(1 +  4   \sqrt{\frac{3 \xi}{2} } \right) =: \boundlrcov .
 \end{equation*}
Note that this upper bound does not depends on the Nyström sample $I$ and thus the probability bound is also valid after intergrating according to the distribution of $I$.

The $\nanchors$ first eigenvalues of  $\ecovbarI$ being simple, following the first lines of the proof of Proposition~\ref{lemma: Borne exponentielle sur les inverses des projecteurs}, there exists $  Q_\nanchors$ such for $q \geq Q_\nanchors$, the event 
\begin{equation*}
 \tilde A  = \left\{ \| \lrcov - \ecovbarI   \| _{\HSk} \leq \boundlrcov \wedge \min_{t =1 \dots \tmax }  \delta_t \right\}  
\cap \left\{ \rcoval_\tmax >0\right\}  \cap \left\{ \rcoval_1 > \dots > \rcoval_\tmax > \rcoval_{\tmax+1} \right\}
\end{equation*}
is satisfied with probability at least $1-e^{-\xi}$. Regarding the last spectral gap,  according to Lemma~\ref{lem:ecovbar}, the last spectral gap $\ecoval_{\nanchors} - \ecoval_{\nanchors+1}$ is lower bounded by $\ecoval_{\nanchors} - \mu_{\nanchors +1}$ for all $q \geq 1$, see Remark~\ref{rq:spectralgapT+1}.

Following the first lines of the proof of Proposition~\ref{lemma: Borne exponentielle sur les inverses des projecteurs}, the eigenvalues of $\ecovbarI$ being simple, we find that Theorem~\ref{theorem: Theorem 2 from Zwald and Blanchard} can be applied on an event of probability at least $1-e^{-\xi}$ and for $q$ large enough:  for $s \leq \nanchors$,
\begin{equation*}
\left \|  \pranc{s} - \precov{s}  \right \| _{\HSk}\leq \frac{2 \left \| \ecovbarI  - \lrcov \right \| _{\HSk}}{\widetilde {\delta}_s}  
\end{equation*}
where  $\widetilde {\delta}_s = \delta_s \wedge \delta_{s-1}$, and where we take $\delta_\nanchors := (\ecoval_{\nanchors} - \mu_{\nanchors +1})/2$. Thus, for $\nlandmarks$  large enough, we have that on the event $\tilde A$: 
\begin{equation}
\label{eq:Terme1}
\left \| \arcov - \arcovbis \right \| _{\HSk} 
\leq  4  \nanchors  \boundepsilon^2  \boundlrcov \sum_{s = 1}^{\nanchors} \frac{1}{\widetilde {\delta}_s} .
\end{equation} 

Regarding the second term in the upper bound \eqref{eq:decompNys}, we have $\precov{s} = \ecovec_s \otimes \ecovec_s$, which  leads to : 
\begin{equation*}
\begin{split}
\arcovbis 
=& \frac{1}{\nobs} \sum_{i = 1}^{\nobs} \sum_{s,s^\prime = 1}^{\nanchors} (\precov{s} \residual_i)\otimes (\precov{s^\prime} \residual_i) 
\\=& \frac{1}{\nobs}  \sum_{i = 1}^{\nobs} \sum_{s,s^\prime = 1}^{\nanchors} 
\left \langle \ecovec_s , \residual_i \right \rangle_{\Hk} \left \langle \ecovec_{s^\prime} , \residual_i \right \rangle_{\Hk} \ecovec_s \otimes \ecovec_{s^\prime}
\\=& 
\sum_{s,s^\prime = 1}^{\nanchors} \left \langle \ecovec_s , \rcov \ecovec_{s^\prime} \right \rangle_{\Hk} \ecovec_s \otimes \ecovec_{s^\prime}.
\end{split}
\end{equation*}
We also have that for $\ecovbar$:
\begin{equation*}
\begin{split}
\sum_{s,s^\prime = 1}^{\nanchors} \left \langle \ecovec_s , \ecovbar \ecovec_{s^\prime} \right \rangle_{\Hk} \ecovec_s \otimes \ecovec_{s^\prime}
=& \sum_{s,s^\prime = 1}^{\nanchors}\ecoval_{s^\prime} \left \langle \ecovec_s , \ecovec_{s^\prime} \right \rangle_{\Hk}\ecovec_s \otimes \ecovec_{s^\prime}
\\=& \sum_{s = 1}^{\nanchors}\ecoval_s \ecovec_s \otimes \ecovec_s
\\=& \ecov_\nanchors.
\end{split}
\end{equation*} 
Finally, we obtain:
\begin{equation*}
\begin{split}
\left \| \arcovbis - \ecovbar \right \| _{\HSk}
=& \sum_{s,s^\prime = 1}^{\nanchors} \left \langle \ecovec_s  , (\rcov - \ecovbar) \ecovec_{s^\prime} \right \rangle_{\Hk}\ecovec_s \otimes \ecovec_{s^\prime}
\\ \leq&
\sum_{s,s^\prime = 1}^{\nanchors} \left \| \rcov - \ecovbar \right \| _{\HSk}\left \| \ecovec_s \right \| _{\Hk}^2 \left \| \ecovec_{s^\prime} \right \| _{\Hk}^2
\\ \leq& 
\nanchors^2  
\left \| \rcov - \ecovbar \right \| _{\HSk}.
\end{split}
\end{equation*}
We then apply Proposition \ref{proposition: Inegalite de concentration sur sigma epsilon emprique et theorique} (for all the data) to control this term  and combine it  with \eqref{eq:Terme1} to conclude the proof of the proposition.
\end{proof}

\subsection{Probability bound on the inverse Nyström residual covariance}

\begin{proposition}
\label{prop:SigmainvN}
Assume that Assumptions \ref{AhomoscedasticityN}, \ref{AsimpleN} \ref{Aboundedkernel} and \ref{Aboundeddesign}  are verified. When $q$ and $n$ are large enough, with probability $1-2e^{-\xi}$:
\begin{equation*}
 \left \| \arcovt{}^{-1} - \ecovt^{-1} \right \| _{\HSk} 
\leq \frac{\eta(n,q,\zeta)}{\ecoval_{\tmax} - \eta(n,q,\zeta)} \left ( \sum_{t = 1}^{\tmax} \frac{2}{\tilde \delta_t}+\frac{\tmax}{\ecoval_\tmax} \right), 
\end{equation*}
where $\eta(n,q,\zeta)$ is given in Proposition~\ref{proposition:concentration nystrom residual covariance}. 
\end{proposition}
\begin{proof}
We follow exactly the lines of Proposition~\ref{lemma: Borne exponentielle sur les inverses des projecteurs}, but starting from Proposition~\ref{proposition:concentration nystrom residual covariance} instead of Proposition~\ref{proposition: Inegalite de concentration sur sigma epsilon emprique et theorique}. 
\end{proof}

\section{Auxiliary results}

\subsection{McDiarmid Inequality}

\begin{theorem}[McDiarmid inequality \cite{mcdiarmid_method_1989}]
\label{theorem: McDiarmid}
If $\dy_1,\dots,\dy_\nobs$ are   i.i.d. random variables in a measurable space $\sY$ and the function
\begin{eqnarray*}
&\sY^\nobs &\longrightarrow \R
\\f:&\fy_1,\dots,\fy_\nobs &\longmapsto f(\fy_1,\dots,\fy_\nobs)
\end{eqnarray*}
is such that for all $i \in \{ 1, \dots, \nobs \}$, we have:
\begin{equation*}
\begin{split}
\sup_{\fy_1,\dots,\fy_\nobs,\fy_{i^\prime}\in\sY} \big|f(\fy_1,\dots,\fy_i,\dots,\fy_\nobs)-f(\fy_1,\dots,\fy_{i^\prime},\dots,\fy_\nobs) \big| \leq c_i.
\end{split}
\end{equation*}
Then we have with probability lower than $e^{-\xi}$ that:
\begin{equation*}
\begin{split}
f(\fy_1,\dots,\fy_\nobs) - \E\big(f(\fy_1,\dots,\fy_\nobs)\big) \geq \sqrt{\frac{\xi}{2} \sum_{i = 1}^{\nobs}c_i^2},
\end{split}
\end{equation*}
\end{theorem}
and we also have with probability lower than $e^{-\xi}$ that:
\begin{equation*}
\begin{split}
\E\big(f(\fy_1,\dots,\fy_\nobs)\big) - f(\fy_1,\dots,\fy_\nobs)  \geq \sqrt{\frac{\xi}{2} \sum_{i = 1}^{\nobs}c_i^2}.
\end{split}
\end{equation*}

\subsection{Perturbation bounds}

In this section, we use some results from operator perturbations theory to obtain an exponential bound on a truncated inverse of $ \rcov$, that is used in the proof of Theorem \ref{theorem: Asymptotic distribution of the truncated Hotelling-Lawley test statistic}. We apply the following result from   \cite{zwald_convergence_2005} and \cite{koltchinskii_random_2000}  rewritten in the particular case where only the $T$ first eigenvalues are assumed to be simple.
  
\begin{theorem}[Theorem 2 from \cite{zwald_convergence_2005}, Lemma 5.2 from \cite{koltchinskii_random_2000}]
\label{theorem: Theorem 2 from Zwald and Blanchard}
Let $A \in \HSk$ be a symmetric positive operator such that its   $T$ first eigenvalues are positive and simple: $\ecoval_1>\ecoval_2>\dots \ecoval_\tmax > \ecoval_{\tmax+1} \geq \ecoval_{\tmax+1}  \geq \ldots \geq 0$. 
The spectral gaps of $A$ at dimension  $t$ are denoted $\delta_t = (\ecoval_t - \ecoval_{t+1})/2$  and $\widetilde {\delta}_t = \delta_t \wedge \delta_{t-1}$. Let $B \in \HSk$ a symmetric operator such that $A+B$ is a positive operator  such that the $T$ first eigenvalues are also positive and simple. If for some $t \leq \tmax$  we have $\left \| B \right \| _{\HSk} \leq    \widetilde{\delta}_t/2$, then 
\begin{equation*}
\begin{split}
\left \| \Pi_t(A) - \Pi_t(A+B) \right \| _{\HSk} \leq \frac{2 \left \| B \right \| _{\HSk}}{\widetilde {\delta}_t},
\end{split}
\end{equation*}
where $\Pi_t(A)$ denotes the orthogonal projector onto the one-dimensional subspace of $\Hk$ spanned by the $t^{th}$ eigenfunction of $A$. 
\end{theorem}

\subsection{Results on Orthogonal Projectors}

\begin{lemma}
\label{lemma: sommes de pi ij}
Let $\Px = (\Pxij{i}{j})_{i,j \in \{ 1, \dots, \nobs \}} \in \mtrx{\nobs}$ be the matrix of an orthogonal projector of rank $\rangx$, then we have:
\begin{equation*}
\begin{split}
&\sum_{i = 1}^{\nobs}\sum_{j = 1}^{\nobs}\Pxij{i}{j} ^2 = \rangx,
\\ &\sum_{i = 1}^{\nobs} \Pxij{i}{i}^2 \leq \rangx,
\\ &\sum_{i = 1}^{\nobs}(1 - \Pxij{i}{i}) ^2 \leq  \nobs - \rangx .
\end{split}
\end{equation*}
\end{lemma}

\begin{proof}
As $\Px' \Px'  = \Px$, we directly have $\sum_{j = 1}^{\nobs}\Pxij{i}{j} ^2 = \Pxij{i}{i}$ and by computing the trace we find that:
\begin{equation*}
\begin{split}
\sum_{i = 1}^{\nobs}\sum_{j = 1}^{\nobs}\Pxij{i}{j} ^2 = \rangx.
\end{split}
\end{equation*}
Then:
\begin{equation*}
\begin{split}
\sum_{i = 1}^{\nobs} \Pxij{i}{i}^2 
\leq& \sum_{i = 1}^{\nobs}\sum_{j = 1}^{\nobs}\Pxij{i}{j} ^2  
\\\leq& \rangx,
\\ \sum_{i = 1}^{\nobs}(1 - \Pxij{i}{i}) ^2 
=& \sum_{i = 1}^{\nobs}(1 - 2\Pxij{i}{i} +  \Pxij{i}{i}^2) 
\\ \leq&  \nobs - \rangx.
\end{split}
\end{equation*}
\end{proof}

\section{Kernel tricks}

\label{App:KernelTricks}

In this section, we apply kernel tricks to compute explicit expressions for all the quantities presented in the paper.


\subsection{Kernel trick for the truncated kernel Hotelling-Lawley statistic}

\subsubsection{Diagonalization of the residual covariance operator. }

We start with the determination of the eigenfunction and associated eigenvalues of the residual covariance operator $\rcov$. Observe that $\residuals\in\Hk^\nobs$ and $\rcov = \nobs^{-1} \residuals^\star \residuals = \nobs^{-1	} \sum_{i = 1}^{\nobs}\residual_i \otimes \residual_i \in \HSk$, where $\residuals = \Pxp \eY \in\Hk^\nobs$  (see Section~\ref{subs:opHkn} for more details on operators defined on $\in\Hk^\nobs$).

We then define $\rcovkt =  \nobs^{-1} \residuals \residuals^\star = \Big( \nobs^{-1}\left \langle \residual_i , \residual_j \right \rangle_{\Hk} \Big)_{i,j \in \{ 1, \dots, \nobs \}}\in\mtrx{\nobs}$. The operator $\Pxp$ being linear, it can be easily checked that
\begin{equation*}
\begin{split}
\rcovkt = \frac{1}{\nobs} \Pxp \gram \Pxp,
\end{split}
\end{equation*}
where $\gram = (\kernel(\dy_i,\dy_j))_{i,j \in \{1,\dots,\nobs \} }$ is the gram matrix of $\dY$ with respect to $\kernel(\cdot,\cdot)$. Thus the matrix $\rcovkt$ can be numerically computed and diagonalized. The following proposition highlights the link between $\rcov$ and $\rcovkt$. 

\begin{proposition}
\label{proposition: Diagonalisation de sigma epsilon}
The Hilbert-Schmidt operator $\rcov$ has the same spectrum than the matrix $\rcovkt$. Moreover, if $\rcoveckt \in \R^n$ is an unit eigenvector of $\rcovkt$ associated with the eigenvalue $\rcoval \Hk$, then an unit eigenfunction $\rcovec \in $ of $\rcov$ associated with $\rcoval$ may be obtained by :
\begin{equation}
\label{eq: formule de f en fonction de u}
\begin{split}
\rcovec = \frac{1}{\sqrt{\nobs \rcoval}}  \rcoveckt'\residuals.
\end{split}
\end{equation}
\end{proposition}

\begin{proof}
It can be checked that $ \nobs^{-1	}(\residuals\residuals^\star) \residuals =  \nobs^{-1	}\residuals(\residuals^\star \residuals)= \residuals \rcov \in \Hk^n$, where $\residuals$ is interpreted as an application defined as in \eqref{eq:gA}. Thus we have the relation 
\begin{equation}
\label{eq:relation_duale_kernel_trick}
\rcovkt \residuals = \residuals \rcov \in \Hk^n.
\end{equation}
Then, if $\rcovec \in \Hk$ is an unit eigenfunction of $\rcov$ associated with the eigenvalue $\rcoval$, we have (according to the application defined in \eqref{HkcontreH}):
\begin{equation*}
\begin{split}
&\rcov \rcovec = \rcoval \rcovec
\\ \Leftrightarrow & \residuals \rcov \rcovec = \rcoval \residuals \rcovec \in \R^n 
\\ \Leftrightarrow &  \rcovkt \big(\residuals\rcovec\big) = \rcoval \big(\residuals \rcovec\big)
\end{split}
\end{equation*}
Thus the vector $\residuals \rcovec\in\R^\nobs$ is an eigenvector of $\rcovkt$ associated with the eigenvalue $\rcoval$, and the spectrum of $\rcov$ is included in the spectrum of $\rcovkt$. On the other hand, let $\rcoveckt \in \R^n$ be an unit eigenvector of $\rcovkt$ associated with the eigenvalue $\rcoval$. We have:
\begin{equation*}
\begin{split}
& \rcoveckt' \rcovkt= \rcoval \rcoveckt'
\\ \Leftrightarrow & \rcoveckt' \rcovkt\residuals = \rcoval  \rcoveckt' \residuals
\\ \Leftrightarrow & \rcoveckt' (\residuals \rcov) = \rcoval  \rcoveckt' \residuals
\\ \Leftrightarrow & \rcov \big(\rcoveckt' \residuals\big)= \rcoval \big(\rcoveckt'\residuals\big),
\end{split}
\end{equation*}
where the last equivalence is due to $\rcoveckt' (\residuals \rcov) = \rcoveckt' \big(\rcov \residual_i\big)_{i=1,\dots,\nobs}= \rcov \sum_{i=1}^\nobs \rcoveckt_i \residual_i = \rcov (\rcoveckt' \residuals)$.


Thus the vector $\rcoveckt'\residuals\in\Hk$ is an eigenvector of $\rcov$ associated with the eigenvalue $\rcoval$, and the spectrum of $\rcovkt$ is included in the spectrum of $\rcov$. We conclude that $\rcov$ and $\rcovkt$ share the same spectrum. 

Let $\rcoveckt$ be an unit eigenvector of $\rcovkt$ associated with the eigenvalue $\rcoval$. We have that:
\begin{equation*}
\begin{split}
\left \| \rcoveckt' \residuals \right \| _{\Hk}^2 
=& \left \langle \sum_{i = 1}^{\nobs}\rcoveckt_i \residual_i , \sum_{j = 1}^{\nobs}\rcoveckt_j \residual_j \right \rangle_{\Hk}
\\=& \sum_{i = 1}^{\nobs} \sum_{j = 1}^{\nobs}  \rcoveckt_i \left \langle \residual_i , \residual_j \right \rangle_{\Hk} \rcoveckt_j
\\=& \nobs \rcoveckt^\prime \rcovkt \rcoveckt
\\=&  \nobs \rcoval  \rcoveckt^\prime  \rcoveckt 
\\=& \nobs \rcoval
\end{split}
\end{equation*}
Thus, $(\nobs \rcoval)^{-\frac{1}{2}} \rcoveckt'\residuals \in\Hk$ is an unit eigenfunction of $\rcov$. 
\end{proof}

\newcommand{\prey}{\mathbf{K}_\tmax}
\newcommand{\aprey}{{\mathbf{K}_\tmax^{\Anc}}}
\subsubsection{Computation of the truncated kernel Hotelling-Lawley trace statistic}
\paragraph{}
Let $\rcovecskt = (\rcoveckt_1,\dots,\rcoveckt_\tmax) \in \mtrx{\nobs,\tmax}$ and $\rcovals = \operatorname{diag}(\rcoval_1,\dots,\rcoval_\tmax) \in \mtrx{\tmax}$ be the matrix containing the $\tmax$ first eigenvectors and eigenvalues of $\rcovkt$ respectively, and let 
$\mathbf{D} = \XXXprime \Lmanova^\prime (\Lmanova \XXinv \Lmanova^\prime )^{-1} \Lmanova \XXX = (d_{i,j})_{i,j\in \{1,\dots,\nobs \} } \in \mtrx{\nobs}$. $\mathbf{D}$ is an orthogonal projector, as $\mathbf{D}^{\prime}=\mathbf{D}$ and $\mathbf{D}^2=\mathbf{D}$ and we have $\Hmanova 
= \eY^\star \mathbf{D} \eY=  \sum_{i,j = 1}^{\nobs} d_{i,j} \ey{i} \otimes \ey{j}$. Then : 
\begin{equation}
\label{eq:kernel_trick_statmanova}
\begin{split}
\statmanova 
=& \trhs(\THLcov) 
\\=& \left \langle \rcovt^{-1}, \Hmanova\right \rangle_{\HSk}
\\=& \sum_{i,j = 1}^{\nobs}\sum_{t = 1}^{\tmax} d_{i,j} \rcoval_t^{-1} \left \langle \rcovec_t \otimes \rcovec_t , \ey{i} \otimes \ey{j} \right \rangle_{\HSk}
\\=& \sum_{i,j = 1}^{\nobs}\sum_{t = 1}^{\tmax} d_{i, j} \rcoval_t^{-1} \left \langle \rcovec_t,\ey{i} \right \rangle_{\Hk}\left \langle \rcovec_t,\ey{j} \right \rangle_{\Hk}
\\=&\trmtrx(\THLcovkt),
\end{split}
\end{equation}
where $\prey  
:=  \Big( \rcoval_t^{-\frac{1}{2}}\left \langle \rcovec_t,\ey{i} \right \rangle_{\Hk}\Big)
_{t \in \{1,\dots,\tmax\},i\in\{ 1, \dots, \nobs \}}\in\mtrx{\tmax,\nobs}$. From Proposition~\ref{proposition: Diagonalisation de sigma epsilon}, we find that
\begin{equation}
\label{eq:KTKY}
\prey  
 = \nobs^{-\frac{1}{2}} \rcovals^{-1} \rcovecskt^{\prime} \Pxp \gram .
\end{equation}

The following proposition shows the relation between $\THLcovkt$ and $\THLcov$. 
\newcommand{\hprey}{\Psi}
\begin{proposition}
\label{proposition: Diagonalisation de THLcov}
The Hilbert-Schmidt operator $\THLcov$ has the same spectrum than the matrix $\prey \mathbf{D} \prey^\prime \in \mtrx{\tmax}$. Moreover, if $\THLcoveckt \in \R ^\tmax$ is an unit eigenvector  of $\THLcovkt$ associated with the eigenvalue $\rcoval$, then an unit eigenfunction $\THLcovec \in\Hk$ of $\THLcov$ associated with $\rcoval$ may be obtained by :
\begin{equation}
\label{eq: formule de f en fonction de u THLcov}
\begin{split}
\THLcovec = \frac{\THLcoveckt' \hprey \Hmanova  }{\sqrt {\THLcoveckt' \prey \mathbf{D} \gram \mathbf{D} \prey' \THLcoveckt
 }} 
\end{split}
\end{equation}
where $\hprey = \nobs^{-\frac{1}{2}} \rcovals^{-1} \rcovecskt^{\prime} \Pxp \eY \in \Hk^T$.
\end{proposition}

\begin{proof}
On the one hand, from \eqref{eq:KTKY} we find that $\prey =\nobs^{-\frac{1}{2}} \rcovals^{-1} \rcovecskt^{\prime} \Pxp  \left( \eY \eY^\star \right) =  \hprey \eY^\star$, where $\hprey \eY^\star$ is seen as the linear application from $\Hk^n $ to $\Hk^T$ defined for $\boldsymbol{h} \in \Hk^n$  by $(\hprey \eY^\star)( \boldsymbol{h})  := \hprey ( \eY^\star \boldsymbol{h}) = (\eY^\star \boldsymbol{h} (\hprey_1) , \dots, \eY^\star \boldsymbol{h} (\hprey_\tmax) )  \in \Hk ^\tmax  $   according the applications defined in and \eqref{eq:gA} and \eqref{eq:adjoint}.

On the other hand, we have
\begin{eqnarray*}
\rcovt^{-1} &=&
 \sum_{t = 1}^{\tmax} (\rcoval_{t}^{-1/2} \rcovec_t)  \otimes (\rcoval_{t}^{-1/2} \rcovec_t) \\
&=& \left(n^{-1/2} \rcovals^{-1} \rcovecskt^{\prime} \residuals  \right)^\star
\left( n^{-1/2} \rcovals^{-1} \rcovecskt^{\prime} \residuals \right) \hskip 1cm \textrm{according to Proposition \ref{proposition: Diagonalisation de sigma epsilon}}\\
&=& \hprey^\star \hprey.
\end{eqnarray*}
Thus, $\THLcov = (\hprey^\star \hprey) (\eY^\star \mathbf{D} \eY)$ and $\THLcovkt = \hprey \eY^\star \mathbf{D} \eY \hprey^\star$ according to the properties of the adjoint operator. In particular, we have that $\hprey^\star \THLcovkt =  \hprey^\star \hprey \eY^\star \mathbf{D} \eY \hprey^\star$ and thus 
\begin{equation*}
\hprey^\star \THLcovkt= \THLcov \hprey^\star .
\end{equation*}
By taking the adjoint on each term and composing by $\Hmanova$, it gives
\begin{equation}
\label{eq:HLdualite}
\begin{split}
\hprey \eY^\star \mathbf{D} \eY \THLcov = \THLcovkt \hprey \eY^\star \mathbf{D} \eY.
\end{split}
\end{equation}
Also note that according to \eqref{eq:gstarg} and Proposition~\ref{proposition: Diagonalisation de sigma epsilon},
\begin{equation*}
\begin{split}
\hprey \hprey^\star &= \rcovals^{-\frac{1}{2}} \Big( \left \langle \rcovec_t , \rcovec_{t^\prime} \right \rangle_{\Hk} \Big)_{t,t^\prime \in \{ 1, \dots, \tmax \}} \rcovals^{-\frac{1}{2}} \\
&= \rcovals^{-1}.
\end{split}
\end{equation*}

Then, if $\THLcovec \in \Hk$ is an unit eigenfunction of $\THLcov$ associated with the eigenvalue $\THLcoval$, we have $\THLcov \THLcovec = \THLcoval \THLcovec$ and with \eqref{eq:HLdualite} it gives that
\begin{equation*}
  \hprey \eY^\star \mathbf{D} \eY \THLcov \THLcovec = \THLcoval \hprey \eY^\star \mathbf{D} \eY \THLcovec
\end{equation*}
Note that this inequality is in $\R^T$ according to the application defined in Equation \eqref{HkcontreH}. Finally we obtain  
\begin{equation*}
\THLcovkt \Big(\hprey \eY^\star \mathbf{D} \eY\THLcovec\Big) = \THLcoval \Big(\hprey \eY^\star \mathbf{D} \eY \THLcovec\Big)  
\end{equation*}
which means that the vector $\hprey \eY^\star \mathbf{D} \eY \THLcovec = \hprey  \Hmanova   \THLcovec \in\R^\tmax  $ is an eigenvector of $\THLcovkt$ associated with the eigenvalue $\THLcoval$. The spectrum of $\THLcov$ is thus included in the spectrum of $\THLcovkt$. On the other hand, let $\THLcoveckt \in \R^\tmax $ be an unit eigenvector of $\THLcovkt$ associated with the eigenvalue $\THLcoval$. We have $\THLcoveckt' \THLcovkt  = \THLcoval \THLcoveckt' $ and   since $\hprey \eY^\star \mathbf{D} \eY \in \R^\tmax$, then  
\begin{equation*}
 \THLcoveckt' \THLcovkt \hprey \eY^\star \mathbf{D} \eY = \THLcoval \THLcoveckt' 
\hprey \eY^\star \mathbf{D} \eY,
\end{equation*}
then using $ \prey =  \hprey \eY^\star$, $\rcovt^{-1} = \hprey^\star \hprey$ and $\Hmanova 
= \eY^\star \mathbf{D} \eY$ it gives
\begin{equation*}
 \THLcoveckt' (\hprey \eY^\star \mathbf{D} \eY \THLcov) = \THLcoval \Big( \THLcoveckt' 
\hprey \eY^\star \mathbf{D} \eY \Big).
\end{equation*}
Finally we obtain
\begin{equation*}
\THLcov \Big(\THLcoveckt' \hprey \eY^\star \mathbf{D} \eY\Big) = \THLcoval \Big( \THLcoveckt' 
\hprey \eY^\star \mathbf{D} \eY \Big),
\end{equation*}
see the proof of Proposition \ref{proposition: Diagonalisation de sigma epsilon} for the details of the last equivalence. Thus the vector $\THLcoveckt' 
\hprey \eY^\star \mathbf{D} \eY   =   \THLcoveckt' 
\hprey \Hmanova \in\Hk$ is an eigenvector of $\THLcov$ associated with the eigenvalue $\THLcoval$, and the spectrum of $\THLcovkt$ is included in the spectrum of $\THLcov$. We conclude that $\THLcov$ and $\THLcovkt$ share the same spectrum. Let $\THLcoveckt \in \R^\tmax$ be an unit eigenvector of $\THLcovkt$ associated with the eigenvalue $\THLcoval$. We have that:  
\begin{equation*}
\begin{split}
\left \| \THLcoveckt'  \hprey \Hmanova \right \| _{\Hk}^2 
=&    \langle \THLcoveckt' \hprey  \eY^\star \mathbf{D} \eY ,  \THLcoveckt'  \hprey \eY^\star \mathbf{D} \eY \rangle_\Hk  \\
=&    \langle \THLcoveckt' \prey \mathbf{D} \eY ,  \THLcoveckt' \prey  \mathbf{D} \eY \rangle_\Hk  \\
=&  \langle z '   \eY ,  z'  \eY \rangle_\Hk \\ 
=&   z ' \gram z   , 
\end{split}
\end{equation*}
where $z  = (\THLcoveckt' \prey \mathbf{D})' \in \R ^\nobs$. Thus, $(\THLcoveckt' \prey \mathbf{D} \gram \mathbf{D} \prey' \THLcoveckt
)^{-\frac{1}{2}} \THLcoveckt' \hprey \Hmanova  $ is an unit eigenfunction of $\THLcov$.


\end{proof}

Let $Y_0$ a new point in $\mathcal Y$ and $\ey{0}$ its featurization in $\Hk$. We can finally apply a kernel trick to compute the projection of $\ey{0}$ on the eigenfunction $\THLcovec$ of $\THLcov$ given by \eqref{eq: formule de f en fonction de u THLcov}. The coordinate of the projection of $\ey{0}$ on the eigenspace of $\THLcovec$  is given by
\begin{equation*}
\begin{split}
 \langle  \THLcovec ,  \ey{0}  \rangle_\Hk   
 & = \frac 1 {\sqrt {\THLcoveckt' \prey \mathbf{D} \gram \mathbf{D} \prey' \THLcoveckt
 }}  \langle   \THLcoveckt' \hprey \Hmanova    ,  \ey{0}  \rangle_\Hk \\ 
 & = \frac 1 {\sqrt {\THLcoveckt' \prey \mathbf{D} \gram \mathbf{D} \prey' \THLcoveckt
 }}  \langle   \THLcoveckt' \hprey  \eY^\star \mathbf{D} \eY  ,  \ey{0}  \rangle_\Hk \\
 & = \frac 1 {\sqrt {\THLcoveckt' \prey \mathbf{D} \gram \mathbf{D} \prey' \THLcoveckt  }} \THLcoveckt' \prey \mathbf{D}  \: \boldsymbol{k}(\dY, Y_0)
\end{split}
\end{equation*}
where $\boldsymbol{k}(\dY, Y_0) = (   k( Y_1, Y_0), \dots, k( Y_n, Y_0))'$.

\subsection{Kernel trick for the Nyström statistic}
\label{sec:The Nystrom method}

\newcommand{\lPx}{\mathbf{P}_{\lX}}
\newcommand{\lPxp}{\mathbf{P}_{\lX}^{\perp}}

\subsubsection{Computation of the landmarks and anchors}
To compute the Nystr\"om approximation of the truncated kernel Hotelling-Lawley trace statistic \eqref{eq:nystrom_statistic}, we need a kernel trick to determine the eigenfunctions and eigenvalues of $\arcov$. 

Let $\lY=(\ly_1,\dots,\ly_ \nlandmarks)$ be the $\nlandmarks<\nobs$ landmarks sampled from $\dY$. 
We denote $\flandmarks:\{ 1, \dots, \nlandmarks \} \longrightarrow \{ 1, \dots, \nobs \}$ such that $\lY=(\dy_{\flandmarks(1)},\dots,\dy_{\flandmarks(\nlandmarks)})$ and the associated explanatory variables are $\lX = (\lx{1},\dots,\lx{\nlandmarks})^\prime \in\mtrx{\nlandmarks,\nx}$. Let $\lproj\in\mtrx{\nlandmarks,\nobs}$ such that the $i^{th}$ column of $\lproj^\prime$ is the vector of $\R^\nobs$ full of zero except in $\flandmarks(i)$ where it is equal to $1$, then we have $\lY = \lproj \dY$ and $\lX = \lproj \dX$.

Let $\elY=(\ely{1},\dots,\ely{\nlandmarks})$ be the embeddings of the landmarks in $\Hk$. We define the Gram matrices $\gramly = (\left \langle \ely{i},\ely{j}\right \rangle_{\Hk}))_{i,j\in\{ 1, \dots, \nlandmarks \}} \in \mtrx{\nlandmarks}$, $\gramyz = (\left \langle \ey{i},\ely{j}\right \rangle_{\Hk}))_{i\in\{ 1, \dots, \nobs \},j\in\{ 1, \dots, \nlandmarks \}} \in \mtrx{\nobs,\nlandmarks}$ and $\gramzy = \gramyz^\prime$. We have $\elY=\lproj\eY$, $\gramly = \elY\elY^\star$, $\gramyz = \eY\elY^\star$ and $\gramzy = \elY \eY^\star$. 

We assume that $\lX$ is full rank and denote $\lPx = \lX (\lX^\prime \lX)^{-1} \lX^\prime$ and $\lPxp = (\identity_\nlandmarks - \lPx)$, then we have the vector of landmark residuals $\lresiduals=  (\lresidual_1,\dots,\lresidual_\nlandmarks)'$ in $\Hk^\nlandmarks$ such that :
\begin{equation}
\label{eq:landmarks_residuals_expression}
\begin{split}
\lresiduals = \lPxp \elY.   
\end{split}
\end{equation}
The landmark residual covariance operator is such that $\lrcov = \nlandmarks^{-1}\lresiduals^\star\lresiduals$. The $\nanchors$ anchors $ \anc_1,\dots,\anc_\nanchors $ are defined as the orthonormal set of eigenfunctions of $\lrcov$ associated with its highest and non-increasing set of eigenvalues encoded in the diagonal matrix $\lrcovals = \diag(\lrcoval_1,\dots,\lrcoval_\nanchors)\in\mtrx{\nanchors}$. We also define  $\lrcovkt = \nlandmarks^{-1} \lresiduals\lresiduals^\star = \nlandmarks^{-1} \lPxp \gramly \lPxp$. The following proposition highlights the link between $\lrcov$ and $\lrcovkt$: 
\begin{proposition}
\label{proposition: Diagonalisation de landmark residual covariance}
The Hilbert-Schmidt operator $\lrcov$ has the same spectrum than the matrix $\lrcovkt$. Moreover, if $\lrcoveckt$ is an unit eigenvector of $\lrcovkt$ associated with the eigenvalue $\lrcoval$, then an unit eigenfunction $\lrcovec$ of $\lrcov$ associated with $\lrcoval$ may be obtained by :
\begin{equation}
\label{eq: formule de f en fonction de u landmarks}
\begin{split}
\lrcovec = \frac{1}{\sqrt{\nlandmarks \lrcoval}}  \lrcoveckt{}' \lresiduals 
\end{split}
\end{equation}
\end{proposition}

\begin{proof}
The proof is exactly the same than the proof of Proposition \ref{proposition: Diagonalisation de sigma epsilon}. 
\end{proof}

We denote $\lrcovecskt = (\lrcoveckt_1,\dots,\lrcoveckt_\nanchors)\in\mtrx{\nlandmarks,\nanchors}$ the matrix where the $\nanchors$ columns form an orthonormal set of eigenvectors of $\lrcovkt$ associated with the eigenvalues $\lrcoval_1,\dots,\lrcoval_\nanchors$, then for $i\in\{ 1, \dots, \nanchors \}$, we have :
\begin{equation}
\label{eq:anchors_expression}
\begin{split}
\ancMat =(\anc_1,\dots,\anc_\nanchors)^\prime =  \frac{1}{\sqrt{\nlandmarks}} {\lrcovals{}}^{-\frac{1}{2}} {\lrcovecskt}^{\prime} \lresiduals \in \Hk ^\nanchors.
\end{split}
\end{equation}

\subsubsection{Diagonalisation of the Nyström residual covariance operator}

Let $\aresiduals = (\aresidual_1,\dots,\aresidual_\nobs)$ be the vector of $\Hk^\nobs$ containing the projections of the residuals onto $\Span(\anc_1,\dots,\anc_\nanchors)$. We have $\aresidual_i = \sum_{s = 1}^{\nanchors}\left \langle \anc_s , \residual_i \right \rangle_{\Hk}\anc_s$ and
$\arcov=\frac{1}{\nobs} \sum_{i=1}^\nobs \aresidual_i \otimes \aresidual_i$.

We now derive a kernel trick to determine the orthonormal eigenfunctions $\arcovec{1},\dots,\arcovec{\nanchors}$ of $\arcov$ associated with the eigenvalues $\arcoval{1},\dots,\arcoval{\nanchors}$.

We define $\arcovkt \in \mtrx{\nanchors}$ with general term $\langle \anc_s , \rcov \anc_{s'} \rangle_{\Hk}$. The following proposition highlights the link between $\arcov$ and $\arcovkt$.


\begin{proposition}
\label{proposition: Diagonalisation de nystrom residual covariance}
The Hilbert-Schmidt operator $\arcov$ has the same spectrum than the matrix $\arcovkt$. Moreover, if $\arcoveckt$ is an unit eigenvector of $\arcovkt$ associated with the eigenvalue $\arcoval{}$, then an unit eigenfunction $\arcovec{}$ of $\arcov$ associated with $\arcoval{}$ may be obtained by :
\begin{equation}
\label{eq: formule de f en fonction de u nystrom residual}
\begin{split}
\arcovec{} = (\arcoveckt{})' \ancMat .
\end{split}
\end{equation}
\end{proposition}

\begin{proof}
Observe that:
\begin{equation*}
\begin{split}
\arcov &=\frac{1}{\nobs}\sum_{i=1}^\nobs \Big(\sum_{s = 1}^{\nanchors}\left \langle \anc_s , \residual_i \right \rangle_{\Hk}\anc_s\Big)\otimes\Big(\sum_{s' = 1}^{\nanchors}\left \langle \anc_{s'} , \residual_i \right \rangle_{\Hk}\anc_{s'}\Big)
\\&=\frac{1}{\nobs}\sum_{s,s' = 1}^{\nanchors} \left \langle \anc_s , \Big( \sum_{i=1}^\nobs \residual_i \otimes \residual_i \Big) \anc_{s'}\right \rangle_{\Hk}\anc_s \otimes\anc_{s'}
\\&= \ancMat^\star \arcovkt \ancMat.
\end{split}
\end{equation*}
As the anchors $\anc_1,\dots,\anc_\nanchors$ form an orthonormal basis of $\Hka$, we have $\ancMat \ancMat ^\star = \big(\hip{\anc_i,\anc_j}\big)_{i,j\in\{1,\dots,\nanchors\}} = \identity_\nanchors \in \mtrx{\nanchors}$. Thanks to this relation, we directly have that $\ancMat \arcov = \arcovkt \ancMat$. Let $\arcoveckt{}$ be the eigenvector of $\arcovkt$ associated with the non-zero eigenvalue $\arcoval{}$. We have $ \arcoveckt{} '(\arcovkt \ancMat) =  \arcoveckt{} ' (\ancMat \arcov) = \arcov (\arcoveckt{}'\ancMat)$, where the last equality is obtained similarly to the proof of Proposition \ref{proposition: Diagonalisation de sigma epsilon}. We also have $(\arcoveckt{}'\arcovkt) \ancMat = \arcoval{} \arcoveckt{}'\ancMat$, thus $\arcoveckt{}'\ancMat$ is an eigenfunction of $\arcov$ associated with the eigenvalue $\arcoval{}$, and the non-zero eigenvalues of $\arcovkt$ are included in the set of non-zero eigenvalues of $\arcov$.

Now let $\arcovec{}$ be an eigenfunction of $\arcov$ associated with the non-zero eigenvalue $\arcoval{}$. As $\arcov$ is self-adjoint, for $i\in\{1,\dots,\nobs\}$ we have $\Big((\Anc \arcov)\arcovec{}\Big)_i = \langle \arcov \anc_i , \arcovec{} \rangle_\Hk = \langle \anc_i ,  \arcov \arcovec{} \rangle_\Hk$. Thus, we have $\arcovkt \Anc \arcovec{} = \Anc \arcov \arcovec{} = \arcoval{} \Anc \arcovec{}$. Then $\Anc  \arcovec{}\in\R^\nanchors$ is an eigenvector of $\arcovkt$ associated with the eigenvalue $\arcoval{}$, thus the non-zero eigenvalues of $\arcov$ are included in the set of non-zero eigenvalues of $\arcovkt$. Thus, $\arcov$ and $\arcovkt$ share the same spectrum. As we have $\hn{\arcoveckt{}' \Anc} = 1$, then $\arcovec{} = \arcoveckt{}'\Anc$. 
\end{proof}

In practice, the full expression of $\arcovkt$ is such that : 
\begin{equation*}
\label{eq:kernel_trick_nystrom_residual_covariance}
\begin{split}
\arcovkt = \frac{1}{\nobs\nlandmarks} \lrcovals^{-\frac{1}{2}} {\lrcovecskt}^\prime \lPxp \gramzy \Pxp \gramyz \lPxp \lrcovecskt \lrcovals^{-\frac{1}{2}}
\end{split}
\end{equation*}

\subsubsection{Computation of the Nystr\"om statistic}
\paragraph{}
Let $\arcovecskt = (\arcoveckt_1,\dots,\arcoveckt_\tmax) \in \mtrx{\nanchors,\tmax}$ and $\arcovals = \operatorname{diag}(\arcoval{1},\dots,\arcoval{\tmax}) \in \mtrx{\tmax}$ be the matrix containing the $\tmax$ first eigenvectors and eigenvalues of $\arcovkt$ respectively. The Nystr\"om statistic $\statmanovany$ can be written with respect to $\arcovecskt$, $\arcovals$ and $\mathbf{D}$. With the same development than in Equation \eqref{eq:kernel_trick_statmanova}, we have : 
\begin{equation*}
\begin{split}
\statmanovany
=\trmtrx(\aprey^\prime \mathbf{D} \aprey),
\end{split}
\end{equation*}
where $\aprey = (\nlandmarks {\arcovals})^{-\frac{1}{2}} {\arcovecskt}^\prime {\lrcovals}^{-\frac{1}{2}} {\lrcovecskt}^\prime \lPxp\gramzy = \Big( \arcoval{t}^{-\frac{1}{2}}\left \langle \arcovec{t},\ey{i} \right \rangle_{\Hk}\Big)_{t \in \{1,\dots,\tmax\}, i \in \{1,\dots,\nobs\}}\in\mtrx{\tmax,\nobs}$.


\subsection{Kernel trick for the diagnostics and kernel Cook's distance}
\label{app:diagnostics}

\subsubsection{Computation of the Diagnostic Plots}
For $t\in\{ 1, \dots, \tmax \}$, the $t^{th}$ kernel response plot contains the values of the $t^{th}$ column of the matrix representing the projections of the response embeddings onto the eigenfuction $\rcovec_t$:
\begin{equation*}
 \begin{split}
 \Big( \left \langle \ey{i} , \rcovec_t \right \rangle_{\Hk}\Big)_{i \in \{ 1, \dots, \nobs \},t \in \{ 1, \dots, \tmax \}} &=
\Big( \left \langle \ey{i} , \tfrac{1}{\sqrt{\nobs \rcoval_t}}  \rcoveckt'_t\residuals \right \rangle_{\Hk}\Big)_{i \in \{ 1, \dots, \nobs \},t \in \{ 1, \dots, \tmax \}} \\&=
\Big( \left \langle \ey{i} , \tfrac{1}{\sqrt{\nobs \rcoval_t}}  \rcoveckt'_t\Pxp \eY \right \rangle_{\Hk}\Big)_{i \in \{ 1, \dots, \nobs \},t \in \{ 1, \dots, \tmax \}} \\&=
\tfrac{1}{\sqrt{\nobs}}\gram \Pxp \rcovecskt \rcovals^{-\frac{1}{2}} \in \mtrx{\nobs,\tmax}.
 \end{split}
\end{equation*}
The $t^{th}$ kernel residual plot contains the values of the $t^{th}$ column of the matrix representing the projections of the residuals onto the eigenfuction $\rcovec_t$:
 \begin{equation*}
 \begin{split}
 \Big( \left \langle \residual_i , \rcovec_t \right \rangle_{\Hk}\Big)_{i \in \{ 1, \dots, \nobs \},t \in \{ 1, \dots, \tmax \}} &=
\Big( \left \langle \Pxpi{i} \eY , \tfrac{1}{\sqrt{\nobs \rcoval_t}}  \rcoveckt'_t\Pxp \eY \right \rangle_{\Hk}\Big)_{i \in \{ 1, \dots, \nobs \},t \in \{ 1, \dots, \tmax \}} \\&=
\tfrac{1}{\sqrt{\nobs}}\Pxp \gram \Pxp \rcovecskt \rcovals^{-\frac{1}{2}} \in \mtrx{\nobs,\tmax},
 \end{split}
\end{equation*}
where $\Pxpi{i}$ is the $i$-th  row of $\Pxp$. Both quantities are plotted against the values of the $t^{th}$ column of the matrix representing the projections of the predicted embeddings onto the eigenfuction $\rcovec_t$:
\begin{equation*}
 \begin{split}
 \Big( \left \langle \hey{i} , \rcovec_t \right \rangle_{\Hk}\Big)_{i \in \{ 1, \dots, \nobs \},t \in \{ 1, \dots, \tmax \}} &=
\Big( \left \langle \Pxi{i} \eY , \tfrac{1}{\sqrt{\nobs \rcoval_t}}  \rcoveckt'_t\Pxp \eY \right \rangle_{\Hk}\Big)_{i \in \{ 1, \dots, \nobs \},t \in \{ 1, \dots, \tmax \}} \\&=
\tfrac{1}{\sqrt{\nobs}}\Px \gram \Pxp \rcovecskt \rcovals^{-\frac{1}{2}} \in \mtrx{\nobs,\tmax},
 \end{split}
\end{equation*}
where $\Pxi{i}$ is the $i$-th  row of $\Px$.

\subsubsection{Computation of the kernel Cook's distance}

Let $\mathbf{g} = (g_1,\dots,g_\nLmanova)$ a random vector of $\Hk^\nLmanova$.
We define the covariance of $\mathbf{g}$ as a $\nLmanova \times \nLmanova$ matrix where the coordinates are elements of $\HSk$, such that for $i,j \in \{ 1, \dots, \nLmanova \}$, we have $\cov(\mathbf{g})_{i,j} = \cov(g_i,g_j) = \E\left( (g_i - \E(g_i)) \otimes (g_j - \E(g_j))\right)$. 

In order to extend the Cook's distance to our framework in RKHS (see fo instance \cite{diaz-garcia_sensitivity_2005}), we would like to introduce a Cook's distance of the form
\begin{equation}
\label{eq:CookA}
\mathcal{D}_{\mathbf{L}}(i) = \nLmanova^{-1} \trhs\left( (\Lmanova(\hparams - \hparams_{(i)}))^\star 
\mathbf{A}
(\Lmanova (\hparams- \hparams_{(i)}))\right),
\end{equation}
where $\mathbf{A} \in \mathcal{M}_{d,d} (\HSk)$ is a regularized inverse of $\cov(\Lmanova\hparams)$. 

Now we compute the expectation $\E (\Lmanova\hparams)$ and covariance $\cov(\Lmanova\hparams)$ of $\Lmanova\hparams$. Note that the combination of Equations \eqref{eq:linear model on embeddings} and \eqref{eq: empirical model parameters} gives 
$\Lmanova\hparams = \Lmanova\params + \Lmanova\XXX \errors$. Let $j \in \{ 1, \dots, \nLmanova \}$, we have $(\Lmanova\hparams)_j = (\Lmanova\params)_j + \lxxxj'\errors$, where $\lxxxj \in \R^\nobs$ is the $j^{th}$ column of $W=\LXXXprime$. Then we have $\E \Big((\Lmanova\hparams)_j\Big) = (\Lmanova\params)_j$. Let $i,j \in \{ 1, \dots, \nLmanova \}$, we have:
\begin{equation*}
 \begin{split}
 \cov\Big( (\Lmanova\hparams)_i,(\Lmanova\hparams)_j \Big) 
 =& \E\left ( ((\Lmanova\hparams)_i - (\Lmanova\params)_i)\otimes ((\Lmanova\hparams)_j - (\Lmanova\params)_j) \right)
 \\ =& \E \left ( (\lxxxi_{\ast, i}'\errors) \otimes (\lxxxj'\errors) \right ) 
 \\ =& \sum_{k,l = 1}^{\nobs}\lxxxi_{k,i} \lxxxi_{l,j}  \E \left ( \error_k\otimes \error_l \right )
 \\ =& \sum_{k = 1}^{\nobs}\lxxxi_{k,i} \lxxxi_{k,j}  \ecov_k \\ =& \sum_{k = 1}^{\nobs}\lxxxi_{k,i} \lxxxi_{k,j}  \ecovt + \sum_{k = 1}^{\nobs}\lxxxi_{k,i} \lxxxi_{k,j}  \ecovtci{k}
 \\ =& (\Lmanova\XXinv\Lmanova')_{i,j} \ecovt + \sum_{k = 1}^{\nobs}\lxxxi_{k,i} \lxxxi_{k,j}  \ecovtci{k}.
 \end{split}
 \end{equation*} 
 Since this covariance is unknown, we can only propose an approximation of this last through the spectral truncation we use in the paper.  We consider the truncated covariance $\cov_\tmax(\Lmanova\hparams)$ of general term $(\Lmanova\XXinv\Lmanova')_{i,j} \ecovt$, which is approximated through $\widehat{\cov}_\tmax(\Lmanova\hparams)$ of general term $(\Lmanova\XXinv\Lmanova') _{i,j} \rcovt$. 

We are now in position to introduce the truncated Cook distance we propose in this paper: 
\begin{equation*}
\mathcal{D}_{\mathbf{L},\tmax}(i) 
: = 
 \nLmanova^{-1} \trhs\left( (\Lmanova(\hparams - \hparams_{(i)}))^\star 
(\Lmanova\XXinv\Lmanova')^{-1}  \otimes  \rcovt^{-1}
(\Lmanova (\hparams- \hparams_{(i)}))\right),
\end{equation*}

For $i \in \{ 1, \dots, \nobs \}$, let the least squares estimator $\hparams_{(i)}$   computed without  using observation $i$. By generalizing to RKHS standard calculations, see for instance \cite{diaz-garcia_sensitivity_2005}, we obtain
\begin{equation*}
\begin{split}
\Lmanova(\hparams - \hparams_{(i)}) = \left (\begin{matrix} \frac{ \lxxxi_{i,1}}{1 - \Pxij{i}{i}} \residual_i \\ \vdots \\  \frac{\lxxxi_{i,\nLmanova}}{1 - \Pxij{i}{i}} \residual_i  \end{matrix} \right) \in \Hk^{\nLmanova}.
\end{split}
\end{equation*}

According to \eqref{eq:formequadra}, we then have that: 
\begin{equation*}
\begin{split}
\mathcal{D}_{\mathbf{L}, \tmax}(i) 
=& \frac{1}{\nLmanova}\trhs\left( \sum_{j,k = 1}^{\nLmanova} \frac{ \lxxxi_{i,j}}{1 - \Pxij{i}{i}} \frac{ \lxxxi_{i,k}}{1 - \Pxij{i}{i}} (\Lmanova\XXinv\Lmanova')_{j,k}^{-1} \residual_i \otimes \rcovt^{-1} \residual_i \right)
\\=& \frac{1}{\nLmanova} \sum_{j,k = 1}^{\nLmanova} \frac{ \lxxxi_{i,j}}{1 - \Pxij{i}{i}} \frac{ \lxxxi_{i,k}}{1 - \Pxij{i}{i}} (\Lmanova\XXinv\Lmanova')_{j,k}^{-1} \trhs\left(\residual_i \otimes \rcovt^{-1}. \residual_i \right)
\end{split}
\end{equation*}
We remark that:
\begin{equation*}
\begin{split}
\sum_{j,k = 1}^{\nLmanova} \frac{ \lxxxi_{i,j}}{1 - \Pxij{i}{i}} \frac{ \lxxxi_{i,k}}{1 - \Pxij{i}{i}} (\Lmanova\XXinv\Lmanova')_{j,k}^{-1} 
= \frac{\lxxxi_{i,\ast} (\Lmanova\XXinv\Lmanova')^{-1} \lxxxi_{i,\ast}^\prime}{(1 - \Pxij{i}{i})^2}
\end{split}
\end{equation*}
with $\lxxxi_{i,\ast}$ being the $i^{th}$ row of $W$, and
\begin{equation*}
\begin{split}
\trhs\left(\residual_i \otimes \rcovt^{-1} \residual_i \right) 
=& \sum_{t = 1}^{\tmax}\rcoval_t^{-1} \left \langle \rcovec_t , \residual_i \right \rangle_{\Hk}^2.
\end{split}
\end{equation*}
We recognize the inner products contained in the matrix $\frac{1}{\sqrt{\nobs}}\Pxp \gram \Pxp \rcovecskt \rcovals^{-\frac{1}{2}}$ , thus we have:
\begin{equation*}
\begin{split}
\trhs\left(\residual_i \otimes \rcovt^{-1} \residual_i \right) 
=&\frac{1}{\nobs}(\Pxp \gram \Pxp \rcovecskt \rcovals^{-2}\rcovecskt ^\prime\Pxp \gram \Pxp )_{i,i}
\end{split}
\end{equation*}
Finally, we have an expression for the Cook distance $\mathcal{D}_{\mathbf{L}}$ associated with the $i^{th}$ observation:
\begin{equation*}
\begin{split}
\mathcal{D}_{\mathbf{L},\tmax}(i)  = \frac{\lxxxi_{i,\ast} (\Lmanova\XXinv\Lmanova')^{-1} \lxxxi_{i,\ast}^\prime}{\nLmanova\nobs(1 - \Pxij{i}{i})^2}(\Pxp \gram \Pxp \rcovecskt \rcovals^{-2}\rcovecskt ^\prime\Pxp \gram \Pxp )_{i,i}.
\end{split}
\end{equation*}

\subsection{Computational costs}
\paragraph{Computational cost of the TKHL statistic}
We compute the projection matrix $\Pxp$ for $O(\nx^3)$ operations. We compute and diagonalize the matrix $\rcovkt$ for $O(\nobs^3)$ operations. Then we compute the matrix $\prey$ and the TKHL statistic for $O(\tmax^2\nobs+ \nobs^2\tmax)$ operations. As $\nobs\gg \nx$ and $\nobs\gg\tmax$, the global computational cost of the algorithm is $O(\nobs^3)$. 

\paragraph{Computational cost of the TKHL statistic}
We compute the projection matrix $\lPxp$ for $O(\nx^3)$ operations. We compute and diagonalize the matrix 
$\lrcovkt$ for $O(\nlandmarks^3)$ operations. Then we compute and diagonalize the matrix $\arcovkt$ for 
$O(\nanchors^3 + \nanchors^2 \nlandmarks + \nanchors \nlandmarks^2 + \nanchors\nlandmarks \nobs + \nanchors\nobs^2)$. Then we compute the matrix $\aprey$ and the Nystr\"om TKHL statistic for 
$O(\tmax^2\nanchors + \tmax \nanchors^2 + \tmax \nanchors \nlandmarks + \tmax \nlandmarks^2 + \tmax \nlandmarks \nobs + \tmax \nobs^2)$ operations.  Typically, we have $\nobs\gg\nlandmarks\geq\nanchors \geq\tmax$ and $\nx$ is small, then the global computational cost of the algorithm is $O(\nobs^2(\nanchors+\tmax) + \nlandmarks^3 + \nanchors^3)$, which is lower than $O(\nobs^3)$. 

\section{Supplementary Figures}

\begin{figure}[h]
  \centering
   \begin{subfigure}{0.43\textwidth}
        \includegraphics[trim={0 0 0 0},clip, width=0.8\textwidth]{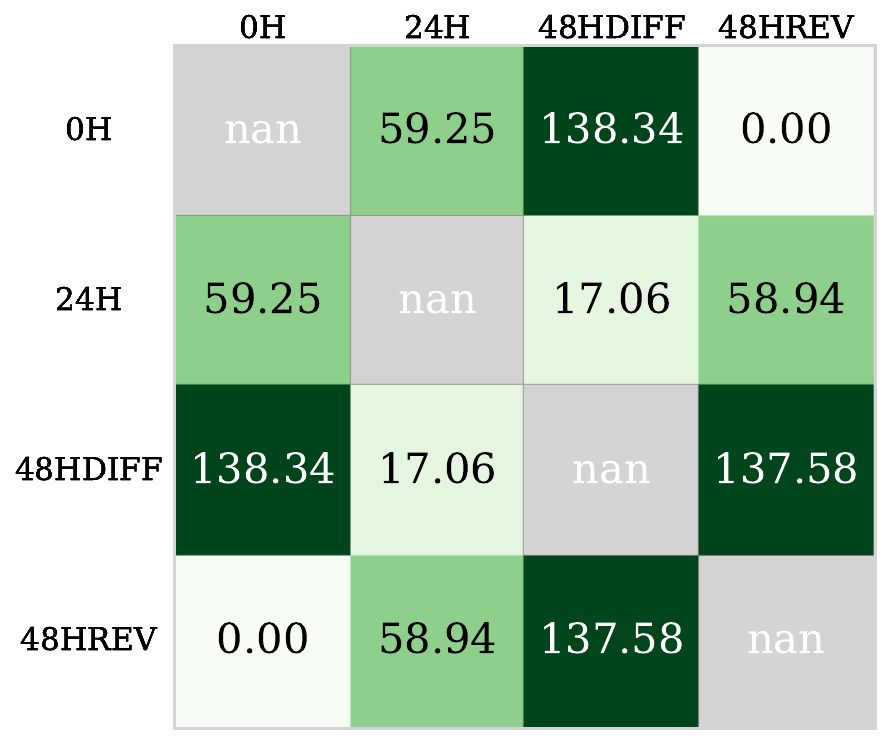}
      \label{fig:heat1}
      \caption{}
    \end{subfigure}
    \begin{subfigure}{0.43\textwidth}
      \includegraphics[trim={0 0 0 0},clip, width=0.8\textwidth]{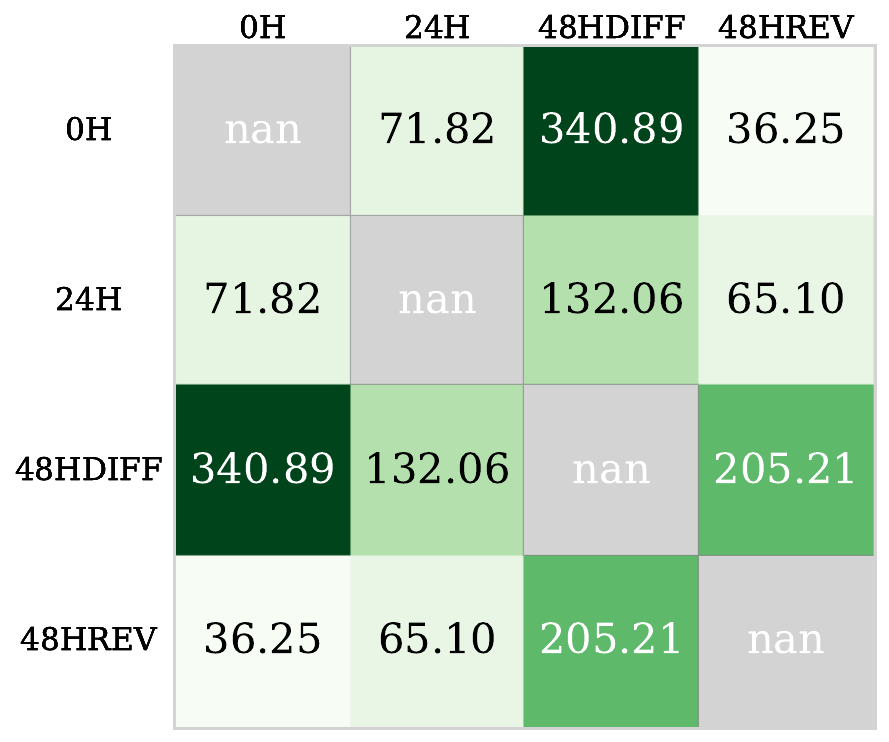}
      \label{fig:heat2}
      \caption{}
    \end{subfigure}
    \begin{subfigure}{0.43\textwidth}
      \includegraphics[trim={0 0 0 0},clip, width=0.8\textwidth]{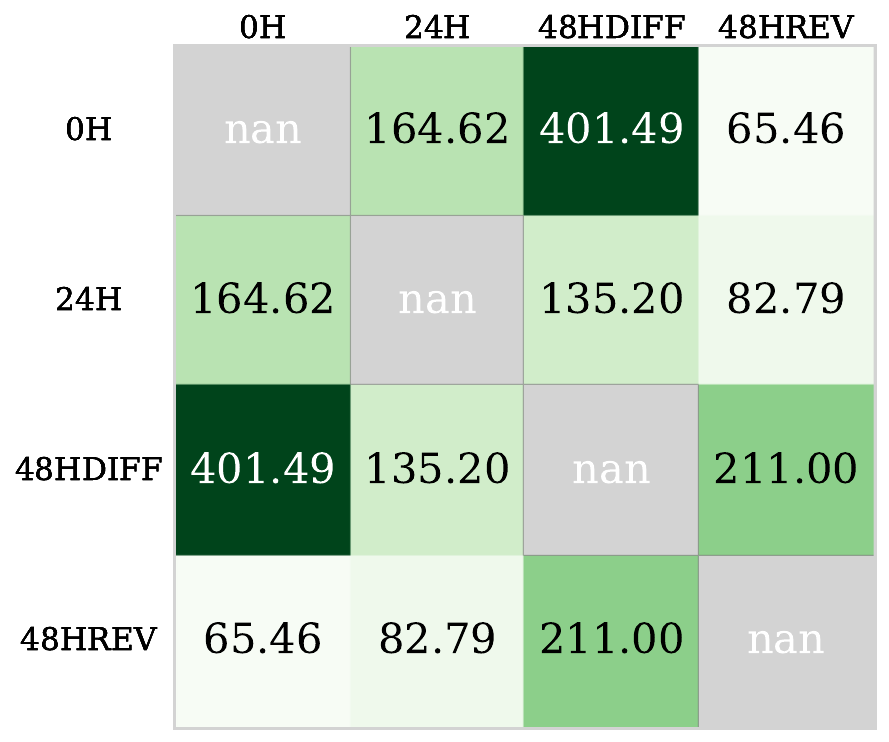}
      \label{fig:heat3}
      \caption{}
    \end{subfigure}
    \begin{subfigure}{0.43\textwidth}
      \includegraphics[trim={0 0 0 0},clip, width=0.8\textwidth]{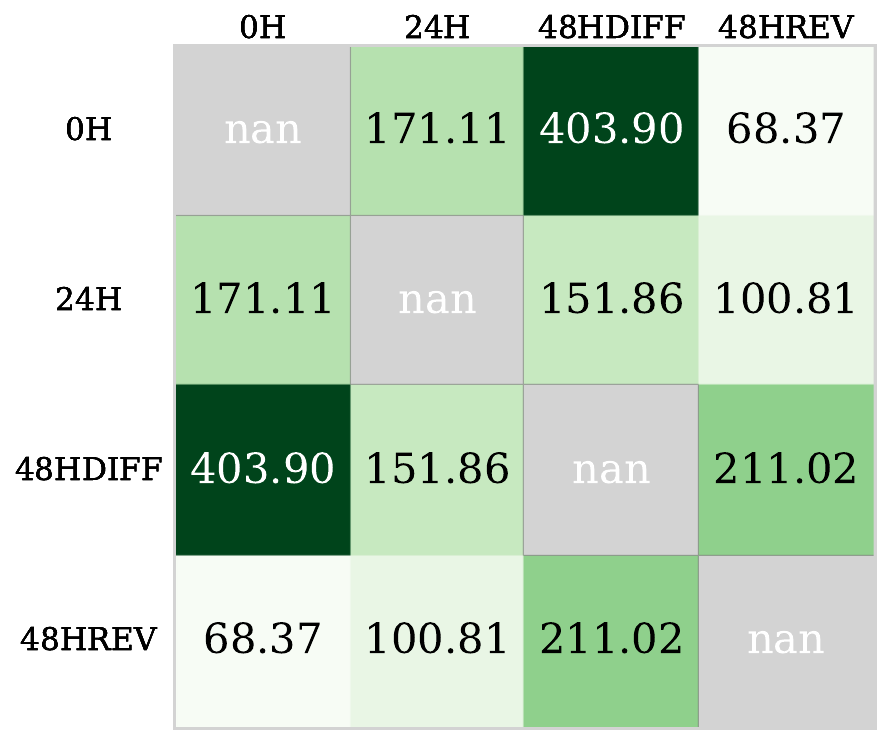}
      \label{fig:heat4}
      \caption{}
    \end{subfigure}
  \caption{Heatmaps of the values of the TKHL statistic associated with (a) $t=1$, (b) $t=2$, (c) $t=3$ and (d) $t=4$, obtained with pairwise tests for different comparisons between media.}
\label{fig:heatmaps}
\end{figure}

\end{document}